\documentclass[aps,prx,twocolumn, notitlepage,longbibliography,superscriptaddress,nofootinbib]{revtex4-2}
\usepackage{amsmath}
\usepackage{amsthm}
\usepackage{thm-restate}
\usepackage{amsfonts}
\usepackage{amssymb}
\usepackage{dsfont}
\usepackage{graphicx}
\usepackage{physics}
\usepackage{tikz}
\usepackage[hidelinks]{hyperref}
\usepackage{mathtools}
\usepackage{comment}
\usepackage{multirow}
\usepackage{ulem}
\usepackage{pst-node}
\usepackage{tikz-cd}
\usepackage[shortlabels]{enumitem}
\usepackage{mathtools}
\usepackage[ruled,vlined,linesnumbered]{algorithm2e}

\usepackage[english]{babel}

\newtheorem{theorem}{Theorem}
\newtheorem{corollary}{Corollary}[theorem]
\newtheorem{definition}{Definition}
\newtheorem{lemma}{Lemma}

\newtheorem{fact}{Fact}
\newtheorem{claim}{Claim}

\usepackage{pst-all}
\usepackage{leftidx}

\def\Hs{\mathsf{H}}
\def\M{\mathcal{M}}
\def\Z{\mathbb{Z}}
\def\C{\mathcal{C}}
\def\D{\mathcal{D}}
\def\B{\mathcal{B}}
\def\R{\mathcal{R}}

\def\L{\mathcal{L}}

\DeclareMathOperator*{\Motimes}{\text{\raisebox{0.25ex}{\scalebox{0.8}{$\bigotimes$}}}}

\newcommand{\ZZ}{\mathbb{Z}}

\newcommand{\be}{\begin{equation}}
\newcommand{\ee}{\end{equation}}
\newcommand{\bc}{\begin{center}}
\newcommand{\ec}{\end{center}}
\newcommand{\nin}{\noindent}

\newcommand{\non}{\nonumber}
\newcommand{\lo}{\overline}

\newcommand{\as}{\mathsf{a}}
\newcommand{\bs}{\mathsf{b}}
\newcommand{\cs}{\mathsf{c}}
\newcommand{\ds}{\mathsf{d}}

\newcommand{\fs}{\mathsf{f}}
\newcommand{\gs}{\mathsf{g}}

\newcommand{\cA}{\mathcal{A}}\newcommand{\cB}{\mathcal{B}}
\newcommand{\cC}{\mathcal{C}}

\definecolor{dualblue}{RGB}{3,101,192}

\usetikzlibrary{arrows.meta}
\makeatletter
\newlength\lrvec@height
\newlength\lrvec@width
\newif\iflrvec@same@height
\def\lrvec{\@ifstar\slrvec@\lrvec@}
\newcommand{\slrvec@}[2][.4ex]{
  \lrvec@same@heighttrue
  \mathpalette\lrvec@@{{#1}{#2}}
}
\newcommand{\lrvec@}[2][.4ex]{
  \lrvec@same@heightfalse
  \mathpalette\lrvec@@{{#1}{#2}}
}
\def\lrvec@@#1#2{\lrvec@@@#1#2}
\def\lrvec@@@#1#2#3{%
  \iflrvec@same@height
    \settoheight{\lrvec@height}{$\m@th#1 \mathbf{T}#3$}
  \else
    \settoheight{\lrvec@height}{$\m@th#1#3$}
  \fi
  \settowidth{\lrvec@width}{$\m@th#1#3$}
  \kern.08em
  \raisebox{#2}{\raisebox{\lrvec@height}{\rlap{%
    \kern-.05em
    \begin{tikzpicture}[<-> /.tip={To[width=.4em, length=.2em]}]
      \draw [<->] (-.05em,0)--(\lrvec@width+.05em,0);
    \end{tikzpicture}%
  }}}%
  #3
  \kern.08em
}
\makeatother

\usepackage[bottom]{footmisc}

\begin{document}

\title{A topological theory for qLDPC: non-Clifford gates and magic state fountain on homological product codes with constant rate and beyond the $N^{1/3}$ distance barrier}

\author{Guanyu Zhu}
\affiliation{IBM Quantum, T.J. Watson Research Center, Yorktown Heights, NY 10598 USA}

\begin{abstract}

We develop a topological theory for fault-tolerant quantum computation in quantum low-density parity-check (qLDPC) codes. We show that there exist hidden simplicial or CW complex structures encoding the topological data for all qLDPC and CSS codes obtained from product construction by generalizing the Freedman-Hastings code-to-manifold mapping. This is achieved by building manifolds  from the Tanner graphs of the skeleton classical or quantum codes, which further form a product manifold and an associated thickened product code defined on its triangulation. One can further deformation retract the manifold back to a CW complex which supports a non-topological code with minimal overhead suitable for near-term implementation. Both types of codes admit  cohomology operations including cup product which can induce non-Clifford gates. When applying this mapping to a 3D hypergraph product code obtained from the product of 3 copies of good classical expander codes, we obtain non-Clifford logical CCZ gates via constant depth circuits on a code with constant stabilizer weight $w=O(1)$,  constant rate $K=\Theta(N)$, and polynomial distance $D=\Omega(N^{1/3})$. When applied to logical CCZ on 3D homological product codes consisting of the product of a pair of good quantum and classical LDPC codes, we can further improve the distance  to $D=\Omega(\sqrt{N})$ exceeding the  $N^{1/3}$ distance barrier implied by the Bravyi-König bound for conventional topological codes with the aid of non-Euclidean geometries. 
Our work suggests that it is feasible to  apply native logical non-Clifford gates on qLDPC codes  or directly inject high-fidelity magic states as resources (`\textit{magic state fountain}') without the distillation process. For the homological product  construction, the fountain can inject $\Theta(\sqrt{N})$ magic states in parallel in a single round. 




\end{abstract}

\maketitle

\tableofcontents

\section{Introduction}

With the recent advancement of quantum computing technology,  we have entered the era of fault-tolerant quantum computing at the scale of $O(10)$ to $O(100)$ qubits \cite{Bravyi:2024wc}.  A fundamental question towards further scaling up fault-tolerant quantum computation is how to minimize the space-time overhead. 

In recent years, significant progress has been made on the theory of \textit{quantum low-density parity check}  (qLDPC) codes in terms low-overhead quantum information storage \cite{fiberbundlecode21,9490244,hastingswr21,pkldpc22,Quantum_tanner, lh22,guefficient22,dhlv23,lzdecoding23,gusingleshot23, Bravyi:2024wc}.  This includes the discovery of the \textit{asymptotically good qLDPC codes} by Panteleev and Kalachev \cite{pkldpc22} achieving the optimal storage with constant space overhead and linear distance (see also \cite{Quantum_tanner}).   Nevertheless, fault-tolerant quantum computation requires not only a quantum memory, but also logical operations on top of that.  A fundamental question is hence whether there exists an \textit{asymptotically good quantum processor} which has constant space-time overhead in the computation with parallelizable logical gates and also linear distance, or some construction closely approaches that \cite{gottesman14, nguyen2024quantum}.  

Currently, there have been some ongoing efforts on the study of fault-tolerant logical gates on qLDPC codes \cite{cohen22, huang2023homomorphic,  xu2024fast, cross2024improved, williamson2024low, swaroop2024universal}.   The majority of them focus on performing logical measurements such as lattice-surgery-based protocols \cite{cohen22,  cross2024improved, williamson2024low, swaroop2024universal} and homomorphic measurements \cite{huang2023homomorphic, xu2024fast} to implement logical Clifford gates.  Additional schemes of implementing non-Clifford gates such as the magic state distillation \cite{bravyi2012magic} are required to make the fault-tolerant computation universal.   A brute-force approach would be to inject and distill the magic states with 2D surface codes and then SWAP them into the qLDPC code block. However, the state injection in this case cannot be parallelizable without increasing the space overhead, and one has not fully leveraged the power of qLDPC codes.    Therefore, a parallelizable scheme for magic state injection in qLDPC codes is highly desirable. 

The alternative to magic state distillation is to apply native transversal non-Clifford gates or directly inject high-fidelity magic states in qLDPC without distillation (`\textit{magic state fountain}') \cite{zhu2023non},  which can eliminate the costly space-time overhead of performing multiple rounds of distillation.   A well-known scheme is to apply the transversal $T$ gate on a 3D color code \cite{bombin:2007eh, 
Bombin_2015, Kubica2015, Kubica:2015br, bombin2018transversal, bombin20182d}, or equivalently the transversal CCZ gate on three copies of 3D surface codes \cite{Kubica:2015br, Vasmer2019}.  However, such a scheme requires an additional $O(d)$ space-time overhead  compared to the 2D surface code. The code distance $d$ in this case scales as $d=O(N^{\frac{1}{3}})$ where $N$ is the total number of qubits in a single code block, which scales worse than a 2D surface code $d=O(\sqrt{N})$.  Detailed numerical comparisons have been made between these two approaches, and it was shown that the 3D color code only outperforms the surface-code magic state distillation at an error rate much lower than the threshold \cite{Kubica:2021}.  This fundamental $N^{\frac{1}{3}}$~-distance barrier is implied by the Bravyi-Konig bound \cite{Bravyi:2013dx} stating that for an $n$-dimensional topological code defined on an $n$-dimensional Euclidean lattice, the logical gates have to lie within the $n^\text{th}$ level of Clifford hierarchy.  Therefore, it is so far a significant challenge to find native logical non-Clifford gates that exceeds the $N^{\frac{1}{3}}$-distance barrier.  A crucial observation in this paper is that qLDPC codes essentially correspond to highly non-Euclidean geometry, and it is hence possible to go beyond the $N^{\frac{1}{3}}$-distance barrier, in particular,  a non-Clifford logical gate is found for a family of qLDPC codes with distance $D=\Omega(\sqrt{N})$.  

Another key insight in this work is that qLDPC code is not only efficient for quantum information storage, but also extremely efficient for producing and storing high-fidelity resource states such as magic states.   For example, when using a specific family of constant-rate qLDPC codes with code parameters $[[N, \Theta{(N)}, \Omega(\sqrt{N})]]$ as a magic state fountain, one can inject $\Theta(\sqrt{N})$ magic states in parallel in a single round with an effective distance $D=\Omega(\sqrt{N})$. 

The first scheme of applying non-Clifford logical gate to a high-rate qLDPC code has been proposed recently in Ref.~\cite{zhu2023non}, where collective logical CCZ gates have been realized on homological qLDPC codes defined on 3-manifolds (the quasi-hyperbolic codes) with almost-constant rate (up to logarithmic reduction) and logarithmic distance. A further improvement by a logarithmic factor and achieving a constant encoding rate along with a logarithmic code distance by combining the quasi-hyperbolic codes with the quantum rainbow code has also been proposed in Ref.~\cite{scruby2024quantum}. Moreover, It has been realized in Ref.~\cite{zhu2023non} that a logical CCZ gate implemented by a constant-depth circuit in three identical copies of homological qLDPC codes can be understood as a \textit{cohomology operation} corresponding to a 3-fold \textit{cup product}. Same correspondence also exists for the transversal $T$ gate in a 3D color code. Physically,  this logical gate corresponds to the emergent higher symmetry  \cite{gaiotto2014, benini2019, cordova2022, mcgreevy2022,  barkeshli2023codimension} in a topological quantum field theory (TQFT): the $\ZZ_2^3$ gauge theory,  as has been studied in  Refs.~\cite{barkeshli2023codimension, barkeshli2022higher}. Such a higher symmetry also corresponds to sweeping a gauged symmetry-protected topological (SPT) defect  \cite{Yoshida_gate_SPT_2015, Yoshida_global_symmetry_2016, Yoshida2017387,   Webster_gates_2018, barkeshli2023codimension, barkeshli2022higher}.   A recent work has further explored and  classified logical gates in quantum codes via cohomology operations, which can go beyond the $k$-fold cup products corresponding to the color-code paradigm \cite{Hsin2024:classifying}.  

In algebraic topology, the cohomology operations such as cup products are defined on a \textit{simplicial complex} structure \cite{Hatcher:2001ut} including the special case of a triangulated manifold.  The homological qLDPC codes defined on a manifold naturally admit such cohomology operations as well as the TQFT description.   On the other hand, a large class of qLDPC codes with desirable parameters are defined on a general chain complex that has a large expansion property, including the homological product of expander graphs \cite{Tillich:2014_hyergraph_product, Bravyi:2014bq} or high-dimensional expanders \cite{pkldpc22, lubotzky2018high}.    Nevertheless, a recent breakthrough by Freedman and Hastings \cite{freedman:2020_manifold_from_code} has unified these two different worlds: one relies on the systolic geometry of manifolds while the other focuses on the combinatorics of expanders.  In particular, they show that any qLDPC code that is sparsely liftable can be mapped to a manifold with minimal dimension 11 which has \textit{bounded local geometry} (the corresponding triangulation has bounded degree).  This code-to-manifold  mapping essentially erodes the distinction between general qLDPC (or even more generally CSS) codes and homological qLDPC codes defined on manifolds.  In particular, the expansion properties and code parameters in the expander-based qLDPC codes are inherited by the homological codes defined on the manifold produced by the mapping.   Nevertheless,  this mapping is only applicable to a code defined on a 2D (3-term) chain complex. In order to apply non-Clifford gate which is at least in the third level of Clifford hierarchy, one needs to use a quantum code defined on a 3D (4-term) chain complex.

In this paper, we further generalize the mapping in Ref.~\cite{freedman:2020_manifold_from_code} to a mapping from a classical code to a manifold, which has also been suggested in Ref.~\cite{freedman:2020_manifold_from_code}. The classical code is associated with a \textit{Tanner graph}, which is a bipartite graph with two types of  vertices corresponding to the bit and check variables respectively.  It is also equivalent to a \textit{hypergraph} where checks are placed on the vertices and bits on the hyperedges.  Unfortunately, a general hypergraph is not a simplicial complex, and hence does not simply  admit cohomology operation such as cup products. One could consider using a classical homological code defined on a graph where the bits are all placed on edges instead of hyperedges; however, it is proven that such classical codes cannot be good, i.e., having linear distance,  since the corresponding graph will have 1-cycles with a logarithmic upper bound in size \cite{Meshulam:2018}. One way around this is to use the Sipser-Spielman construction \cite{Sipser_Spielman} to place a local code on each vertex, which is mathematically equivalent to a sheaf \cite{Meshulam:2018}. However, cohomology operation such as cup product can only be defined on a specific type of sheaf codes satisfying certain combinatorial condition \cite{lin2024transversal, golowich2024quantum}. Another way out which is more generally applicable is to go to a high-dimensional topological expander \cite{lubotzky2018high} associated with a higher-dimensional simplicial complex (dim $> 2$) \cite{Meshulam:2018}. Indeed it is  possible to construct good classical codes on high-dimensional random simplicial complex, as shown in Ref.~\cite{10.13069/jacodesmath.617235} (based on the random complex constructions in Refs.~\cite{10.1007/s00454-017-9926-3, 10.48550/arxiv.1010.1400, 10.1007/s00493-006-0027-9, 10.48550/arxiv.math/0609773, 10.48550/arxiv.math/0609773}), although these codes are not LDPC. Here, we instead consider building a higher-dimensional manifold (8-manifold) following the general spirit in Ref.~\cite{freedman:2020_manifold_from_code}, and obtain the simplicial complex from its triangulation. We start with the corresponding Tanner graph (or equivalently a hypergraph) of a classical LDPC code as the skeleton and then `thicken' it into a manifold via handle construction. The manifold has a bounded local geometry due to the sparseness of the input Tanner graph equivalent to the LDPC condition, and hence admits a triangulation with bounded degree.  This in turn gives rise to families of good classical LDPC codes on simplicial complexes that  are high-dimensional topological expanders \cite{lubotzky2018high} without the need of local codes (Theorem \ref{theorem:good_code_on_manifold}), which is interesting in its own right.

\begin{figure}[t]
\includegraphics[width=1\columnwidth]{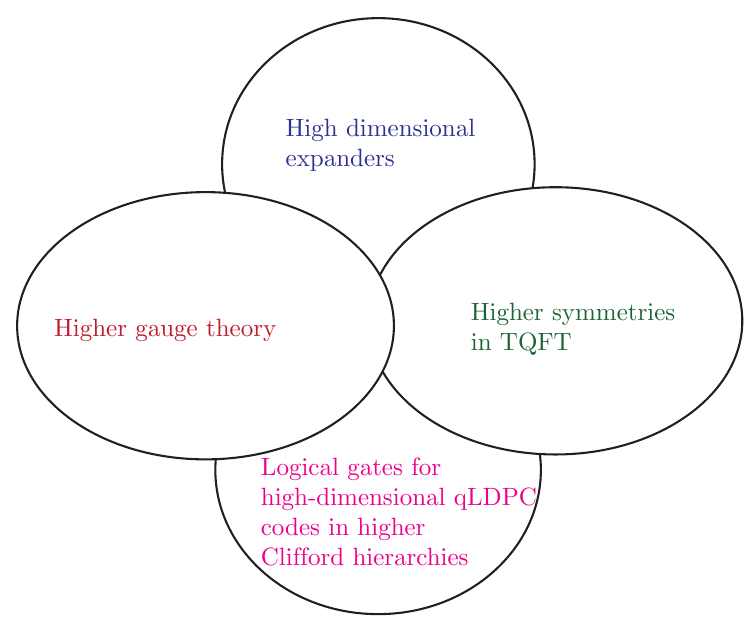}
\caption{Several `higher structures' across the area of physics, computer science and mathematics are connected through this work. }\label{fig:van-diagram}
\end{figure}

One can then use the manifolds built from classical or quantum codes as Legos to further construct a product manifold.  For an input qLDPC code obtained from a product construction (i.e., via a homological product \cite{Bravyi:2014bq} or more generally balanced product \cite{9490244}), one can use this method to build a thickened qLDPC code defined on the product of the manifolds built from either classical or quantum codes.  This gives access to the mapping of higher-dimensional qLDPC codes that go beyond a 2D (3-term) chain complex to manifolds.   

With the simplicial complex structure obtained from the manifold triangulation, we are now able to introduce cohomology operations including the triple cup products, which correspond to the emergent higher symmetries in a higher gauge theory equivalent to a qLDPC code with qubits placed on higher-dimensional simplices (dimensions equal or higher than two).  This gives rise to logical gates in higher Clifford hierarchies. Quite interestingly, we can see the deep connection between a few higher-dimensional structures across the area of physics, computer science and mathematics \cite{lubotzky2018high} through this work, as illustrated in Fig.~\ref{fig:van-diagram}.   

One technical aspect is that the original construction in Ref.~\cite{freedman:2020_manifold_from_code} considers high-dimensional manifolds (at least 11D) in order to have a separation in dimension with the spurious homologies,  which is important from the perspective of systolic geometry and also necessary to avoid short distance in the corresponding subspace code.  However, for the purpose of this paper, we can instead consider a subsystem-code encoding and treat the logical qubits with short logical operators as gauge qubits. In this way, we can use a manifold with much lower dimension (4D) obtained from the classical code to construct the qLDPC codes via product construction.  Practically this can further reduce the average/maximal degree of the triangulations and hence the average/maximal stabilizer weight of the constructed codes.         

In this work, we have obtained two types of qLDPC constructions.  The skeleton of the first construction is based on a 3D hypergraph product code \cite{Tillich:2014_hyergraph_product, Quintavalle:2020_3D} obtained from the homological product of three identical good classical expander codes.  One further thickens each factor classical code into a 4-manifold, and then obtain a thickened 3D hypergraph product code defined on the triangulation of a 12-manifold (qubits  placed on 4-simplices) with constant stabilizer weight $w=O(1)$, constant encoding rate (linear dimension) $K=\Theta(N)$, polynomial subsystem-code distance $D=\Omega(N^{\frac{1}{3}})$, where $K$ and $N$ represent the number of logical and physical qubits respectively.  Then a constant-depth circuit implementing the cohomology operation of a triple cup product between three identical copies of the thickened codes lead to a collective logical CCZ gate, consisting of $\Theta(N)$ CCZ's which equal the total number of $\ZZ_2$ triple intersection points, and  addressing $\Theta(N^{\frac{2}{3}})$ logical qubits (Theorem \ref{theorem:3Dhypergraph}).  The partial addressing issue is resolved in the second construction. 

We now consider the second construction, which is based on a 3D homological product code obtained from the product of a good classical expander code and a good quantum LDPC code.  The classical and quantum codes are then mapped to the 4-manifold and 11-manifold respectively, which gives rise to the thickened homological product code defined on a 15-manifold with constant stabilizer weight $w=O(1)$, constant encoding rate $K=\Theta(N)$, and subsystem-code distance $D=\Omega(\sqrt{N})$ which go beyond the $N^{\frac{1}{3}}$-distance barrier implied by the Bravyi-König bound for conventional topological codes defined on Euclidean lattices \cite{Bravyi:2013dx}. This is possible due to the fact that the good quantum LDPC code in Ref.~\cite{pkldpc22} is constructed from a twisted product (or equivalently a balanced product code). The corresponding product manifold built out of this is highly non-Euclidean and the cycles can have much larger size than the case of a Cartesian product. This thickened code is a tensor product of three copies of non-identical qLDPC codes corresponding to different higher gauge theories where the qubits are placed on 6-simplices, 2-simplices and 7-simplices respectively.  The triple-cup product cohomology operation between these three non-identical copies gives rise to a collective logical CCZ gates addressing all the $K=\Theta(N)$ logical qubits and contain in total $\Theta(N)$ CCZs (Theorem \ref{theorem:homological_product}).   

For the homological product code construction, we further investigate the magic state fountain scheme first envisioned in Ref.~\cite{zhu2023non}, where we can directly inject $\Theta(\sqrt{N})$ non-overlapping high-fidelity CCZ magic states with effective distance $\Omega(\sqrt{N})$ into the qLDPC code in parallel in each single round (Corollary \ref{corollary:homological}). Both the injection rate and fidelity (effective distance) outperform those for the 3D topological color codes defined on a 3D cube with code parameters $[[N,3, \Theta(N^{1/3})]]$, where one can only inject a single CCZ magic state with effective distance $\Theta(N^{1/3})$  (see Sec.~\ref{sec:fountain} for more detailed comparison). We further show how to perform gate teleportation to implement logical CCZ gates using these magic states as resources.     Although we have not reached the optimal $\Theta(N)$ injection rate per round, we note that currently it is not the bottleneck for fault-tolerant computation on qLDPC codes, since so far there is no fully parallelizable logical measurement scheme that can implement $\Theta(N)$ logical Clifford gates (generating the whole Clifford group) in a single logical cycle.  A naive estimate,  assuming non-overlapping logical Pauli measurements can be done in parallel, would lead to at most $O(\sqrt{N})$ logical gates per logical cycle for a qLDPC code with $O({\sqrt{N}})$ distance, which just coincides with the injection rate of the magic state fountain.  Further improvement in the injection rate to $\Theta(N)$ per round would require us to introduce more separable triple intersection structure into the manifold, likely via choosing more non-trivial maps during the handle attachment.    

The above discussion focuses on the asymptotic regime, while cautious readers may question the practicality in near-term implementation since the dimension looks high and the  overhead of subdividing the manifold into triangulations may be a large constant. To eliminate this worry, we further show that one can \textit{deformation retract} the manifold to a cellular chain complex $\L_c$ (often called CW complex \cite{Hatcher:2001ut}) as hinted in Ref.~\cite{freedman:2020_manifold_from_code} and elaborated in Ref.~\cite{guemard2025lifting}, which is isomorphic to the handle chain complex $\L_h$ used for the handle construction.  Each $k$-handle is retracted to its core---the $k$-cell. Due to the isomorphism, the classical or quantum code $\C$ defined on the CW complex is completely the same as the skeleton code $\bar{\C}$, while the hidden CW complex structure has mathematically well-defined cup product which can implement logical gates.  Note the code $\C$ on the CW complex is non-topological,  since $\L_c$ is no longer a discretization of the manifold like the triangulation. This compact realization makes the near-term implementation practical and has minimal overhead.   Interestingly, the isomorphism between $\L_c$ and $\L_h$ presdrves the Poincar\'e duality and $\L_c$ is hence a \textit{Poincar\'e complex}.  

For the conceptual understanding rather than practical purpose, we can also obtain a subspace code construction instead of using subsystem-code encoding for the thickened homological product.  This is achieved in the follow-up paper \cite{zhu2025transversal}.

The work is organized as follows.  In Sec.~\ref{sec:gates_simplicial_complex},  we introduce the general theory of logical gates implemented via cohomology operations on a CSS code defined on a simplicial complex including the case of triangulation of a manifold.  The theory uses an operator-valued cochain formalism, physically corresponding to a gauge field formalism, which has been previously introduced in Refs.~\cite{zhu2023non, barkeshli2023codimension, barkeshli2022higher}.  In particular, we show the construction of triple cup products in a higher gauge theory equivalent to three non-identical copies of CSS codes, which will be used in constructing the non-Clifford logical gates in the two qLDPC codes introduced in this paper. The formalism also shows explicitly how to construct the constant-depth circuits composed of overlapping physical CCZ gates which in term give rise to the logical CCZ gates. In Sec.~\ref{sec:hypergraph-product-code}, we introduce a 3D hypergraph product construction based on good random classical expander codes. In Sec.~\ref{sec:non-Clifford_hypergraph}, we introduce the technique of building manifolds from the Tanner graph of the skeleton classical codes using handle construction, including the 8-manifold construction following the Freedman-Hastings construction in Ref.~\cite{freedman:2020_manifold_from_code} and a modified lower-dimensional construction of 4-manifolds.  We also show the details of the mapping between the cycles/cocycles in the skeleton classical code and the thickened cyclces/cocycles in the corresponding manifold.   These then pave the way for the construction of the thickened 3D hypergraph product code and derive the scaling of the encoding rate and code distance.    
We then show the existence of non-trivial triple cup product structure between three cocycles which geometrically corresponds to the triple intersection of their Poincar\'e dual cycles.   This gives rise to the collective logical CCZ gates.   We then count the number of triple intersection points to estimate the number of CCZ's.  In Sec.~\ref{sec:good_product}, we introduce the thickened 3D homological product code construction, including both the subsystem and subspace code versions.  We then construct its logical CCZ gates via a triple cup product in the higher gauge theories.   We further introduce the magic state fountain scheme and show how to inject magic states in parallel and consume these states for the gate teleportation protocol to implement parallelizable logical CCZ gates, as well as to derive its injection rate. 
 We conclude our paper with the discussion and outlook of future directions and open problems. 

\textit{Note added---} During the preparation of this manuscript, we became aware of several other works on related topics \cite{golowich2024quantum, lin2024transversal, breuckmann2024cups}.  The main difference is that the construction in the above papers needs to impose certain local combinatorial conditions for the underlying codes, while the construction of the present paper is fully topological and is applicable to arbitrary input classical or quantum codes, including the Sipser-Spielman codes (sheafs) with local codes or randomly constructed expander codes without local codes.  This flexibility allows us to use the good qLDPC codes to break the $N^{1/3}$ distance barrier.
The construction in Ref.~\cite{lin2024transversal, golowich2024quantum} has further developed the idea of cup product and triple intersection points in Ref.~\cite{zhu2023non} to the context of quantum sheaf codes.  The specific construction of the algebraic codes in Ref.~\cite{golowich2024quantum} leads to a stabilizer weight $w=polylog(N)$ and has not yet fully satisfied the qLDPC condition. The encoding rate obtained in Ref.~\cite{golowich2024quantum} for a  logical CCZ gate is close to a constant, i.e., $K=\Theta(N^{1-\epsilon})$ and the distance is $O(N^\frac{1}{3})/polylog(N)$.   In contrast, the present paper has achieved constant stabilizer weight $w=O(1)$ and constant encoding rate $K=\Theta(N)$, and has also gone beyond the $N^{\frac{1}{3}}$-distance barrier and achieved an $\Omega(\sqrt{N})$ distance.  On the other hand, the scheme in Ref.~\cite{golowich2024quantum} has also achieved the $\gamma \rightarrow 0$ magic state distillation \cite{wills2024constant, nguyen2024good, golowich2024asymptotically}, which has not yet been achieved in the present paper.

\section{Non-Clifford logical gates on general simplicial  complexes via cohomology operation}\label{sec:gates_simplicial_complex}

In this section, we first describe CSS codes defined on arbitrary simplicial complexes (including the triangulations on a manifold) in the language of $\ZZ_2$ lattice gauge theory and then introduce the operator-valued cochain formalism (also called a gauge field formalism) to describe the cup product cohomology operations.  We then show how to use the cohomology operation to perform a constant-depth circuit corresponding to logical non-Clifford gates.

\subsection{Logical Clifford gates for a 2D complex}
We first consider a quantum  CSS code defined on a 2D simplicial complex $\mathcal{L}$.  The CSS code can be described by the following 2D (3-term) chain complex:
\be
C_{2} \xrightarrow[]{\partial_{2}=\mathbf{H}_Z^T} C_1 \xrightarrow[]{\partial_{1}=\mathbf{H}_X} C_{0},
\ee
where $\mathbf{H}_X$ and $\mathbf{H}_Z$ are the parity check matrix for the $X$- and $Z$-checks respectively, $C_i$ denotes the $i^\text{th}$ chain group, and $\partial_i$ denotes the $i^\text{th}$ boundary map.  The code space is defined by $\mathcal{C}= \mathbb{C}_2^{|H_1(\mathcal{L}; \ZZ_2)|}$, where $H_1(\mathcal{L};\ZZ_2)=\text{Ker}(\partial_1)/\text{Img}(\partial_2)$ is the 1st $\ZZ_2$-homology group.

A dual description of the same code is the following cochain complex:
\be
C^2 \xleftarrow[]{d_{1}=\mathbf{H}_Z} C^1 \xleftarrow[]{d_{0}=\mathbf{H}_X^T} C^{0},
\ee
where $C^i$ and $d_i$ denote the $i^\text{th}$ cochain group and coboundary operator respectively. For simplicity, we sometimes suppress the index $i$ and write the coboundary operator as $d$. In this dual description, the code space is defined by  $\mathcal{C}= \mathbb{C}_2^{|H^1(\mathcal{L}; \ZZ_2)|}$, where $H^1(\mathcal{L};\ZZ_2)=\text{Ker}(d_1)/\text{Img}(d_0)$ is the 1st $\ZZ_2$-cohomology group.

We now introduce an \textit{operator-valued cochain formalism} (also called a \textit{gauge field formalism}) to describe the operators in the code, which originates from a (2+1)D topological quantum field theory (TQFT): the $\ZZ_2$ lattice gauge theory.    We define the operator valued 1-cochain $\hat{a} \in C^1(\mathcal{L}; \ZZ_2)$, with its eigenvalues belonging to $\{0,1\}$. Here, the hat $\hat{\cdot}$ indicates that it is a quantum variable (operator). The coefficient of each edge (1-cell) $e$ in the 1-cochain corresponds to a Pauli-$Z$ operator as 
\be\label{eq:a_Pauli}
(-1)^{ \hat{a} (e)}= Z(e) \in \{-1,1\}.
\ee
Note that the Pauli operator has eigenvalues $\pm 1$ instead.
Physically, $a$ corresponds to the 1-form $\ZZ_2$ electric gauge field.  We note that the eigenstate of the operator-valued 1-cochain $\hat{a}$ is a cochain state $\ket{c^1}$ in the diagonal basis,  where the 1-cochain $c^1 \in C^1 (\L; \ZZ_2)$ (classical variable) stores the bit-string with value 0 or 1 on each 1-simplex, e.g. $\ket{c^1} = \ket{0110011\cdots}$.  We can express the relation as:
\be
\hat{a} \ket{c^1} = c^1 \ket{c^1},
\ee
where $c^1$ just stores the eigenvalues of $\hat{a}$. We can also equivalently express the operator-valued cochain as:
\be
\hat{a}= c^1 \ketbra{c^1}.
\ee

Similarly, there also exists a 1-cochain $\hat{b}$ corresponding to the magnetic $\ZZ_2$  gauge field, which is related to a Pauli-X operator on each edge $e$ as 
\be\label{eq:b_Pauli}
(-1)^{\hat{b} (e)}= X(e) \in \{-1,1\}.
\ee
For simplicity, we will omit the hat $\hat{\cdot}$ on $a$ and $b$, while the reader should keep in mind they are actually operators.

The coboundary operator, a lattice analog of the exterior derivative in the continuum case, acts on the $\ZZ_2$-valued 1-cochain as 
\be
(da)(f) = \sum_{e \subset f}a(e).
\ee
The $Z$-stabilizer of the CSS code localed on face $f$ can be expressed in terms of the operator-valued cochain as
\be\label{eq:Z_stabilizer}
S^Z_f = \prod_{e \subset f} Z(e)= (-1)^{\sum_{e \subset f}a(e)} = (-1)^{da(f)},
\ee
which is also considered as a flux term in the gauge theory. The $Z$-stabilizer condition $S^Z_f=1$ becomes the 0-flux condition \footnote{This can be considered as the lattice analogy of the 0-flux condition in the continuum $\vec{B}=\nabla \times \vec{a}$=0, where $\vec{a}$ is the quantum vector potential.}
\be
da(f)=0.  
\ee

The $X$-stabilizers on the vertex $v$ can be expressed as
\be
S^X_v = \prod_{e \supset v} X(e) = (-1)^{\sum_{e \supset v} b(e)}.
\ee

 The X-stabilizer condition $S^X_v=+1$ corresponds to Gauss's law (with zero charge) in the $\ZZ_2$ gauge theory and can be expressed as \footnote{This is a lattice analog of Gauss's law in the continuous space $Q=\int_\Sigma \vec{E} \cdot d\vec{S} =0$ or equivalently $\nabla \cdot \vec{E}=0$, where $\vec{E}$ is the quantum electric field and corresponds to the cochain $b$ in the lattice gauge theory.}
\be
\sum_{e \supset v} b(e)  = 0.
\ee
Furthermore, the anticommutation relation between the Pauli-$X$ and -$Z$  operator $X_eZ_{e'}=(-1)^{\delta_{e,e'}}Z_eX_{e'}$ leads to the following anticommutation relation:
\be\label{eq:anti-commutation}
  X_e (-1)^{f(a)}X_e=(-1)^{f(a+\bar{e})},
\ee
where $f$ is an aribtrary function of $a$ and $\tilde e$ is the indicator $1$-cochain that takes value $1$ on edge $e$ and 0 otherwise.

We now introduce cup products on a general $n$-dimensional simplicial complex $\mathcal{L}$.  The vertices on the simplicial complex $\L$ are assigned with a fixed global ordering.  
The cup product `$\cup$' of a $p$-cochain $\alpha^p$ and a $q$-cochain $\beta^q$ gives rise to a $(p+q)$-cochain denoted by $\alpha^p \cup \beta^q$ as expressed by the following bilinear map on the cochain groups:
\be
\cup :  C^p(\L) \times C^q(\L) \rightarrow C^{p+q}(\L).
\ee
One can explicitly define the cup product  between  $\alpha^p$ and $\beta^q$ evaluated on a $(p+q)$-simplex $[v_0,v_1,\cdots,v_{p+q}]$ as \cite{Hatcher:2001ut}: 
\begin{widetext}
\begin{align}\label{eq:cup_def}
  (\alpha^p \cup \beta^q)([v_0,v_1,\cdots, v_{p+q}]) 
=\alpha^p([v_0,v_1,\cdots, v_{p}])\beta^q([v_p, v_{p+1},\cdots, v_{p+q}])~. 
\end{align}	
\end{widetext}
Here, the arguments contain the labels of ordered vertices $v_i$ with the ordering $v_0 \prec v_1 \prec v_2 \cdots \prec v_{p+q}$. We illustrate the  $p=q=1$ case in Fig.~\ref{fig:cup_product_triangulation}(a). The cup product also induces a bilinear operation on cohomology:
\be
\cup :  H^p(\L) \times H^q(\L) \rightarrow H^{p+q}(\L).
\ee

We then introduce logical gates in two identical copies of CSS codes supported on a generic 2D simplicial complex $\mathcal{L}$ via cohomological operation.   One can consider the following unitary:
\be\label{eq:logical_CZ}
U = (-1)^{\int_{\Sigma_2} a \cup a'},
\ee
where $a$ and $a'$ correspond to the electric gauge field on the two copies respectively, and $\int_{\Sigma_2}$ is a discrete sum over all the simplices in a non-trivial 2-cycle $\Sigma_2$ of the simplicial  complex $\mathcal{L}$. This sum can also be viewed as a cycle-cocycle pairing between $\Sigma_2$ and $a\cup a'$, which is an inner product between the two vectors with $\ZZ_2$ coefficients associated with the cycle and cocycle respectively.  In the case that the simplicial complex $\L$ forms the triangulation of a 2-manifold $\M^2$, we just take $\Sigma_2=\M^2$ since the only non-trivial 2-cycle is the entire manifold:
\be
U = (-1)^{\int_{\M^2} a \cup a'},
\ee
where $\int_{\M^2}$ means summing over the simplices of the triangulation belonging to  the manifold $\M^2$.    Nevertheless, in the case of a general simplicial complex, the choice of $\Sigma_2$ is not unique and this also leads to the opportunity for targeted  logical gates acting on a subset of logical qubits, which will be studied in details in future works. 

We note that physically $U$ in Eq.~\eqref{eq:logical_CZ} is nothing but the partition function (discrete path integral) of a (1+1)D $\Z_2 \times Z_2$ symmetry-protected topological phase (SPT) corresponding to the type-II cocycle \cite{barkeshli2023codimension}.

We now present the following lemma:
\begin{lemma}\label{lemma_gate_1}
The unitary $U = (-1)^{\int_{\M^2} a \cup a'}$ acting on two copies of CSS codes defined on a 2D simplicial complex is a constant-depth local quantum  circuit  that implements the collective logical CZ gates.
\end{lemma}

\begin{proof}
The cup product evalued on each 2-simplex $[v_0,v_1,v_2]$ ($v_i$ represents the three vertices in the simplex) can be computed as
\be
(a \cup a') ([v_0,v_1,v_2]) = a([v_0,v_1])a'([v_1,v_2]),
\ee
as illustrated in Fig.~\ref{fig:cup_product_triangulation}(a).
We can hence re-express the unitary $U$ as
\begin{align}\label{eq:logical_gate_2D}
\non U=& (-1)^{\int_{[v_0v_1v_2]\in \Sigma_2} a([v_0,v_1])a'([v_1,v_2])} \\
=&\prod_{[v_0 v_1 v_2] \in \Sigma_2} \text{CZ}([v_0,v_1],[v_1,v_2]),
\end{align}
which shows that $U$ is a constant-depth circuit composed of many CZ gates.  For each 2-simplex, the CZ gate is between qubits supported on the edges $[v_0,v_1]$ in the first copy of the quantum code and $[v_1,v_2]$ on the second copy respectively.  The CZ gate in each 2-simplex can overlap with the CZ gates in the neighboring 2-simplexes, which means $U$ is not a strictly transversal gate but only a constant-depth circuit.

\begin{figure}[t]
\includegraphics[width=1\columnwidth]{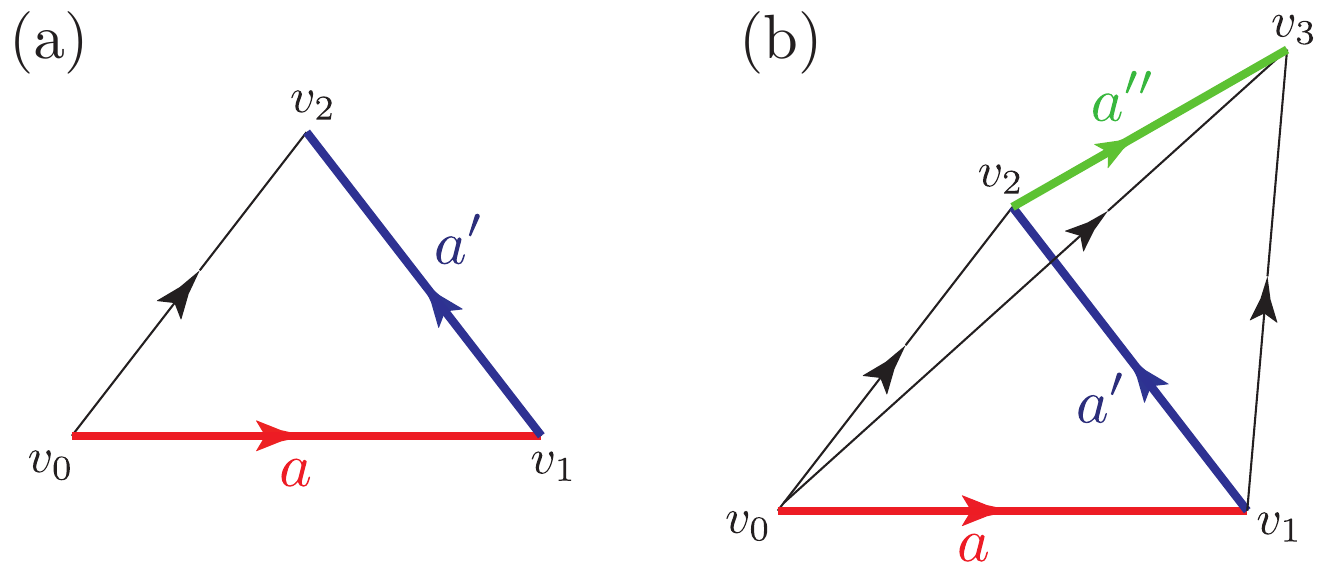}
\caption{Cup product definition on simplicial complexes, where the arrows point from vertices with lower order to vertices with higher order. (a)  Cup product of 1-cochains $a \cup a'$ on 2D simplicial complex. One takes product of cocycle values on the red and blue edges respectively in each 2-simplex (triangle). (b) Cup product of 1-cochains $a \cup a' \cup a''$ on a 3D simplicial complex. One takes the product of cocycle values on the red, blue, and green edges respectively in each 3-simplex (tetrahedron).}\label{fig:cup_product_triangulation}
\end{figure}

Now we need to check whether $U$ is a logical gate, meaning its action should preserve the code space, i.e., 
\be
U:\mathcal{C} \rightarrow \mathcal{C}.
\ee
An equivalent condition for the stabilizer code is that 
\be
P_\mathcal{C}[U, \mathcal{S}] P_\mathcal{C} = 0,
\ee
meaning that $U$ commutes with the stabilizer group $\mathcal{S}$ when projected to the code space $\mathcal{C}$.  Since $U$ is a diagonal gate, it clearly commutes with all the $Z$-stabilizers.   Now we only need to check the commutation relation for the $X$-stabilizers. 

We consider an $X$-stabilizer $S^X_{v;1}$$=$$\prod_{e \supset v} X^{(1)}(e)$ located at vertex $v$ in the 1st copy of the quantum code.   We then consider conjugating $U$ with $S^X_{v;1}$:
\begin{align}
\non & S^X_{v;1} U S^X_{v;1} \\
\non &= \prod_{e \supset v} X^{(1)}(e) \left[(-1)^{\int_{\Sigma_2} a \cup a'}\right]   \prod_{e' \supset v} X^{(1)}(e')  \\
\non &= (-1)^{\int_{\Sigma_2} (a +\sum_{e' \supset v}\bar{e}') \cup a'} \\
&= (-1)^{\int_{\Sigma_2} (a + d \bar{v}) \cup a'},
\end{align}
where we have used Eq.~\eqref{eq:anti-commutation} in the second equality.  In the last line, $d\bar{v}=\sum_{e' \supset v}\bar{e}'$ is the indicator 1-cochain taking value 1 on all the edges $e'$ connected to the vertex $v$ and 0 elsewhere, and $\bar{v}$ is the indicator cochain that is 1 at $v$ and 0 elsewhere.

We now require the following condition for $U$ to be a logical gate:
\be\label{eq:commuting_ground}
P_\mathcal{C} S^X_{v;1} U S^X_{v;1} P_\mathcal{C}= U,
\ee
where $P_\mathcal{C}$ is the projector to the code space $\mathcal{C}$. This condition means $U$ commutes with the stabilizer $S^X_{v;1}$ in the code space, which is then equivalent to the following condition:
\be\label{eq:gauge-invariance}
P_\mathcal{C} (-1)^{\int_{\Sigma_2} (a + d \bar{v}) \cup a'} P_\mathcal{C} = (-1)^{\int_{\Sigma_2} a \cup a'}. 
\ee
The condition in Eq.~\eqref{eq:gauge-invariance} should be considered as the gauge-invariance condition in the code space, since $a \rightarrow a+d \bar{v}$ is just a gauge transformation in the $\ZZ_2$ gauge theory \footnote{This is in analogy to the famous gauge transformation $\vec{a} \rightarrow \vec{a} +  \nabla \chi $ in Maxwell's theory or the corresponding gauge theory in the continuum, where $\vec{a}$ is the quantum vector potential. The gauge invariance of the magnetic field $B=\nabla \times \vec{a}=0$ is satisfied due to the Stokes theorem: $\nabla \times \nabla \chi =0$}.

In order to show the gauge invariance, we need to use the Leibniz rule for the cup product \cite{Hatcher:2001ut}:
\be\label{eq:Leibniz_rule}
d\bar{v}\cup a'= d(\bar{v} \cup a') + \bar{v} \cup d a'.
\ee
note that we have ignored the minus sign and replace it with the plus sign since these cochains are all $\Z_2$-valued (binary variables).   We hence have 
\begin{align}\label{eq:Stokes}
\non P_\mathcal{C}  \int_{_{\Sigma_2}} d\bar{v}\cup a' P_\mathcal{C} =& P_\mathcal{C} \left[\int_{_{\Sigma_2}} d(\bar{v} \cup a') + \int_{_{\Sigma_2}} \bar{v} \cup d a' \right] P_\mathcal{C}  \\
=& P_\mathcal{C} \int_{\partial{_{\Sigma_2}}}(\bar{v}\cup a') P_\mathcal{C}+ 0 =0,
\end{align}
where we have used the Stokes theorem in the second equality and the fact that we are considering a non-tirvial 2-cycle ${\Sigma_2}$ without boundaries, namely $\partial \Sigma_2=0$, to show that the first term is zero.  The second is zero due to the cocycle condition or physically the zeror-flux condition in the code space, i.e., $da'(f)=0$ on any face $f$, which corresponds to the stabilizer condition $S^Z_{f;2} = \prod_{e \subset f} Z^{(2)}(e) = (-1)^{da'(f)}=1$ according to Eq.~\eqref{eq:Z_stabilizer}. We have hence proved Eq.~\eqref{eq:commuting_ground} and \eqref{eq:gauge-invariance} and that $U$ is indeed gauge-invariant in the code space and hence a logical gate.  By symmetry, we can also verify the commutation relation with $S^{(2)}_{X; v}$ and $S^{(3)}_{X; v}$ in a similar way.   Therefore, the unitary $U$ is indeed a logical gate which keeps the code space invariant. 

\begin{figure}[t]
\includegraphics[width=1\columnwidth]{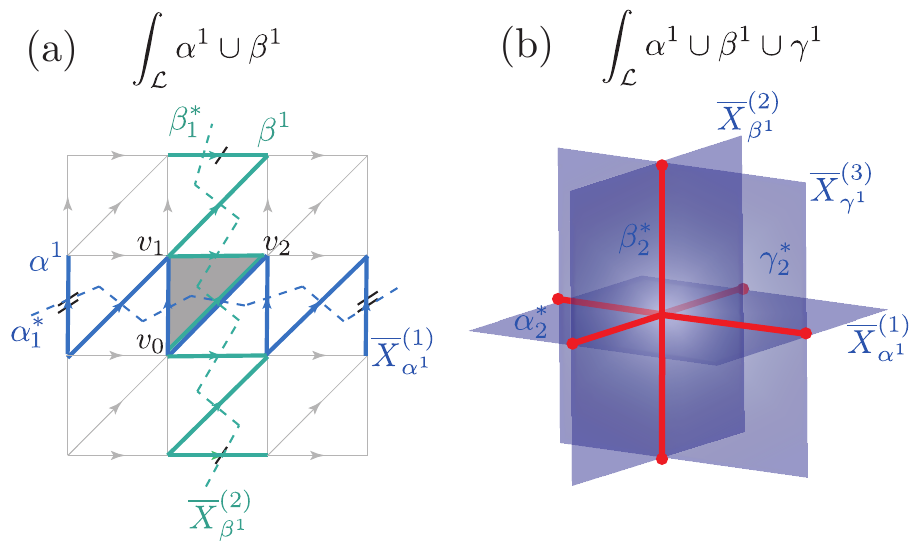}
\caption{(a) Illustration of the cup product and intersection in a 2D triangulation $\L$.  The 1-cocycle $\alpha^1$ (blue) and $\beta^1$ (green) corresponds to the support of the logical-$X$ operators $\lo{X}^{(1)}_{\alpha^1}$ and $\lo{X}^{(2)}_{\beta^1}$. The cup product $\alpha^1 \cup \beta^1$ evaluates non-trivially only on the highlighted (grey) triangle, hence the sum $\int_{\L}\alpha^1 \cup \beta^1=1$. This corresponds to a non-trivial intersection of the Poincar\'e dual 1-cycles $\alpha^*_1$ and $\beta^*_1$ (dashed lines) on the dual triangulation $\L^*$.   (b) The triple cup product sum $\int_{\L} \alpha^1 \cup \beta^1 \cup \gamma^1=1$ on a 3D triangulation $\L$, which corresponds a non-trivial triple intersection of the Poincar\'e dual 2-cycles (membranes) $\alpha^*_2$, $\beta^*_2$ and $\gamma^*_2$.  This can also be understood as the triple intersections of the logical-X membranes $\lo{X}^{(1)}_{\alpha^1}$, $\lo{X}^{(2)}_{\beta^1}$, and $\lo{X}^{(3)}_{\gamma^1}$.  }\label{fig:cup_product_illustrate}
\end{figure}

Mathematically, the gauge invariance condition is associated with the topological invariance of $U$ with respect to arbitrary deformation of the cocycles by coboundaries, which can hence also be called the \textit{coboundary invariance} and expressed as
\begin{align}
\non & P_\mathcal{C} (-1)^{\int_{\Sigma_2} (a^{(1)} + d \chi) \cup (a^{(2)} + d\lambda ) } P_\mathcal{C} \\
=& (-1)^{\int_{\Sigma_2} a^{(1)} \cup a^{(2)}},  
\end{align}
where $d\chi$ and $d\lambda$ are arbitrary coboundaries.

Now we show what type of logical gate this unitary $U$ corresponds to. Since the cochain $a$ and $a'$ in Eq.~\eqref{eq:logical_CZ} becomes cocycles in the code space due to the zero-flux condition $da=da'=0$, we can re-express them using the 1-cocycle basis $\{\alpha^1\}$ and $\{\beta^1\}$ for both copies of codes:
\be
a=\sum_{\alpha^1}\hat{n}_\alpha \alpha^1,  \quad a'=\sum_{\beta^1}\hat{m}_\beta \beta^1,
\ee
where the quantum variable $\hat{n}_\alpha$ and $\hat{m}_\beta$ with eigenvalues $\{0,1\}$ are the winding numbers for cocycles $\alpha^1$ and $\beta^1$ respectively.  
We can hence re-express $U$ as
\begin{align}
\non  U =& \prod_{\alpha^1, \beta^1} (-1)^{\int_{\Sigma_2} (\hat{n}_\alpha \alpha^1 )\cup (\hat{m}_\beta  \beta^1)}
   =\prod_{\alpha^1, \beta^1}  \left[(-1)^{ \hat{n}_\alpha \hat{m}_\beta}\right]^{\int_{\Sigma_2} \alpha^1 \cup \beta^1}\\
   =&\prod_{\alpha^1, \beta^1}  \overline{\text{CZ}}[(\alpha^1; 1), (\beta^1; 2)]^{\int_{\Sigma_2} {\alpha^1 \cup \beta^1}},
\end{align}
where the third equality has used the relation $(-1)^{\hat{n}\hat{m}}$$=$$\overline{\text{CZ}}$.  Here, $(\alpha^1; 1)$ and $(\beta^1; 2)$ are the labels of the logical qubit using the cocycle basis and the copy number.  The logical CZ gate between logical qubits $(\alpha^1; 1)$ and $(\beta^1; 2)$ is only non-trivial (not logical identity) if and only if $\int_{\Sigma_2} {\alpha^1 \cup \beta^1}$ evaluates non-trivially.

In the case that $\L$ is the triangulation of a closed 2-manifold $\M^2$, we can just take $\Sigma_2=\L$, i.e., the sum of 2-simplices on the entire 2-complex $\L$, which is equivalent to a 2-cycle since there is no boundary $\partial\L=0$.  The cup product sum now has a geometric interpretation as the $\ZZ_2$ intersection number of their Poincar\'e dual cycles:
\be
\int_{\L} {\alpha^1 \cup \beta^1} = | \alpha^*_1 \cap \beta^*_1|, 
\ee
as illustrated in Fig.~\ref{fig:cup_product_illustrate}(a).
This can also be interpreted as the $\ZZ_2$ intersection of the two logical-$X$ operators supported on the two cocycles $\lo{X}^{(1)}_{\alpha^1}$ and $\lo{X}^{(2)}_{\beta^1}$.

\end{proof}

\subsection{Logical non-Clifford gates for a 3D complex}
For a 3D simplicial complex $\L$, we still put qubits on the edge (1-cell),  $Z$-check on the face (2-cell), and X-check on the vertex (0-cell), which corresponds to the following chain complex:
\be
C_{3} \xrightarrow[]{\partial_3} C_{2} \xrightarrow[]{\partial_{2}=\mathbf{H}_Z^T} C_1 \xrightarrow[]{\partial_{1}=\mathbf{H}_X} C_{0}.
\ee
We still use 1-cochain $a$ and $b$ as the gauge field which corresponds to the Pauli $Z$ and $X$ operators respectively as in Eqs.~\eqref{eq:a_Pauli} and \eqref{eq:b_Pauli}.  

We can now define the desired unitary via cup products as \be\label{eq:logical_CCZ}
U = (-1)^{\int_{\Sigma_3} a \cup a' \cup a''},
\ee
In the above expression, $a$, $a'$ and $a''$ represent 1-cochains supported on each of the three copies of CSS codes respectively.

We now present the following lemma:
\begin{lemma}\label{lemma_gate_2}
The unitary $U = (-1)^{\int_{\Sigma_3} a \cup a' \cup a''}$ acting on three copies of CSS codes defined on a 3D simplicial complex is a constant-depth local quantum circuit that implements the collective logical CCZ gates.
\end{lemma}

\begin{proof}

In the case of simplicial 3-complex, the cup-product on each 3-simplex $[v_0,v_1, v_2, v_3]$  can be evaluated as
\be
(a \cup a' \cup a'') ([v_0,v_1,v_2, v_3]) = a([v_0,v_1])a'([v_1,v_2])a''([v_2,v_3]),
\ee
as illustrated in Fig.~\ref{fig:cup_product_triangulation}(b).
We can hence re-express the unitary $U$ as
\begin{align}\label{eq:logical_gate_3D}
\non U=& (-1)^{\int_{[v_0,v_1,v_2,v_3] \in \Sigma_3} a([v_0,v_1])a'([v_1,v_2])a''([v_2,v_3])} \\
=&\prod_{[v_0, v_1, v_2, v_3]\in \Sigma_3} \text{CCZ}([v_0,v_1],[v_1,v_2],[v_2,v_3]),
\end{align}
which explicitly shows the corresponding constant-depth circuit composed of many CCZ gates coupling the qubits in three copies of CSS codes. Similar to 2D case, one can also verify that $U$ preserves the code space $\C$ and is hence a logical gate as follows. 

We again consider an $X$-stabilizer $S^{(1)}_{X; v}$$=$$\prod_{e \supset v} X^{(1)}(e)$ in the 1st copy of the quantum code.   We then  conjugate $U$ with $S^{(1)}_{X; v}$:
\begin{align}
\non & S^{(1)}_{X; v} U S^{(1)}_{X; v} \\
\non =& \prod_{e \supset v} X^{(1)}(e) \left[(-1)^{\int_{\Sigma_2} a^{(1)} \cup a^{(2)} \cup a^{(3)}}\right]   \prod_{e' \supset v} X^{(1)}(e')  \\
\non =& (-1)^{\int_{\Sigma_2} (a^{(1)} +\sum_{e' \supset v}\tilde{e}')  \cup a^{(2)} \cup a^{(3)}} \\
=& (-1)^{\int_{\Sigma_2} (a^{(1)} + d \tilde{v}) \cup a^{(2)} \cup a^{(3)}}.
\end{align}

We now require the commutation condition Eq.~\eqref{eq:commuting_ground} to be satisfied such that $U$ is a logical gate, which is again equivalent to the gauge-invariance condition:
\be\label{eq:gauge-invariance2}
P_\mathcal{C} (-1)^{\int_{\mathcal{L}} (a^{(1)} + d \tilde{v}) \cup a^{(2)} \cup a^{(3)}} P_\mathcal{C} = (-1)^{\int_{\mathcal{L}} a^{(1)} \cup a^{(2)} \cup a^{(3)}}. 
\ee
The above gauge invariance can be proven using Eq.~\eqref{eq:Leibniz_rule} and \eqref{eq:Stokes}  as in the 2D case.

 By symmetry, we can also verify the commutation relation with $S^{(2)}_{X; v}$ and $S^{(3)}_{X; v}$ in a similar way. The unitary operator $U$ is hence a logical gate. 
 
Similar to the 2D case, we have also derived  the following coboundary-invarince condition:
\begin{align}
\non & P_\mathcal{C} (-1)^{\int_{\mathcal{L}} (a^{(1)} + d \chi) \cup (a^{(2)} + d\lambda ) \cup (a^{(3)}+d\eta)} P_\mathcal{C} \\
=& (-1)^{\int_{\mathcal{L}} a^{(1)} \cup a^{(2)} \cup a^{(3)}},  
\end{align}
where $d\chi$, $d\lambda$ and $d\eta$ are arbitrary coboundaries.

We note that physically $U$ in Eq.~\eqref{eq:logical_CCZ} is nothing but the partition function (discrete path integral) of a (2+1)D $\Z_2 \times \Z_2 \times \Z_2$ SPT corresponding to the type-III cocycle \cite{barkeshli2023codimension}.

To understand what logical gate $U$ corresponds to, we can re-express $U$ with the 1-cocycle basis $\{\alpha^1\}$, $\{\beta^1\}$ and $\{\gamma^1\}$ as
\begin{align}\label{eq:logical_CCZ}
 U =& \prod_{\alpha^1, \beta^1, \gamma^1} (-1)^{\int_{\Sigma_3} (\hat{n}_\alpha  \alpha^1 )\cup (\hat{m}_\beta  \beta^1) \cup (\hat{l}_\gamma  \gamma^1)} \cr
   =&\prod_{\alpha^1, \beta^1, \gamma^1}\left[(-1)^{ \hat{n}_\alpha \hat{m}_\beta \hat{l}_\gamma}\right]^{\int_{\Sigma_3} \alpha^1 \cup \beta^1 \cup \gamma^1} \cr
   =&\prod_{\alpha^1, \beta^1, \gamma^1} \overline{\text{CCZ}}[(\alpha^1; 1), (\beta^1; 2),(\gamma^1; 3)]^{\int_{\Sigma_3} {\alpha^1 \cup \beta^1 \cup \gamma^1}}. \cr
\end{align}
The logical CCZ gate between logical qubits $(\alpha^1; 1)$, $(\beta^1; 2)$ and $(\gamma^1; 3)$ is only non-trivial (not logical identity) if and only if $\int_{\Sigma_3} {\alpha^1 \cup \beta^1 \cup \gamma^1}$ evaluates non-trivially.

In the case that $\L$ is the triangulation of a 3-manifold $\M^3$, we can just take $\Sigma_3=\L$. The triple cup product sum now corresponds to the $\ZZ_2$ triple intersection number of the Poincar\'e dual 2-cycles:
\be
\int_{\L} {\alpha^1 \cup \beta^1 \cup \gamma^1} =  |\alpha^*_2 \cap \beta^*_2 \cap \gamma^*_2|,
\ee
as illustrated in Fig.~\ref{fig:cup_product_illustrate}(b).
This can also be interpreted as the $\ZZ_2$ triple intersection of the three logical-$X$ membrane operators supported on the three cocycles $\lo{X}^{(1)}_{\alpha^1}$,   $\lo{X}^{(2)}_{\beta^1}$ and $\lo{X}^{(3)}_{\gamma^1}$.
\end{proof}

\subsection{Generalization to $n$-dimensional complex and higher gauge theory}

We can be generalize the above results to $n$ copies of identical CSS codes supported on an $n$-dimensional simplicial complex $\L$ with qubits placed on the 1-simplices (edges). The corresponding 1-form gauge fields $a_{(i)}$ are operator-valued 1-cochains in the $i^\text{th}$ copy of code ($i=1,\cdots,n$). The constant-depth unitary circuit implementing the logical gate is hence:
\begin{equation}
 U=   (-1)^{\int_{\Sigma_n} a_{(1)}\cup a_{(2)}\cup \cdots a_{(n)}}~,
\end{equation}
where $\Sigma_n \in H_n(\L; \ZZ_2)$ is an abitrary $n$-cycle of $\L$.  In the case that the simplicial complex $\L$ forms the triangulation of a manifold $\M^n$, we can just take $\Sigma_n = \M^n$. 
One can verify that $U$ preserves the code space $\C$ in a similar way to the case of lower dimensions.

Now we start deriving the corresponding logical gate of the constant depth circuit $U$.  In the code space $\C$, all the gauge fields are cocycles satisfying $da_{(i)}=0$. We can hence re-express them using the 1-cocycle basis $\{\alpha^1_{(i)}\}$ of the cohomology group $H^1(\L^n;\Z_2)$ for the $i^\text{th}$ copy of codes: 
\begin{equation}
a_{(i)}=\sum_{\alpha^1_{(i)}}\hat{m}_{\alpha, i}\alpha^1_{(i)}~,
\end{equation}
where the quantum variable $\hat{m}_{\alpha, i}$ with eigenvalues $\{0,1\}$ are the winding numbers.  
The constant-depth circuit $U$ can hence be re-expressed as
\begin{align}
 U &=  (-1)^{\sum_{\{\alpha^1_{(i)}\}} \int_{\Sigma_n} (\hat{m}_{\alpha,1} \alpha^1_{(1)} )\cup (\hat{m}_{\alpha,2}  \alpha^1_{(2)}) \cup \cdots \cup (\hat{m}_{\alpha, n} \alpha^1_{(n)})}\cr
 \nonumber &  =\prod_{\{\alpha^1_{(i)}\}}\left[(-1)^{ \hat{m}_{\alpha,1} \hat{m}_{\alpha,2} \cdots \hat{m}_{\alpha, n}}\right]^{ \int_{\Sigma_n} {\alpha^1_{(1)} \cup \alpha^1_{(2)} \cup \cdots \cup \alpha^1_{(n)}}}\\
\nonumber   &=\prod_{\{\alpha^1_{(i)}\}} \overline{C^{n-1}Z}[\alpha^1_{(1)}, \alpha^1_{(2)}, \cdots, \alpha^1_{(n)} ]^{ \int_{\Sigma_n} {\alpha^1_{(1)} \cup \alpha^1_{(2)} \cup \cdots \cup \alpha^1_{(n)}}}~, \\
\end{align}
where we have used the relation $(-1)^{ \hat{m}_1 \hat{m}_2 \cdots \hat{m}_{n}}$$=$$\overline{C^{n-1}Z}$.  Here, $\alpha^1_{(i)}$ serve as the more compact labels of the logical qubits using the cocycle basis and the copy number $i$.  This logical $C^{n-1}Z$ gate is only non-trivial (not logical identity) if and only if the exponent $\int_{\Sigma_n}  {\alpha^1_{(1)} \cup \alpha^1_{(2)} \cup \cdots \cup \alpha^1_{(n)}}$ evaluates non-trivially.  In the case that the $n$D simplicial complex $\L$ is the triangulation of a manifold $\M^n$, we can just take $\Sigma_n=\M^n$.  The cup product sum in the exponent now corresponds to an $n$-fold $\ZZ_2$ intersection number of the Poincar\'e dual $(n-1)$-cycles $\alpha^*_{n-1; (i)}$: 
\begin{align}
& \int_{\M^n} {\alpha^1_{(1)} \cup \alpha^1_{(2)} \cup \cdots \cup \alpha^1_{(n)}} \cr
=& | \alpha^*_{n-1; (1)} \cap \alpha^*_{n-1; (2)} \cdots \cap \alpha^*_{n-1; (n)}|.
\end{align}

So far we have been focused on the case of 1-form gauge theories, which correspond to the CSS codes with qubits placed on the 1-simplices (edges) of the simplicial complex $\L$, i.e., associated with the 1-chain group $C_1$.    More generally, we also consider higher gauge theories, also called  $q$-form gauge theories, which corresponds to the CSS codes with qubits placed on the $q$-simplices of an $n$D simplicial complex $\L$, i.e., associated with the $q$-chain group $C_q$. The $X$-stabilizers and $Z$ stabilizers are placed on the $(q-1)$-simplices and $(q+1)$-simplices respectively. The corresponding chain complex is as follows: 
\begin{align}
&C_n \rightarrow \cdots \rightarrow C_{q+1} \xrightarrow[]{\partial_{q+1}=\mathbf{H}_Z^T} C_q \xrightarrow[]{\partial_{q}=\mathbf{H}_X} C_{q-1} \rightarrow \cdots \cr
&\qquad   \quad \quad \ Z\text{-stabilizer}  \qquad \   
 \text{qubit} \qquad X\text{-stabilizer}.
\end{align}
We then introduce the $q$-form electric gauge fields $(1 \le q \le n)$ as operator-valued $q$-cochains $\hat{a}^q \in C^q(\L; \ZZ_2)$. We note that the eigenstate of the operator-valued $q$-cochain $\hat{a}$ is a cochain state $\ket{c^q}$ in the diagonal basis:
\be
\hat{a}^q \ket{c^q} = c^q \ket{c^q},
\ee
where $c^q \in C^q(\L; \ZZ_2)$ (classical variable) just stores the eigenvalues of $\hat{a}^q$. We can also equivalently express the operator-valued cochain as:
\be
\hat{a}^q= c^q \ketbra{c^q}. 
\ee
\\
Again, for simplicity, we will omit the hat $\hat{\cdot}$ on $a^q$.

Now we investigate the corresponding cohomology operation.  As an example, we consider the triple cup product operation that will give rise to logical non-Clifford gate.   Since it is a $q$-form gauge theory, one can define $q-2$ types of CSS codes on the simplicial complex $\L$ or manifold. We consider the triple cup product operation on three copies of CSS codes (not necessarily identical copies)  with qubits putting on the $q_1$-, $q_2$- and $q_3$-simplices respectively satisfying $q_1+q_2+q_3=n$, with the associated electric gauge fields being $a^{q_1}$, $a^{q_2}$ and $a^{q_3}$ respectively.  The corresponding constant-depth circuit can be expressed as
\be\label{eq:cup_higher_form}
U = (-1)^{\int_{\Sigma_n} a^{q_1} \cup a^{q_2} \cup a^{q_3}},
\ee

We now present the following lemma:
\begin{lemma}\label{lemma_gate_3}
The unitary $U = (-1)^{\int_{\Sigma_n} a^{q_1} \cup a^{q_2} \cup a^{q_3}}$ (with $q_1+q_2+q_3=n$) acting on three copies of CSS code defined on an $n$D simplicial complex is a constant-depth local quantum circuit that implements the collective logical CCZ gates.
\end{lemma}
\begin{proof}
The unitary $U$ in Eq.~\eqref{eq:cup_higher_form} can be evaluated according to the general definition of cup products in Eq.~\eqref{eq:cup_def} as below:
\begin{widetext}
\begin{align}\label{eq:logical_gate_higher_form}
\non U=& (-1)^{\int_{[v_0,v_1, \cdots, v_n] \in \Sigma_n} a([v_0,v_1, \cdots v_{q_1}])a'([v_{q_1}, \cdots, v_{q_1+q_2}])a''([v_{q_1+q_2}, \cdots, v_n])} \\
=&\prod_{[v_0,v_1, \cdots, v_n] \in \Sigma_n} \text{CCZ}([v_0,v_1, \cdots v_{q_1}],[v_{q_1}, \cdots, v_{q_1+q_2}], [v_{q_1+q_2}, \cdots, v_n]).
\end{align}
\end{widetext}
Therefore, $U$ is a constant-depth circuit composed of the product of physical CCZ gates.

The proof that shows $U$ is indeed a logical gate which preserves the codes space follows trivially from the proof in Lemma 
\ref{lemma_gate_2}.

Now we can re-express the circuit $U$ using the cohomology basis $\{\alpha^{q_1}\}$, $\{\beta^{q_2}\}$, and $\{\gamma^{q_3}\}$ to derive the corresponding logical gate:
\begin{align}\label{eq:logical_CCZ_higher}
 U =& \prod_{\alpha^{q_1}, \beta^{q_2}, \gamma^{q_3}} (-1)^{\int_{\Sigma_n} (\hat{n}_\alpha  \alpha^{q_1} )\cup (\hat{m}_\beta  \beta^{q_2}) \cup (\hat{l}_\gamma  \gamma^{q_3})} \cr
   =&\prod_{\alpha^{q_1}, \beta^{q_2}, \gamma^{q_3}}\left[(-1)^{ \hat{n}_\alpha \hat{m}_\beta \hat{l}_\gamma}\right]^{\int_{\Sigma_n} \alpha^{q_1} \cup \beta^{q_2} \cup \gamma^{q_3}} \cr
   =&\prod_{\alpha^{q_1}, \beta^{q_2}, \gamma^{q_3}} \overline{\text{CCZ}}[(\alpha^{q_1}; 1), (\beta^{q_2}; 2),(\gamma^{q_3}; 3)]^{\int_{\Sigma_n} {\alpha^{q_1} \cup \beta^{q_2} \cup \gamma^{q_3}}}. \cr
\end{align}
In the case that the $n$D simplicial complex $\L$ is the triangulation of a manifold $\M^n$, we can just take $\Sigma_n=\L$.  The cup product sum in the exponent now corresponds to an $n$-fold $\ZZ_2$ intersection number of the Poincar\'e dual cycles: 
\begin{align}
& \int_{\L} {\alpha^{q_1} \cup \beta^{q_2} \cup \gamma^{q_3}} =|\alpha^*_{n-q_1}  \cap \beta^*_{n-q_2} \cap  \gamma^*_{n-q_3}|.
\end{align}

\end{proof}
\nin We can easily generalize the triple cup product sum above to a $k$-fold cup product sum. 

All our qLPDC code constructions in this paper utilize the cohomology operations in such higher gauge theories to implement logical non-Clifford gates.

\section{Three-dimensional hypergraph product codes }\label{sec:hypergraph-product-code}

In order to achieve a quantum code with constant rate and power-law distance which supports logical non-Clifford gates, 
we construct a 3D hypergraph product code \cite{Tillich:2014_hyergraph_product, Quintavalle:2020_3D} obtained from the homological product \cite{Bravyi:2014bq} of three good identical classical codes. 

The good 1-cycle code $\mathcal{C}$ with parameters $[n,k,d]$ corresponds to the following chain complex $\mathcal{L}$:
\be\label{eq:classical_chain}
C_1 \xrightarrow[]{\partial_{1}=\mathsf{H}} C_{0},
\ee
where $\mathsf{H}$ is the parity check matrix of the classical code.  The bit and check variables are associated with $C_1$ and $C_0$ respectively. The length of the code (the number of bits) is denoted by $n=\text{dim}(V)$, where $V$ is the vector space spanned by all possible bit configurations represented by $n$-dimensional binary vectors. The code space is the subspace $\mathcal{C}=\text{Ker} \ \mathsf{{H}}=\text{Ker} \ \partial_1 = H_1(\L;\ZZ_2) \subset V$, where $H_1$ represents the 1st $\ZZ_2$-homology group. The code length (number of bits) is $n$. We have the dimension of the classical code equaling to the 1st $\ZZ_2$-Betti number $k=\text{dim}(\text{Ker}(\mathsf{H}))=\text{dim}(H_1(\L; \ZZ_2))=b_1$. For good classical code, the dimension is linear, i.e., $k=b_1=\Theta(n)$, while the distance equals the combinatorial $\ZZ_2$ 1-systole is also linear, i.e., 
\be
d=sys_1(\L;\ZZ_2)=\Omega(n), 
\ee
where we have used the following definition of the systole:
\begin{definition}\label{def:systole}
	We define the combinatorial $\ZZ_2$ $i$-systole $sys_i(\L;\ZZ_2)$ on the simplicial complex $\L$ as
\be
sys_i(\L; \ZZ_2) = \min\{|c_i|: c_i \neq 0 \in H_i(\L;\ZZ_2)\}
\ee
where the Hamming weight $|c_i|$ counts the number of $i$-simplices on which the $i$-cycle $c_i$ is non-zero. \footnote{Throughout this paper, the systole we discuss corresponds to the combinatorial systole including the manifold cases appeared later, where we consider the combinatorial systoles on the triangulation $\L$ of the manifold.}
\end{definition}

We also consider the transposed 0-cocycle code $\mathcal{C}^T=\text{Ker}(\mathsf{H}^T)=\text{Ker}(\delta^0)=\text{Ker}(H^0)$ with the parameters $[n^T,k^T,d^T]$, which is  described by the following chain complex
\be
C_1 \xleftarrow[]{\delta_{0}=\mathsf{H}^T} C_{0}.
\ee
The distance of the transposed code $\C^T$ equals the combinatorial $\ZZ_2$ 0-cosystole:
\be
d^T=sys^0(\L;\ZZ_2),
\ee
where we have used the following definition of the cosystole as:
\begin{definition}\label{def:cosystole}
We define the combinatorial $\ZZ_2$ $i$-cosystole $sys^i(\L;\ZZ_2)$ on the simplicial complex $\L$ as
\be
sys^i(\L; \ZZ_2) = \min\{|c^i|: c^i \neq 0 \in H^i(\L;\ZZ_2)\},
\ee
where the Hamming weight $|c^i|$ counts the number of $i$-simplices on which the $i$-cocycle $c_i$ is non-zero. 
\end{definition}

Now to optimize the performance of the logical CCZ gates, we want to make the transposed code $\mathcal{C}^T$ also a good code.  In order to achieve that, we start with a parity check matrix $\mathsf{H}$ of a good $[n,k,d]$ classical code with full rank. Such a code can be obtained from random construction of good LDPC code.  We then construct a new parity check matrix $\bar{\Hs}=\mathsf{H}^T \mathsf{H}$ and the associated code $\bar{\C}$ with the symmetric property that $\bar{\Hs}=\bar{\Hs}^T=\Hs^T \Hs$. Therefore, we have the corresponding code length $\bar{n}=\bar{n}^T =n$ due to the matrix multiplication rule and $\bar{k}=\bar{k}^T$. We can then obtain  the following useful lemma:  

\begin{lemma}\label{lemma:classical_code}
For a classical code $\bar{\mathcal{C}}$ with the parity check matrix $\bar{\Hs}=\Hs^T\Hs$, where $H$ is the parity check matrix of a good $[n,k,d]$ classical LDPC code with full rank, we have $\bar{\mathcal{C}}=\bar{\mathcal{C}}^T=\text{Ker}(\bar{\mathsf{H}})=\text{Ker}(\bar{\mathsf{H}}^T)=\text{Ker}(\mathsf{H})$.  Moreover, the code $\bar{\C}$ and its transposed code $\bar{\C}^T$ are  also good classical LDPC codes with linear dimension $(\bar{k}=\Theta(n))$ and linear distance $(\bar{d}=\Omega(n))$. 
\end{lemma}

\begin{proof}
Since $\mathsf{H}$ has full rank, meaning there is no redundant checks (the rank now equals the number of rows which is the number of checks), we hence have the following identity through the rank-nullity theorem:
\begin{align}\label{eq:full_rank}
\non \text{rank}(\mathsf{H})=& \text{dim}(\text{Img}(\mathsf{H}))=\text{dim}(V)-\text{dim}(\text{Ker}(\mathsf{H})) \\
=& n-k = n^T, 
\end{align}
where $n^T = \text{dim}(V^T)$ is the number of checks in $\C$ equaling the number of bits in the transposed code $\C^T$ and $V^T$ the vector space associated with the check variables in $\C$ or equivalently the bit variables in $\C^T$. We hence have
\begin{align}
\text{rank}(\mathsf{H}^T)=&\text{rank}(\mathsf{H})=\text{dim}(V)-\text{dim}(\text{Ker}(\mathsf{H})) \cr
=&n^T-k^T=n^T,
\end{align}
where $k^T$ is the dimension of the transposed code and we have $k^T=0$.
This leads to $\text{Ker}(\mathsf{H}^T)=0$.  We thus have 
\be
\text{Ker}(\mathsf{H}^T \mathsf{H})=\text{Ker}(\mathsf{H}),
\ee
which leads to
\be
\text{Ker}(\bar{\mathsf{H}})=\text{Ker}(\bar{\mathsf{H}}^T)=\text{Ker}(\mathsf{H}).
\ee

This means both the 1-cycle code $\bar{\mathcal{C}}=\text{Ker}(\bar{\mathsf{H}})=H_1$ and the transposed 0-cocycle code $\bar{\mathcal{C}}^T= \text{Ker}(\bar{\mathsf{H}}^T)=H^0$ corresponds to the same good classical code with both linear dimension and linear distance.  In particular, we now have a linear scaling for both the $0^\text{th}$ and $1^\text{st}$ Betti numbers corresponding to the dimensions of the of $\bar{\C}$ and $\bar{\C}^T$ respectively: 
\begin{align}
&\bar{k}=b_1=\text{rank}(H_1(\L;\ZZ_2))=\Theta(n), \cr
&\bar{k}^T=b_0=\text{rank}(H_0(\L;\ZZ_2))=\text{rank}(H^0(\L;\ZZ_2))=\Theta(n), \cr
&
\end{align}
where we have used the isomorphism between the $i^\text{th}$ $\ZZ_2$ homology and cohomology groups $H_i(\L; \ZZ_2)$$\cong$$ H^i(\L; \ZZ_2)$ due to the universal coefficient theorem \cite{Hatcher:2001ut}.
Since $\bar{\C}=\bar{\C}^T= \C$, we have 
\be
\bar{d} = \bar{d}^T=d= \Omega(n).
\ee

Although the code space of $\bar{\C}$ and $\C$ is identical, their corresponding tanner graph is still different. Nevertheless, since both $\Hs^T$ and $\Hs$ are sparse matrices ($\C$ is a classical LDPC code), their multiplication $\Hs^T H=\bar{\Hs}$ is also a sparse matrix which means the code $\bar{\C}=\bar{\C}^T$ is also an LDPC code. 
\end{proof}

We note that here we focus on the case that the good classical code $\C=\text{Ker}(\Hs)$ with a full-rank parity check matrix is obtained from a random bipartite expander graph.   One can have the same construction using the Sipser-Spielman code \cite{Sipser_Spielman}, while the proof for the properties of $\bar{\Hs}=\Hs^T \Hs$ will be a bit different.  We will leave this discussion for future works or updated version of this paper.

We now take a homological product of three identical copies of code $\mathcal{\bar{C}}=\text{Ker}(\bar{\mathsf{H}})= \text{Ker}(\mathsf{H}^T\mathsf{H})$ to form a 3D hypergraph product code with parameters $[N, K, D]$, where $N=3n^3$. The total chain  complex of the 3D hypergraph product code has the product form   $\tilde{X}= X \otimes X' \otimes  X''$ and corresponds to the chain complex
\be
\tilde{C}_{3} \xrightarrow[]{\partial_3} \tilde{C}_{2} \xrightarrow[]{\partial_{2}=\mathsf{H}_Z^T} \tilde{C}_1 \xrightarrow[]{\partial_{1}=\mathsf{H}_X} \tilde{C}_{0}.
\ee
where $\tilde{C}_i$ denotes the $i$-chain group of the total complex $\tilde{X}$.
In particular, we have the following relation to the chain groups of the three classical codes:
\be
\tilde{C}_i =\bigoplus_{p+q+t=i} C_p \otimes C_q \otimes C_t.
\ee
We also express it explicitly as 
\begin{align}\label{eq:tensor_decomposition}
\non \tilde{C}_0=&C_0 \otimes C'_0 \otimes C''_0, \\ \non \tilde{C}_1=&C_1 \otimes C'_0 \otimes C''_0 + C_0 \otimes C'_1 \otimes C''_0 + C_0 \otimes C'_0 \otimes C''_1, \\
\non \tilde{C}_2=&(C_1 \otimes C'_1 \otimes C''_0) \oplus (C_0 \otimes C'_1 \otimes C''_1) \oplus (C_1 \otimes C'_0 \otimes C''_1) \\
\non \oplus & (C_2 \otimes C'_0 \otimes C''_0) \oplus (C_0 \otimes C'_2 \otimes C''_0) \oplus (C_0 \otimes C'_0 \otimes C''_2), \\  
\tilde{C}_3 =& C_1 \otimes C_1 \otimes C_1,
\end{align}
which tell us how to compose the vertices, edges, faces, and cubes with vertices and edges in the input classical code. 
For the new code, we can use the Künneth theorem
\begin{align}\label{eq:Kunneth_general}
\tilde{H}_i =& \bigoplus_{p+q+t=i} H_p \otimes H_q \otimes H_t \cr
\tilde{H}^i =& \bigoplus_{p+q+t=i} H^p \otimes H^q \otimes H^t
\end{align}
which gives rise to the corresponding 1st homology and cohomology, i.e.,
\begin{align}\label{eq:Kunneth}
\tilde{H}_1=& (H_1 \otimes H'_0 \otimes H''_0) \oplus (H_0 \otimes H'_1 \otimes H''_0) \cr 
&\oplus (H_0 \otimes H'_0 \otimes H''_1) \cr
\tilde{H}^1=& (H^1 \otimes H'^0 \otimes H''^0) \oplus (H^0 \otimes H'^1 \otimes H''^0) \cr
& \oplus (H^0 \otimes H'^0 \otimes H''^1).
\end{align}

One can quickly verify that the dimension, which equals the first Betti number, is still linear:
\begin{align}
K=&\tilde{b}_1 = b_1\cdot b'_0\cdot b''_0 + b_0\cdot b'_1 \cdot b''_0 + b_0 \cdot b'_0 \cdot b''_1 \cr
=& \Theta(n)\cdot \Theta(n) \cdot \Theta(n)= \Theta(N),
\end{align}
where $K$ and $N$ represent the total number of logical and physical qubits respectively.
The $Z$-distance is the combinatorial 1-systole 
\be
D_Z = sys_1(\mathcal{\tilde{L}}; \ZZ_2)= \Omega (n) = \Omega(N^{1/3}),
\ee
and the $X$-distance is the combinatorial 1-cosystole
\be
D_X = sys^1(\mathcal{\tilde{L}}; \ZZ_2)= \Omega (n^2) = \Omega(N^{2/3}),
\ee
where we have used the Künneth theorem in Eq.~\eqref{eq:tensor_decomposition}.
The overall distance is hence
\be
D= \text{min}(D_Z, D_X)=\Omega(N^{1/3}).
\ee

\section{Non-Clifford logical gates on thickened 3D hypergraph-product codes}\label{sec:non-Clifford_hypergraph}

\subsection{Building manifolds from classical codes}\label{sec:handle_construction}

\subsubsection{General introduction}

We generalize the code-to-manifold mapping in  Ref.~\cite{freedman:2020_manifold_from_code} from the case of quantum codes to the case of classical codes.  The manifolds built from the classical or quantum codes will then serve as the Legos of the product construction of the thickened qLDPC codes.

 Roughly speaking,  the construction of the manifold can be considered via a `\textit{plumber's view}' as  thickening the tanner graph of a skeleton classical code 
along extra-dimensions with the thickness being only $O(1)$ as illustrated by Fig.~\ref{fig:dictionary}, which hence only leads to a constant overhead of the thickened classical code defined on the manifold. 

A general classical expander code corresponds to a 2-term $\ZZ_2$-chain complex in Eq.~\eqref{eq:classical_chain} can be defined on a bipartite Tanner graph  or equivalently a  hypergraph. 
 When viewed as a Tanner graph $G_T=(V_B, V_C, E)$, the bits and checks are represented by the two types of vertices $V_B$ and $V_C$ represented by circles and squares  respectively, while the edge $E$ connecting them encodes the boundary map,  as illustrated in Fig.~\ref{fig:dictionary}(a,c,e).  When viewed as a hypergraph $G_h=(V, E_h)$, the checks are on the vertices $V$ while the bits are associated with a hyperedge $E_h$ composed by the edges connecting to a circle.   
In order to construct a manifold from this classical code, we first need to lift the $\ZZ_2$-chain complex to a $\Z$-chain complex. 
Here is the definition:
\begin{definition}\label{def:lift}\cite{freedman:2020_manifold_from_code} $\Z$-lift: 
Given a chain complex over $\ZZ_2$, defined by boundary operators $\partial_j: C_j \rightarrow C_{j-1}$ for some sequence of $C_j$, where the boundary operators are regarded as finite matrices with entries 0, 1, a \textit{lift} of that chain complex is a chain complex over $\ZZ$ defined by a sequence of lifted boundary operators $\hat{\partial}_j$, such for all $j$, $\hat{\partial}_j$ is a lift of ${\partial}_j$.
\end{definition}

 For classical code, one can just use a ``naive" lift, which maps $0$ mod $2$ to $0 \in \ZZ$ and $1$ mod $2$ to $1 \in \ZZ$  and obtain the following lifted chain complex $\hat{\mathcal{X}}$:
\be\label{eq:chain_lift}
\hat{C}_1 \xrightarrow[]{\hat{\partial}=\hat{\bar{\mathsf{H}}}} \hat{C}_{0},
\ee
where $\hat{\partial}$ and $\hat{\bar{\mathsf{H}}}$ are the lifted boundary map  and parity check matrices with $\ZZ$ coefficients respectively obtained from the $\ZZ_2$ parity check matrix $\bar{\mathsf{H}}$.

\begin{figure*}[hbt]
\includegraphics[width=1\textwidth]{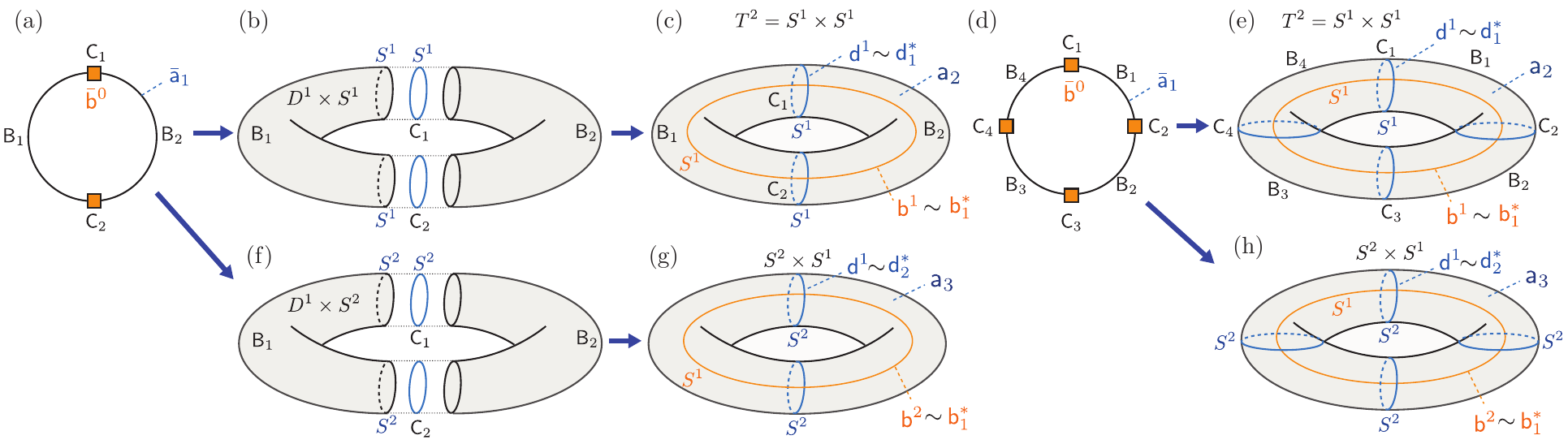}
\caption{(a) A toy example of a repetition code $\bar{\C}$ with two bits and two checks. One of the checks is redundant. The codeword of both $\bar{\C}$ and $\bar{\C}^T$ are highlighted, labeled as $\bar{\as}_1$ and $\bar{\bs}^0$ respectively. (b) The check is mapped to a ``1-cell" $S^1$ and the bit is mapped to a ``2-cell" $D^1 \times S^1$. (c) When  attaching the ``2-cells" to the ``1-cell" one obtains a torus. The codeword of $\bar{\C}$ is mapped to $\as_2$ while the codeword of  $\bar{\C}^T$ is mapped to $\bs^1 \sim \bs^*_1$. There exist a spurious 1-cocycle $\ds^1 \sim \ds^*_1$.  (d, e) An illustration of a general repetition code defined on a circle mapped to a torus. The spurious 1-cocycle $\ds^1 \sim \ds^*_1$ has only $O(1)$ size and will hence decrease the distance of the transposed code $\C^T$. (f) Map check to a ``2-cell" $S^2$, and bit to a ``3-cell" $D^1 \times S^2$.  (g) Attaching the cell together give rise to $S^2 \times S^1$.  The codeword of $\bar{\C}$ and  $\bar{\C}^T$ are mapped to 3-cycle $\as_3$ and 2-cocycles $\bs^2$ respectively, which have separated dimensions with the spurious 1-cocycle $\ds^1$. (h)  Mapping for the general repetition code. The spurious 1-cocycle $\ds^1$ does not affect the distance of the transposed code $\C^T$. }\label{fig:spurious_illustration}
\end{figure*}

As has been observed in Ref.~\cite{freedman:2020_manifold_from_code}, one cannot directly obtain the manifold and the underlying simplicial or cellular  complex structure from the above lifted chain complex.  The issue is that on a manifold or more generally a simplicial complex, and edge (1-cell) is always adjacent to two vertices (0-cells).  This is not true for a general two-term chain complex corresponding to a hypergraph where a hyperedge can be adjacent to more than two vertices. 

In order to resolve this issue, we need to promote the bits from 1-cells (edges) to higher-dimensional cells. An obvious attempt is to promote the bits to 2-cells (faces), and checks to 1-cells.  Note that,  topologically speaking, a $k$-cell corresponds to a $k$-dimensional ball $D^k$. We will sometimes also put the quotation mark to the word as ``$k$-cell" to refer to an object that is not a $k$-ball following the convention in Ref.~\cite{freedman:2020_manifold_from_code}. It should be considered as a dressed cell that is itself composed of cells of dimension equal and lower than $k$. We then attach the ``2-cells" to ``1-cells" according to the lifted boundary map $\hat{\partial}$ in order to build a cellular complex, which can be thickened to a manifold with a boundary, and attached to a copy of itself to reach a closed manifold.  The ``1-cell" we associate to each check $\mathsf{C}_j$ is formed by attaching a 1-cell $D^1$ (interval) to a 0-cell $D^0$ (point), which becomes a circle $S^1$.  The ``2-cell" associated to the bit $\mathsf{B}_i$ is formed by removing $f(i)$ 2-disks on a 2-sphere, i.e., $S^2\backslash\sqcup_{m=1}^{f(i)}D^2_m$, where $f(i)$ specifies the total number of checks that the bit $\mathsf{B}_i$ is connected to given by the lifted boundary map $\hat{\partial}$, and $\sqcup$ represents a disjoint union.

However, such a construction may lead to spurious homologies which do not correspond to the homology group in the input skeleton classical code, as has been pointed out by Ref.~\cite{freedman:2020_manifold_from_code}.   We consider the following toy example which is also mentioned in Ref.~\cite{freedman:2020_manifold_from_code}.   We consider the following lifted boundary map $\hat{\alpha}=\hat{\bar{\Hs}}$  obtained from a classical code according to Eq.~\eqref{eq:chain_lift} with two bits in $\hat{C}_1$ and two checks in $\hat{C}_0$:
\be
\hat{\bar{\Hs}}= 
\begin{pmatrix}
     1 & 1 \\
     1 & 1
   \end{pmatrix}.
\ee
The associated tanner graph is shown in Fig.~\ref{fig:spurious_illustration}(a). Note that both checks $\mathsf{C}_1$ and $\mathsf{C}_2$ check the same set of bits $\mathsf{B}_1$ and $\mathsf{B}_2$, so one of the checks is redundant, and this redundancy gives rise to 0th-Betti number $\bar{b}_0=1$. On the other hand the 1st-Betti number is $\bar{b}_1=1$ since the code only encode 1 logical bit. Now applying the above scheme of shifting the dimensions, we promote each check originally associated with $\hat{C}_0$ to a ``1-cell" $S^1$, as illustrated in Fig.~\ref{fig:spurious_illustration}(b). We then promote each bit originally associated with $\hat{C}_1$ to a ``2-cell", which is formed by a 2-sphere with two 2-disks being removed: $S^2 \backslash (D^2 \sqcup D^2)$, equivalent to a cylinder $D^1 \times S^1$. We now attach each ``2-cells" to both ``1-cells" by gluing them along the $S^1$, this gives rise to a torus $T^2=S^1 \times S^1$ as illustrated in Fig.~\ref{fig:spurious_illustration}(c).   The second betti number is hence $b_2=b_0=1$ corresponds to the 1 logical bit, since the torus has only 1 connected component. There is hence no issue for the codeword 1-cycle $\bar{\as}_1$ in the skeleton classical code $\bar{\C}$, which is mapped to the 2-cycle $\as_2$ on the torus, 
as illustrated in Fig.~\ref{fig:spurious_illustration}(b). Now the issue is that the 1st Betti number is $b_1=2$ coming from both the logitudinal (orange) and meridian (blue) 1-cycles $S^1$. However, the 0th Betti number corresponding to the redundancy of the check in the input classical code $\bar{\C}$ is only $\bar{b}_0=1$, so we have got a spurious 1st-homology in the manifold we construct.  

We can see that this spurious homology is also present in the more general case of a repetition code defined on a circle as illustrated in Fig.~\ref{fig:spurious_illustration}(d).  The main problem is that the codeword of the transposed skeleton code $\bar{\C}^T=\text{Ker}(\bar{\Hs}^T)$, namely the 1-cocycle $\bar{\bs}^0$ is mapped to the 1-cocycle $\bar{\bs}^1 \in H^1(T^2; \ZZ_2)$ on the torus.  
However, there is a spurious 1-cocycle $\ds^1$ with $O(1)$ size (in both the discrete and continuous sense), whose Poincar\'e dual ${\bar{\bs}}^*_1$ is the meridian (blue) 1-cycle $S^1$ (blue), which gives rise to an $O(1)$ 1-cosystole $sys^1(T^2; \ZZ_2)$  on the torus.   This leads to the consequence that the transposed classical code on the torus $\C^T=H^1(T^2; \ZZ_2)$ has an O(1) distance if we consider the code as a subspace code.  When taking the product of the manifolds built from the skeleton classical codes, the spurious cycles/cocycles may lead to small systole/cocystoles at the dimension of the logical cycles/cocycles, which then leads to short distance.    We note that this is not a dead end, since one can use the subsystem-code idea to treat those short spurious cycles/cocycles as the logical operators of gauge qubits.   We will return to this approach in Sec.~\ref{sec:4-manifold} to Sec.~\ref{sec:magic_rate_hypergraph}.

Instead, we will first go to higher dimensions and associate bits to ``3-cells", and checks to ``2-cells", which will lead to a clear separation between the dimensions of the logical (co)cyclces having large (co)systoles with the spurious (co)cycles having $O(1)$ (co)systoles, following the approach in Ref.~\cite{freedman:2020_manifold_from_code}.  

As we can see for the toy example of the classical skeleton code in Fig.~\ref{fig:spurious_illustration}(a), $\mathsf{C}_j$ $\mathsf{B}_i$ are mapped to ``2-cell" $S^2$ and ``3-cell" $D^1 \times S^2$ respectively as shown in Fig.~\ref{fig:spurious_illustration}(b). Now when attaching the ``3-cells" to the ``2-cells" and realize a 3-manifold $S^2 \times S^1$.  We can see now the second and third Betti numbers are $b_2=b_3=1$,  which faithfully corresponding to the Betti numbers of the skeleton classical code $\bar{b}_0 =\bar{b}_1=1$ with a shift in dimension.  The spurious homology still lies in dimension 1, which gives rises to $b_1=1$ orignated from the longitudinal (orange) 1-cycle $S^1$.   All these properties remain the same in the more general case of a repetition code on a circle as illustrated in Fig.~\ref{fig:spurious_illustration}(d, h).
The codeword of the skeleton classical code $\bar{\C}$, i.e., 1-cycle $\bar{\as}_1$ is mapped to 3-cycle $\as_3$ on the 3-manifold, which in this case is just the entire 3-manifold $S^2 \times S^1$. On the otherhand, the codeword of the transposed skeleton code $\bar{\C}^T$, i.e., 0-cocycle $\bar{\bs}^0$ is mapped to the 2-cocycle $
\as^2$ whose Poincar\'e dual 1-cycle $\as^*_1$ is supported on the longitudinal (orange) cycle $S^1$ as illustrated in Fig.~\ref{fig:spurious_illustration}(h).  
Therefore, both the 3-systole $sys_3(S^2 \times S^1; \ZZ_2)$ and 2-cosystole $sys^2(S^2 \times S^1; \ZZ_2)$, which correspond to the distances of the classical code $\C=H_3(S^2 \times S^1; \ZZ_2)$ and the transposed code $\C^T=H^2(S^2 \times S^1; \Z_2)$, are large and have the same scaling $\Theta(n)$ as the skeleton classical code $\bar{\C}$, where $n$ is the total number of bits in $\bar{\C}$. Note that the 2-systole $sys_2(S^2 \times S^1; \ZZ_2)$ has only $O(1)$ size due to the short meridian (blue) 2-cycle $S^2$. However, the 2-cycle does not corresponds to the codeword of $\C$ or $\C^T$. We hence obtain a classical code $\C$ along with its transposed code  $\C^T$ on the 3-manifold  which have the same code parameter scaling as the skeleton classical code $\bar{\C}$ and its transpose $\bar{\C}^T$. Also note that there exists a short $O(1)$-size spurious 1-cocycle $\ds^1$, whose Poincar\'e dual 2-cycle $\ds^*_2$ is supported on the $S^2$ (blue) of size $O(1)$. Nonetheless, its dimension 1 is separated from dimension 3 and 2 of the logical cycles and cocycles of $\C$ and $\C^T$ and hence does not cause any issue for the classical code.

Up to now, the above 3-manifold construction works perfectly for the repetition code. However, issues arise for more general classical codes where the check $\mathsf{C}_i$ can be connected to more than two bits $\mathsf{B}_i$. The problem is that in a 3-manifold, a ``2-cell" $S^2$ cannot be attached by more than two ``3-cells" $D^1 \times S^2$, which is related to the fact that in the corresponding triangulation a 2-simplex cannot be attached by more than two 3-simplices.  To avoid this problem, one needs to go to  higher dimensions as will be described in the following subsection.

\subsubsection{Handle construction of the 8-manifold}\label{sec:handle_8-manifold}

In the following, instead of associating the bits and checks in the skeleton classical code to the ``$k$-cells" and ``$(k-1)$-cells" and build a cellular complex which can then be thickened into a manifold as mentioned previously, we pursue an alternative approach via a handle construction of the manifold by associating the bits and checks to $k$-handles and $(k-1)$-handles, which is conceptually much easier.  After building the manifold $\M$, we will obtain a new thickened classical LDPC code defined on the triangulation $\L$ of the manifold $\M$ with the bits and checks defined on the $k$-simplices and $(k-1)$-simplices, which can then be used as the building blocks (along with the triangulated manifolds built from the quantum code shown in  Sec.~\ref{sec:good_product}) to construct thickened qLDPC codes by taking the product of the corresponding manifold triangulations.  The qubits and checks in these qLDPC codes will be placed on the simplices on the corresponding triangulations. This will be the approach in the entire Sec.~\ref{sec:non-Clifford_hypergraph} and Sec.~\ref{sec:good_product}.

Although the classical or quantum codes built on the triangulation will on have a constant overhead compared to the input skeleton codes,  later in Sec.~\ref{sec:deformation_retraction} we will choose an even more compact construction by deformation retracting the manifolds back to a cellular complex, more commonly referred to as a \textit{CW complex} in the literature \cite{Hatcher:2001ut}. This is achieved by retracting each $i$-handle in the manifold into a $i$-cell in the CW complex $\L_c$.  In this way, the bits and checks in the skelton classical code $\bar{\C}$ are directly associated to the $k$-cells and $(k-1)$-cells instead of the simplices, and the new classical code  defined on the CW complex $\L_c$ is completely the same as the input skeleton code $\bar{\C}$, and the same for the quantum code constructed in Sec.~\ref{sec:good_product}.  In this way, the constructed product quantum codes will have a minimum overhead compared to the product of the skeleton codes.

Back to the handle construction, an $r$-dimensional $k$-handle is a pair
\begin{equation}
  h_k=(D^k \times D^{r-k}, \partial D^k \times D^{r-k}),  
\end{equation}
which represents a $r$-dimensional manifold $D^k \times D^{r-k}$ along with its attaching region $\partial D^k \times D^{r-k} = S^{k-1} \times D^{r-k}$, where $D^k$ and $S^{k-1}$ represent a $k$-dimensional ball and a ($k-1)$-dimensional sphere respectively.  We call $D^k$ and $D^{r-k}$ the \textit{core} and \textit{co-core}  of a $k$-handle respectively, as illustrated in Fig.~\ref{fig:handle_introduction}(a). We also define the co-attaching region as $D^k \times \partial D^{r-k}$. Roughly speaking, a $k$-handle can be considered as a thickened $k$-cell. Sometimes we will just write the $k$-handle as $h_k=D^k \times D^{r-k}$ with the attaching region not expressed explicitly for conciseness.   Handlebodies are unions of a sequence of handles where each handle is attached (glued) along its attaching region to the previous union of handles.   We can build the handlebodies starting with $0$-handles, which are $r$-balls ($r$-disks) with no attaching region, i.e., $h_0=(0 \times D^{r}, \varnothing)$. One then step by step attaches handles of increasing indices $k$ to the previously constructed handlebodies until attaching the $r$-handles $h_{r}=(D^{r} \times 0, S^{r-1} \times 0)$ which has the empty co-attaching region to close the $r$-manifold.

\begin{figure}[t]
    \centering
    \includegraphics[width=1\linewidth]{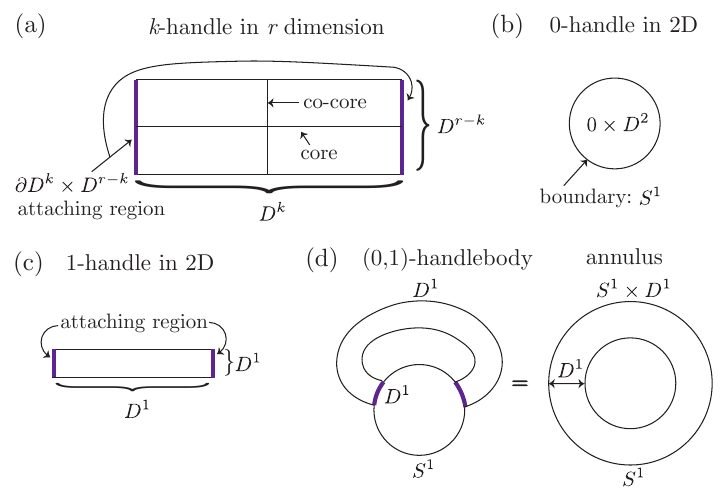}
    \caption{(a) Anatomy of the $k$-handle in $r$ dimension. (b) A 0-handle in 2D $0 \times D^2$ with no attaching region. The boundary is $S^1$. (c) The 1-handle $D^1 \times D^1$ in 2D with attaching region $D^1$ (purple).  (d) Attach the 1-handle to the boundary of the 0-handle along the attaching region $D^1$ (purple).  This gives rise to a (0,1)-handlebody homeomorphic to an  annulus $S^1 \times D^1$.  }
    \label{fig:handle_introduction}
\end{figure}

\begin{figure*}[hbt]
\includegraphics[width=1\textwidth]{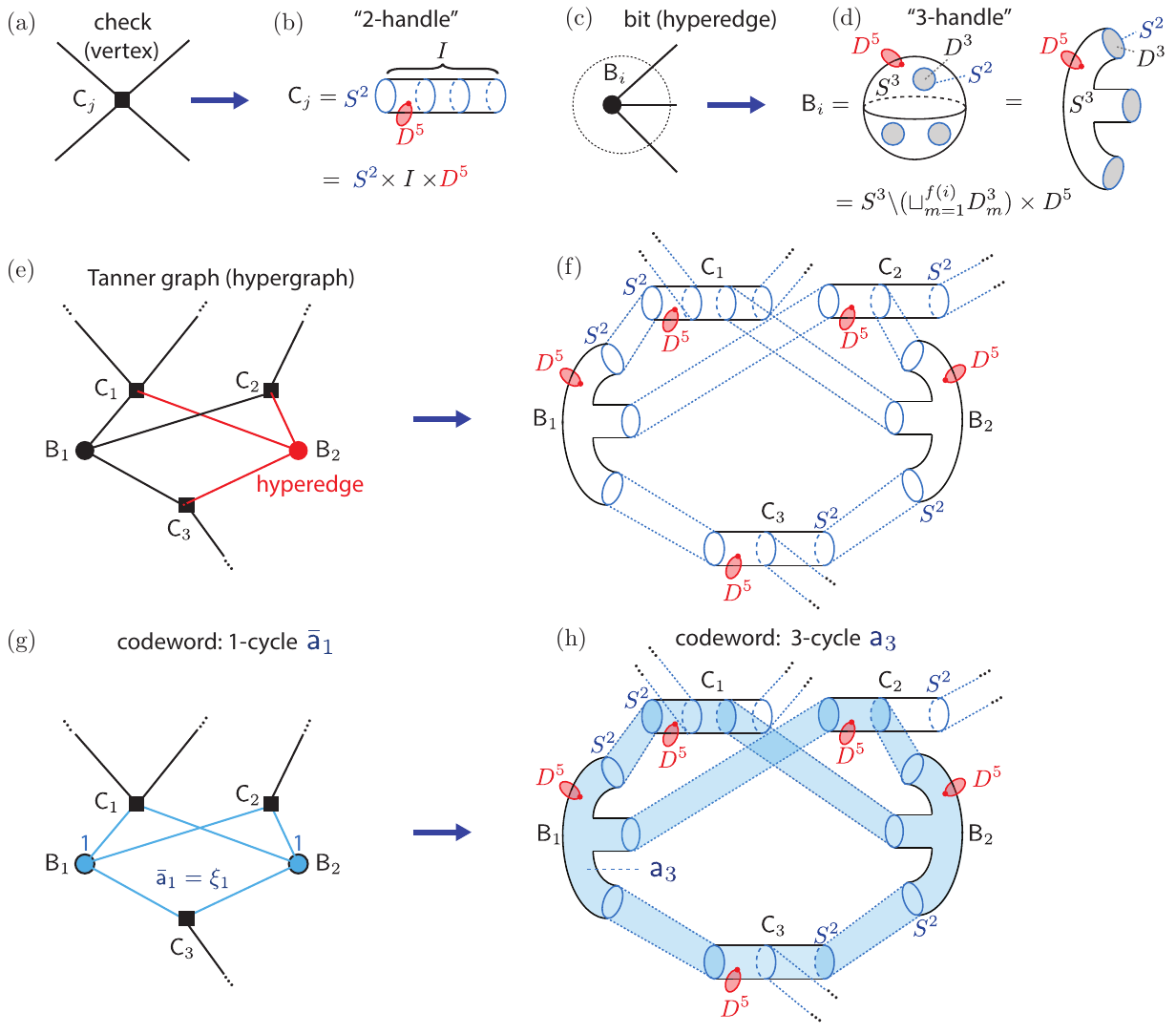}
\caption{A plumber's view of classical codes. (a) A check $\mathsf{C}_j$ on the vertex (square).  (b) The check $\mathsf{C}_j$ is mapped to a ``2-handle" $\mathsf{C}_j=S^2 \times I \times D^5$.  (c) A bit $\mathsf{B}_i$ on a hyperedge composed of multiple edges connected to a single circle.  (d) The bit  $\mathsf{B}_i$ is mapped to a ``3-handle" corresponding to a 3-sphere $S^3$ with multiple 3-disks (3-balls) $D^3$ being removed and then thickened by $D^5$. It has multiple legs which can be connected to the neighboring ``2-handles".  (e) A bipartite Tanner graph with two types of vertices corresponding to checks (squares) and bits (circles) respectively.  Equivalently one can consider it as a hypergraph where the checks are placed on the vertices (square) and the bits placed on the hyperedge (red). (f) The Tanner graph is thickened to a handlebody where the ``3-handles" are attached to the adjacent ``2-handles" according to the boundary map in the skeleton Tanner graph and via gluing them along the attaching region $S^2 \times D^5$ of the ``3-handles".  (g) A codeword of the classical code $\bar{\C}$ is a 1-cycle $\bar{\as}_1$ supported on a subset of hyperedges (blue). (h) The codeword is thickened to a 3-cycle $\as_3$ (blue) in the 3-handlebody $H$.  }\label{fig:dictionary}
\end{figure*}

Here we will construct the following handle  chain complex $\L_h$ corresponding to the handle decomposition of an 8-dimensional manifold  $\M^{8}$ \footnote{We note that $r=7$ should also work. We choose $r \ge 8$ since it is conceptually simpler:  the left and right portions are separated by trivial group $0$ and trivial boundary maps in the middle.}:
\begin{align}\label{eq:handle_chain_complex}
&C_{8} \rightarrow \cdots \rightarrow C_{4} \rightarrow C_3 \xrightarrow[]{\hat{\partial} \sim \bar{\mathsf{H}}} C_{2} \rightarrow C_1 \rightarrow \cdots , \cr
&\qquad   \quad \ \ \  \qquad \qquad    
 \text{bit} \quad \quad  \text{check}  \cr 
\end{align}
where $C_i=\text{span}_\ZZ(i\text{-}\text{handles})$. The boundary map $\partial_4$ instructing how the 4-handles are attached to the 3-handles can be obtained from the information in the lifted parity check matrix $\hat{\partial}=\bar{\mathsf{H}}$ (omit the hat for simplicity).  This means  the portion of the above handle chain complex $C_3 \xrightarrow[]{\hat{\partial} \sim \bar{\mathsf{H}}} C_{2}$ is isomorphic to the lifted chain complex Eq.~\eqref{eq:chain_lift}    obtained from the classical code. 
Here, the 4-chain group $C_4$ is trivial (meaning there is no 4-handles in this decomposition), and we can re-express the above handle chain complex $\L_h$ as
\begin{align}\label{eq:long_chain}
& C_8 \rightarrow C_7 \rightarrow C_{6} \xrightarrow[]{\hat{\partial}^T \sim \bar{\mathsf{H}}^T}  C_5 \rightarrow 0  \rightarrow C_3 \xrightarrow[]{\hat{\partial} \sim \bar{\mathsf{H}}} C_{2} \rightarrow C_1 \rightarrow C_0 \cr
\end{align}
Note that there is a symmetry between $C_3$ and $C_5$, and between $C_2$ and $C_6$, which is a reflection of Poincar\'e duality of a manifold.  We can further write down the dual handle chain complex $\L^*_h$:  
\begin{align}\label{eq:long_chain_dual}
 \cdots \leftarrow C_1^* \leftarrow C_{2}^* \xleftarrow[]{\hat{\partial} \sim \bar{\mathsf{H}}}  C^*_3 \leftarrow 0  \leftarrow C^*_5 \xleftarrow[]{\hat{\partial}^T  \sim \bar{\mathsf{H}}^T} C^*_{6} & \leftarrow C^*_7 \cdots \cr 
\end{align}
Note that the dual $i$-handle in dimension $r$ corresponds to an $(r-i)$-handle. We can see that $\hat{\partial} =\bar{\Hs}$ instructs  the attachment between the dual 3-handles and the dual 2-handles, which are equivalent to the 5-handles and 6-handles respectively. 

As hinted in Ref.~\cite{freedman:2020_manifold_from_code} and further elaborated in Ref.~\cite{guemard2025lifting} (see also Ref.~\cite{guemard2025lifts}),  the above handle chain complex $\L_h$ can be turned into a cellular chain complex $\L_c$  via a deformation retraction, which retracts the $k$-handles to $k$-cells. 
The cellular chain complex is often called a \textit{CW complex} in the literature \cite{Hatcher:2001ut} with the following formal definition:
\begin{definition}
A CW complex is built inductively by attaching 
$k$-cells (copies of open 
$k$-balls $D^k$) via attaching maps from their boundaries 
$S^{k-1}$ into the $(k-1)$-skeleton with the following conditions:
\begin{enumerate}
\item Closure-finite (C): The closure of each cell meets only finitely many other cells.
\item	
Weak topology (W): A set is closed if and only if its intersection with each cell-closure is closed.
\end{enumerate}

\end{definition}
\nin Therefore, the handle chain complex and the CW (cellular) complex are isomorphic to each other: $\L_h \cong \L_c$, as will be used later in Sec.~\ref{sec:deformation_retraction}.  

Note that we can always shift left portion of $\L_h$ in Eq.~\eqref{eq:long_chain} towards the left by inserting more trivial group 0 in the middle chains.  This will increase the total space dimension to $r>8$.

We now introduce the detailed procedure of a handle construction of the smooth 8-manifold $\M^{8}$.

Following Ref.~\cite{freedman:2020_manifold_from_code}, we first build a 3-handlebody $H$, which corresponds to a 8-manifold with boundary that is the union of $k$-handles of indices $k=0,1,2, 3$.   This can be viewed as the right part of the handle chain complex in Eq.~\eqref{eq:long_chain} starting from $C_3$.  We then take another identical copy of the 3-handlebody $H^*$, which can be viewed as the upside-down 3-handlebody corresponding to the left part of the dual handle chain complex in Eq.~\eqref{eq:long_chain_dual}. We hence obtain the closed manifold $\M^{8}=\mathcal{D}H$ as the double of $H$ by gluing the two copies $H$ and $H^*$ along their boundary $\partial H$ with an identity map $id_{\partial H}$, i.e., 
\be
\mathcal{D}H=H \cup_{id_{\partial H}} H^*. 
\ee
In the second copy the dual $k$-handles are equivalent to the $(8-k)$ handles such that we obtain the full handle chain complex in Eq.~\eqref{eq:long_chain_dual}.   It is through this doubling process that one builds in the Poincar\'e  duality structure of the manifold.

Instead of attaching handles with standard order of increasing indexes,  we will follow the approach in Ref.~\cite{freedman:2020_manifold_from_code} and attach dressed  $k$-handles \footnote{Note that the naming convention of `dressed handle' was first introduced in Ref.~\cite{guemard2025lifting}.} that are themselves handlebodies consisting of handles of with indices equaling or lower than $k$ such as 0, 1 and 2.   In this way, some lower-indexes handles will be introduced after some higher-index handles, while the handles can be reordered by indexes through handle sliding. For simplicity, we will just denote the dressed $k$-handle as ``$k$-handle" following the convention in Ref.~\cite{freedman:2020_manifold_from_code}, where the use of quotation mark is similar to that for ``$k$-cell". The ``$k$-handle" can be expressed as
\be\label{eq:dressed-k-handle}
\tilde{h}_k = (N^k \times D^{r-k} , \partial N^k \times D^{r-k}),
\ee
with $r=8$ for this subsection. Here $N^k$ represents a dressed core.  In the following, we will attach ``$k$-handles" with the order of increasing indexes.  Note that each dressed ``2-handle" and ``3-handle" only contain exactly one bare 2-handle and 3-handle respectively.  Therefore the lifted boundary map $\hat{\partial} \sim \bar{\Hs}$, which instructs how the bare 3-handles are attached the bare 2-handles, also specifies how the dressed ``3-handles" are attached to the dressed ``2-handles".  Anatomy of the handle structure of the dressed handles will be discussed in details in Sec.~\ref{sec:retraction_detail}.

We first assign a 0-handle $h_0=0 \times D^{8}=D^{8}$ to each check in the classical code, and then attach a 2-handle $h_2=D^2~\times~D^6$ to the boundary of each 0-handle $h_0$ [the 2D analog of attaching a 1-handle to a 0-handle is illustrated in Fig.~\ref{fig:handle_introduction}(b-d)].  The attaching region of $h_2$ is $\partial D^2 \times D^6 = S^1 \times D^6$, and the attaching map in this case is
\be
S^1 \times D^6 \hookrightarrow  \partial D^{8}=S^7.
\ee
We hence get a $(0,2)$-handlebody $\sqcup_{j=1}^{n^T} (S^2 \times D^6)_j \equiv \sqcup_{j=1}^{n^T} \mathsf{C}_j$, where $\sqcup$ represents a disjoint union,  $\mathsf{C}_j$ represents the ``2-handle" corresponding to the $j^\text{th}$ check, and $n^T$ is the number of checks, with the dressed call being $N^2=S^2$. Note that this ``2-handle" is composed of one 0-handle and one 2-handle. We also call $\mathsf{C}_j$ the check ``handle", which can also be re-written as $\mathsf{C}_j = (S^2 \times I \times D^5)_j$ as illustrated in Fig.~\ref{fig:dictionary}(a,b), with $I$ being an interval. Note that when choosing the symmetric parity check matrix $\bar{\mathsf{H}}= \Hs^T \Hs $ according to Sec.~\ref{sec:hypergraph-product-code}, we have $n^T=n$ equal to the number of bits.

We then attach ``3-handles" $\mathsf{B}_i$ to each ``2-handles" $\mathsf{C}_j$,  where $\mathsf{B}_i$ refers to the ``3-handles" corresponding to the $i^\text{th}$ bit (also call them bit ``handles"). The attachment is determined by the lifted boundary map $\hat{\partial}$ from the classical code as
\be\label{eq:attaching_rule}
\hat{\partial}\mathsf{B}_i= \sum_{j \in I_i} \mathsf{C}_j,
\ee
where $I_i$ indexes all the $\mathsf{C}_j$ incident to $\mathsf{B}_i$.   Note that all the coefficient on the right of the above equation is $+1$ since we have chosen the naive lift that maps all $1 \mod 2$ to $1 \in \ZZ$.  This data is also encoded in the Tanner graph or equivalently the hypergraph of the skeleton classical code as illustrated in the correspondence between Fig.~\ref{fig:dictionary}(e) and (f).  
As shown in Fig.~\ref{fig:dictionary}(c,d), each ``3-handle" has the form $\mathsf{B}_i=N_i^3 \times D^5$ where the $D^3$ in the standard 3-handle is replaced by a punctured 3-sphere  $N^3_i=S^3 \backslash \sqcup_{m=1}^{f(i)}D^3_m$ as a dressed core. Here, $f(i)=|I_i|$ represents the number of check ``handles" $\mathsf{C}_j$ that the $i^\text{th}$ bit ``handle" is incident to given by the lifted boundary map $\hat{\partial}$.   The attaching region of a ``3-handle" $\mathsf{B}_i$ is $ \sqcup_{m=1}^{f(i)} (\partial  D^3 \times D^5)_{i, m} = \sqcup_{m=1}^{f(i)} (S^2 \times D^5)_{i,m}$, similar to the case of the standard 3-handle.  We then glue the attaching region of the ``3-handles" to the boundary of the ``2-handles" $\partial \mathsf{C}_j =(S^2 \times \partial D^6)_j =(S^2 \times S^5)_j$ [see Fig.~\ref{fig:dictionary}(f)].  The corresponding attaching map from $B_i$  to the adjacent $\mathsf{C}_j$ is hence:
\be
(S^2 \times D^5)_{i, m}  \hookrightarrow (S^2 \times S^5)_j, 
\ee
with the attaching rule specified by Eq.~\eqref{eq:attaching_rule}.  We hence have built the desired 3-handlebody $H$ from the classical code.  After taking the double of $H$, we obtain the closed 8-manifold $\M^{8}=\mathcal{D}H$ as mentioned above.  

Now we show how the logical information is mapped from the skeleton classical code to the thickened code on the manifold $\M^{8}$. A codeword of the skeleton classical code $\bar{\C}$ corresponds to the 1-cycle $\bar{\mathsf{a}}_1 =\xi_1$ on the hypergraph $G_h$, and is illustrated in Fig.~\ref{fig:dictionary}(g), where the highlighted hyperedges are the support of the codeword.  Now the 1-cycle $\xi_1$ forms the skeleton of the codeword in the manifold code $\C=H_3(\M^8; \ZZ_2)$ corresponding to the 3-cycle $\mathsf{a}_3$ as well as its Poincar\'e dual 5-cocycle ${\mathsf{a}^*}^5$, as illustrated in Fig.~\ref{fig:dictionary}(g,h).  Namely, there exists the following mapping:
\be\label{eq:3-cycle_mapping}
\bar{\mathsf{a}}_1  \rightarrow \mathsf{a}_3 \sim {\mathsf{a}^*}^5. 
\ee
On the other hand, the codeword of the transposed skeleton classical code $\bar{\C}^T$ corresponding to the 0-cocycle $\bar{\bs}^0$ on the hypergraph $G_h$ is promoted to the 2-cocycle $\bs^2$ as well as its Poincar\'e dual 6-cycle $\mathsf{b}^*_6$, in the transposed manifold code $\C^T=H^2(\M^8; \ZZ_2)$, i.e.,
\be\label{eq:2-cocycle_mapping}
\bar{\mathsf{b}}^0  \rightarrow \mathsf{b}^2 \sim \mathsf{b}^*_6.
\ee
Besides, the short 0-cycles and 1-cocycle with $O(1)$ size in $\bar{\C}$ are mapped to the 1-cycles and 2-cocycles as well as their Poincar\'e duals in $\M^8$ respectively, i.e., 
\be
\bar{\mathsf{b}}_0  \rightarrow \mathsf{b}_2 \sim {\mathsf{b}^*}^6, \quad \bar{\mathsf{a}}^1  \rightarrow \mathsf{a}^3 \sim {\mathsf{a}_5^*}.
\ee

\subsubsection{Asymptotically good classical LDPC codes built on manifolds}

\begin{theorem}\label{theorem:manifold_scaling}
The $r$-dimensional triangulated manifold $\M^r$ $(r \ge 8)$ built from the handle construction with the input of Tanner graph $G_T$ of a classical LDPC code $\bar{\C}=\text{Ker}(\bar{\Hs})$ and its transposed code $\bar{\C}^T=\text{Ker}(\bar{\Hs}^T)$ with parameters $[n,k,d]$ and $[n^T,k^T,d^T]$ respectively  satisfies the following properties:
\begin{enumerate}
\item
$\M^r$ has a bounded local geometry, i.e., each vertex in its triangulation is adjacent to $O(1)$ simplices.  
\item 
$\M^r$ contains $\Theta(n)$ number of $k$-simplices $(0 \le k \le r)$, i.e., $\text{dim}(C_k)=\Theta(n)$, where $C_k$ represents the $k^\text{th}$ chain group. 
\item  $b_3=\dim[H_3(\M^r;\ZZ_2)]=k$; \\  $b_2=\dim [H^2(\M^r;\ZZ_2)]=k^T$
\item $sys_3(\M^r;\ZZ_2)= \Omega(d)$,  $sys^2(\M^r;\ZZ_2)=\Omega(d^T)$, where $d$ and $d^T$ are the distances of $\bar{\C}$ and $\bar{\C}^T$ respectively. 
\end{enumerate}
\end{theorem}
\nin (Note that the proof of the above theorem is essentially contained in the proof of Theorem 1.2.1 in Ref.~\cite{freedman:2020_manifold_from_code} which considers the case of mapping a quantum LDPC code to a manifold with bounded geometry. The proof uses the language of geometry and only deals with the codeword related to cycle but not the one related to cocycle.  Here, we instead provide a proof using a combinatorial language more suitable for the QI audience.)

\begin{proof}
In the handle construction, we associate each bit to a dressed ``3-handle" $\mathsf{B}_i$ and each check to a dressed ``2-handle" $\mathsf{C}_j$.  The attaching region of $\mathsf{B}_i$ is attached to the boundary of  $\mathsf{C}_j$ according to the lifted boundary map $\hat{\partial}=\bar{\Hs}$. Now one can try to build a triangulation with $O(1)$ simplices for each ``3-handle" and ``2-handle", and then attach them together to build the 3-handlebody $H$ and then glue it to another upside-down copy $H^*$ to form a double $\D H$ which gives rise to the triangulation $\L$ of the $r$-manifold $\M^r$ with only an O(1) overall  overhead (i.e., the number of simplices are proportional to the number of bits $n$ in the skeleton code $\bar{\C}$ with an $O(1)$ constant, equivalent to property 2 above). 

To show this is possible, we give a concrete algorithm to build the desired triangulation $\L$.  Note that each ``3-handle" and ``2-handle" has a boundary.  This includes the attaching regions of the ``3-handle" $S^2 \times D^{r-3}$, which is attached to the boundary of the ``2-handle"   $\partial \mathsf{C}_j =S^2 \times S^{r-3}$.  After attaching all the ``3-handles" to the ``2-handles", one obtains the 3-handlebody $H$ which is a 3-manifold with boundaries $\partial H$. This boundary is killed when taking the double, which glues $H$ with the upside-down copy $H^*$ along their common boundary  $\partial H =\partial H^*$.  

The algorithm of generating the triangulation $\L^r$ is listed as follows:

\textbf{Algorithm 1} (generate $\L^r$)
\begin{enumerate}
    \item    
We first generate a finite triangulation on each ``2-handle" and ``3-handle", denoted by $\L_2$ and $\L_3$ respectively.  This can be done with some standard triangulation algorithms such as Delauney triangulation (see CGAL library \cite{cgal:eb-24b}).  Note that both ``2-handles" and ``3-handles" are manifolds with simple topology:   $\mathsf{C}_j = S^2 \times D^{r-2}$ and $\mathsf{B}_i= (S^3 \backslash \sqcup_{m=1}^{f(i)}D^3_m) \times D^{r-3} $. Therefore both of them can hence be triangulated with $O(1)$ simplices, namely they remain a constant when scaling up the skeleton classical code $\bar{\C}$ by increasing the number of bits $n$.  The LDPC condition of  $\bar{\C}$ make sure that the number of $\mathsf{C}_j$-``handles" $f(i)$ that a $\mathsf{B}_i$ is connected is upper-bounded by some constant, such that the number of simplices to triangulate $\mathsf{B}_i$ can remain $O(1)$.   
    \item
We then attach the ``3-handles" $\mathsf{B}_i$ to ``2-handles" $\mathsf{C}_j$ along the attaching region $S^2 \times D^{r-3}$.  Since the subcomplexes on the boundaries of  $\mathsf{B}_i$ and $\mathsf{C}_j$, denoted by $\partial\L_2$ and $\partial\L_3$, may not be the same, one needs to re-triangulate $\L_2$ and $\L_3$ to match their boundaries ($\partial \L_2 = \partial \L_3$)  before attaching them by identifying the boundary subcomplexes.  As we will show below, the re-triangulation keeps the number of simplices in each ``handle" to be an $O(1)$ constant.  This gives rise to the 3-handlebody $H$ with a triangulation $\L_H$. 

\item 
We glue the 3-handlebody $H$ with its upside-down copy $H^*$ along their common boundary $\partial H = \partial H^*$.  Since $H$ and $H^*$ are identical and have the same triangulation, i.e., $\L_H = \L_{H^*}$,  the subcomplexes on their common boundaries are also the same, i.e., $\partial \L_H = \partial\L_{H^*}$.  We can hence identify the boundary subcomplexes $\partial \L_H$ and $\partial \L_{H^*}$ to glue the handlebodies together, which generates the triangulation $\L$ of the entire manifold $\M^r$.     
\end{enumerate}

\begin{figure*}[hbt]
\includegraphics[width=0.8\textwidth]{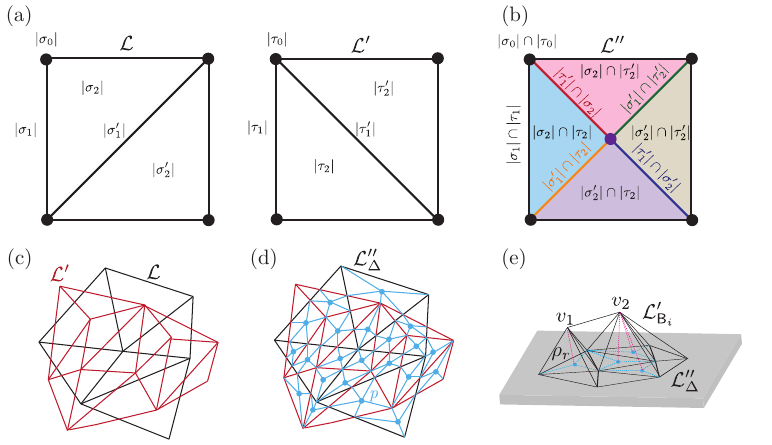}
\caption{(a) Two different triangulations $\L$ and $\L'$ of the same space with labeled 0-simplices, 1-simplices and 2-simplices.  (b) Overlay $\L$ and $\L'$ together. The intersection of the original simplices gives rise to new simplices which form a common subdivision $\L''$ of both $\L$ and $\L'$.  (c) Overlay two more general triangulations $\L$ (black) and $\L'$ (red), which gives rise to a cellular complex. (d) For any cell in the cellular complex that is not a simplex, triangulate it by inserting a vertex $p$ in the center and coning it by connecting it to all the original vertices of the cell. This gives rise to a triangulation $\L''_\Delta$ which forms a common subdivision of both $\L$ and $\L'$.  (e) When the boundary triangulation of a ``handle" gets subdivided (blue dashed line), one can triangulate each bulk $r$-cell $\rho_r$ which is not a simplex  simply by connecting the added vertex to the only vertex of $\rho_r$ that is not on the boundary.   }\label{fig:re-triangulation}
\end{figure*}

We now elaborate on the re-triangulation and boundary matching  procedure in step 2, and show why it keeps the number of simplices being an $O(1)$ constant.   We shall use the following fact:
\begin{fact}\label{fact:common_subdivision}
Any two triangulations of the same piecewise-linear manifold have a common subdivision. 
\end{fact}
Note that the triangulated manifold $\M^r$ we have constructed  here is a piecewise-linear manifold.
This fact is proven abstractly in Ref.~\cite{Rourke_book}. Here we give a constructive proof of this fact based on an ``overlay" algorithm to generate such common subdivision, which can be the subroutain in Algorithm 1.  A subdivison is defined as follows:
\begin{definition}
    The triangulation $\L''$ is a subdivision of a triangulation $\L$, denoted by $\L''  \triangleleft \L$, if all the simplices in  $\L''$ is contained in some simplices in $\L$. 
\end{definition}

As a warm-up, we first start with the  simple case for 2D triangulations:

\textbf{Algorithm 2} (\textit{2D overlay})
\begin{enumerate}
    \item 
\textit{Pairwise intersection}: for the input of two different triangulations $\L$ and $\L'$ of the same piecewise-linear 2-manifold $\M^2$,
one can overlay them on top of each other, as illustrated in Fig.~\ref{fig:re-triangulation}(a,b,c). This gives rise to an intersection cell complex:
\be
 \L'' = \{|\sigma_i| \cap |\tau_j| \neq \varnothing : \sigma_i \in \L, \tau_j \in \L' \ (\text{for} \ i,j=0, 1,2)  \},
\ee

which consists of a collection of cells that are the intersection of the simplices in $\L$ and $\L'$.  Here, $|\sigma_i|$ stands for a geometric simplex located on the manifold $\M^2$, while $\sigma_i$ represents an abstract simplex of $\L$.
For example, $|\sigma_1| \cap |\tau_1|$ represents the vertex at the intersection of the two edges $|\sigma_1|$ and $|\tau_1|$ from the different triangulations,  $|\sigma_2| \cap |\tau_1|$ represents the part of an edge $|\tau_1|$ lies inside $|\sigma_2|$, and $|\sigma_2| \cap |\tau_2|$ represents the 2-cell at the intersection of simplices $|\sigma_2|$ and $|\tau_2|$, as shown in Fig.~\ref{fig:re-triangulation}(b).  
    \item  
\textit{Subdivision by coning}: one then further triangulate the cell complex $\L''$ into a simplical complex (triangulation) $\L''_\Delta$.  For each 2-cell $\rho_2$  in $\L''$ that is not a 2-simplex (triangle), we can pick a point $p$ in its interior (e.g., its barycenter) as an apex and then subdivide the 2-cell $\rho_2$ by \textit{coning}, i.e., connecting the vertices of each edge $[v_i, v_{i+1}]$ in $\rho_2$ to the apex $p$ by lines to form a 2-simplex $[v_i, v_{i+1}, p]$, as illustrated in Fig~\ref{fig:re-triangulation}(d) by blue dots and lines.  
\end{enumerate}

We then present the ``overlay" algorithm for general $r$-dimensional piecewise-linear manifolds:

\textbf{Algorithm 3} (\textit{$r$-dimensional overlay})
\begin{enumerate}
\item
\textit{Pairwise intersection}: for two different triangulations $\L$ and $\L'$ of the same piecewise-linear $r$-manifold $\M^r$,
one can overlay them and obtain an intersection cell complex:
\be
 \L'' = \{|\sigma_i| \cap |\tau_j| \neq \varnothing : \sigma_i \in \L, \tau_j \in \L' \ (\text{for} \ i,j=0, 1, \cdots r)  \},
\ee
where each $i$-cell $\rho_i \in \L''$ is an $i$-dimensional convex polytope.

\item
\textit{Subdivision by coning}:   triangulate $\L$ by induction on cell dimension $k=0, 1, \cdots r$.
\begin{itemize}
    \item For $k=0$: the 0-cells $\rho_0$'s are already 0-simplices (vertices).
    \item  For $k=1, 2, \cdots , r$:  suppose every cell of dimension $< k$ is triangulated and the triangulations agree on shared faces.
For each $k$-cell  $\rho_k \in \L''$ that is not yet a $k$-simplex, pick a point $p_k$ in its relative interior (e.g., its barycenter).  \textit{Cone} the already-triangulated boundary $\partial \rho_k$ to the apex $p_k$, which means that for each $(k-1)$-cell on the boundary $\rho_{k-1} \in \partial \rho_k$, we connect it to $p_k$ to form a $k$-simplex $[\rho_{k-1}, p]$ subdividing $\rho_k$. This yields a triangulation of $\rho_k$ that matches its neighbors (because all use the same triangulation on shared faces by the inductive hypothesis).
\item
After the above $r$ iterations, we obtain a \textit{common subdivision} $\L''_\Delta$ satisfying both $\L''_\Delta \triangleleft \L$ and $\L''_\Delta \triangleleft \L'$.

\end{itemize}
\end{enumerate}

When attaching the ``3-handles" $\mathsf{B}_i$ to ``2-handles" $\mathsf{C}_j$ along the attaching region $S^2 \times D^{r-3}$, we need to match their boundary triangulations on the attaching regions, denoted by $\partial\L_{\mathsf{B}_i}|_{S^2 \times D^{r-3}} =\L$ and  $\partial \L_{\mathsf{C}_j}|_{S^2 \times D^{r-3}}= \L'$ respectively. Since $\L$ and $\L'$ could be different triangulations, we need to convert them to a common subdivision $\L''_\Delta$ so they can match in the attaching region. We subdivide the boundary triangulations on both sides $\L$ and $\L'$ into $\L''_\Delta$ [illustrated in Fig.~\ref{fig:re-triangulation}(e) with the subdivision of the bottom boundary represented by blue dashed lines], which modifies both $\L_{\mathsf{B}_i}$ and $\L_{\mathsf{C}_j}$ on their boundary near the attaching region and yield new cellulations $\L_{\mathsf{B}_i}'$ and $\L_{\mathsf{C}_j}'$ respectively. Note that  in the new cellulations on both sides, some $r$-cells $\rho_r$ next to the subdivided boundary triangulations $\L''_\Delta$  are non-longer simplices.  Therefore, we can subdivide these $r$-cells $\rho_r$.  This can be simply achieved by connecting all the newly introduced vertices in the common subdivision $\L''_\Delta$ to the only vertex of $\rho_r$ that is not located on the boundary  $\L''_\Delta$, as illustrated by the pink dahsed lines in Fig.~\ref{fig:re-triangulation}(e). Now we can attach the ``3-handles" $\mathsf{B}_i$ to ``2-handles" $\mathsf{C}_j$ along their common boundary triangulation $\L''_\Delta$ on the attaching region. We hence complete the re-triangulation and boundary matching procedure in step 2 of Algorithm 1.  Note that the subdivision procedure above only introduces an $O(1)$ constant overhead in each ``2-handles" and ``3-handles".         

When completing Algorithm 1, we have successfully built the triangulation $\L^r$ of the entire $r$-manifold $\M^r$.  We first prove property 1 of Theorem \ref{theorem:manifold_scaling}, i.e., the bounded local geometry of $\M^r$, which is equivalent to the bounded degree of vertices in the triangulation $\L^r$.  Note that in each step of Algorithm 1 generating the triangulation $\L^r$, including its subroutines, the property of the bounded vertex degree is preserved:

We start with a finite triangualtion with $O(1)$ simplices in each ``2-handle" and ``3-handle" (gauranteed by the LDPC property of the skeleton classical code $\bar{\C}$) which has bounded vertex degree.  We then attach the ``3-handles" to the ``2-handles", and the boundary matching procedure via the common subdivision using Algorithm 3 also just introduces an $O(1)$ additional simplices in each ``handle".  Finally, the doubling process also just identify finite number of simplices on the boundary of each pair of identical ``handles". We hence conclude the final triangulation $\L^r$ has bounded vertex degree. 

We then prove property 2 of  Theorem \ref{theorem:manifold_scaling}.  In each step of Algorithm 1 and its subroutines, the number of $k$-simplices ($0 \le k \le r$) in each ``handle" remains an $O(1)$ constant. There are hence only $\Theta(n)$ $k$-simplices in $\L^r$,  corresponding to an $O(1)$ constant  overhead.

To prove property 3, since the handle chain complex $\L_h$ in Eq.~\eqref{eq:handle_chain_complex} contains all the homology information of the manifold $\M^r$, we can use the isomorphism between the lifted  chain complex $\hat{\mathcal{X}}$ of the skeleton classical code Eq.~\eqref{eq:chain_lift}  and a portion of the handle chain complex $\L_h$, which gives rise to the following isomorphism between the homology groups when taking the $\ZZ_2$ coefficients:
\begin{align}
H_3(\M^r; \ZZ_2) \cong H_1(G_h; \ZZ_2),   H^3(\M^r; \ZZ_2) \cong H^1(G_h; \ZZ_2)     \cr
\non H_2(\M^r; \ZZ_2) \cong H_0(G_h; \ZZ_2),   H^2(\M^r; \ZZ_2) \cong H^0(G_h; \ZZ_2).     \cr
\end{align}
This gives rise to the Betti number relation: 
\begin{align}
b_3(\M^r; \ZZ_2)=&\bar{b}_1(G_h; \ZZ_2)=k, \cr
b_2(\M^r; \ZZ_2)=&\bar{b}_0(G_h; \ZZ_2)=k^T.
\end{align}

Finally, we prove property 4 about the (co)systole scaling of the manifold.   First, we need to establish a precise combinatorial relation between every $\ZZ_2$ 3-cycle class $[\as_3]$ in $\M^r$ and the corresponding $\ZZ_2$ 1-cycle class $[\bar{\as}_1]$ in the skeleton classical code $\bar{\C}$,  and between every $\ZZ_2$ 2-cocycle class $[\bs^2]$  and the corresponding $\ZZ_2$ 0-cocycle class $[\bar{\bs}^0]$ in $\bar{\C}$ [see Eqs.~\eqref{eq:3-cycle_mapping} and \eqref{eq:2-cocycle_mapping}].    This can be achieved by deformation retraction from the manifold $\M^r$ and its handle complex $\L_h$ to a CW (cellular) complex $\L_c$, which are isomorphic to each other $\L_h \cong \L_c$.  As mentioned before, one can retract every $k$-handle $D^k \times D^{r-k}$ to its core $D^k$ that becomes a $k$-cell. This also effectively retracts the dressed ``$k$-handle" to its dressed core $N^k$ [see Eq.~\eqref{eq:dressed-k-handle}], which can be considered as a dressed ``$k$-cell" (see Sec.~\ref{sec:deformation_retraction} for details).   As shown in Eq.~\eqref{eq:long_chain}, a portion of the CW  complex $\L_c$ is just isomorphic to the lifted chain complex $\hat{\mathcal{X}}$ of the skeleton classical code $\bar{\C}$. Therefore, the 3-cycle $\tilde{\as}_3$ on the CW  complex $\L_c$ is supported on a set of 3-cells contained in the dressed ``3-cells" $\{N^3_i\}$ completely determined by
the 1-cycle $\bar{\as}_1$ in $\hat{\mathcal{X}}$, where the minimum number of 3-cells in $\{N^3_i\}$ is   the code distance $d$ of $\bar{\C}$.  When pulling the unique 3-cycle $\tilde{\as}_3$ in the CW  complex $\L_c$ back to the 3-cycle class $[\as_3]$ in the triangulated manifold $\M^r$, we know that there must exist some representative $\dot{\as}_3$ in the class $[\as_3]$ that travels through the set of ``3-handles" $\{\mathsf{B}_i = N^3_i \times D^{r-3}\}$ containing the set of dressed cores  $\{N^3_i\}$ mentioned above, such that it is deformed to $\tilde{\as}_3$ under the deformation retraction.   More concretely, the representative $\dot{\as}_3$ is the union of \textit{extended cores} of a set of ``3-handles" $\{\mathsf{B}_i \}$, which is also called a \textit{3-spine}.  Here, extended cores means continuing the dressed core $N^3_i$ of the ``3-handle" $\mathsf{B}_i$ to the ``2-handles" $\mathsf{C}_j$ which it attaches to \cite{freedman:2020_manifold_from_code}.  The 3-spine and extended cores of ``3-handles" have already been illustrated in Fig.~\ref{fig:dictionary_modified}(h), which is highlighted in blue (see Fig.~\ref{fig:cycle_simplicial} for the lower-dimensional example).

\begin{figure}[t]
    \centering
    \includegraphics[width=1\linewidth]{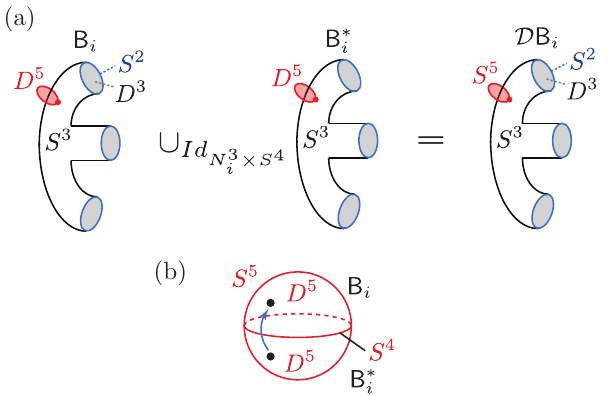}
    \caption{(a) The doubleed ``3-handle" $\D\mathsf{B}_i$ is realized by gluing the ``3-handle" $\mathsf{B}_i=N^3_i \times D^5$ with its identical upside-down copy $\mathsf{B}^*_i$ along their common boundary $N^3_i \times S^4$. The boudled ``3-handle" hence becomes $\D\mathsf{B}_i= N^3_i \times S^5$.  (b) We have used the fact that two copies of $D^5$ glued along their common boundary $S^4$ becomes $S^5$. One can transport the dressed core $N^3_i$ from the ``south-hemisphere" $\mathsf{B}^*_i$ to the `north-hemisphere" $\mathsf{B}_i$ as indicated by the arrow.}
    \label{fig:doubled_handle}
\end{figure}

Note that manifold $\M^r$ only contain dressed ``2-handles" $\mathsf{C}_j$ and ``3-handles" $\mathsf{B}_i$ in the 3-handlebody $H$, as well as their dual in the upside-down handlebody $H^*$, i.e., the ``$(r-2)$-handles" $\mathsf{C}^*_j$ and ``$(r-3)$-handles" $\mathsf{B}^*_i$. Now during the doubling process,  the ``2-handles" and ``3-handles" are glued together with their dual ``handles" respectively, which form the doubled ``handles" denoted by $\D{\mathsf{C}_j}$ and $\D{\mathsf{B}_i}$ respectively.   We note that the processes of handle attachment and doubling commute,  which means one can also first construct the doubled handles and then attach them together.   For example, a pair of identical ``3-handles" $\mathsf{B}_i$ and $\mathsf{B}^*_i$ in the 8-manifold ($r=8$) case can be glued together to the doubled handle:
\be
\D \mathsf{B}_i = (N^3_i \times D^5) \cup_{Id_{N^3_i \times S^4}} (N^3_i \times D^5) = N^3_i \times S^5,
\ee
as illustrated in Fig.~\ref{fig:doubled_handle}(a).
Here, we have used fact that two 5-balls $D^5$ glued along their common $S^4$ boundary forms a 5-sphere $S^5$, i.e., $D^5 \cup_{Id_{S^4}} D^5 = S^5$, as illustrated in  Fig.~\ref{fig:doubled_handle}(b).  We will illustrate the  $\D{\mathsf{C}_j}$ with its lower-dimensional version in the 4-manifold constructed in Sec.~\ref{sec:4-manifold} (see Fig.~\ref{fig:Double_illustration}).

Due to the geometric constraint from the product structure of the ``3-handles", i.e., $N^3_i \times D^{r-3}$ ($N^3_i$ being the dressed core),  any 3-cycle $\as_3$ in the class $[\as_3]$ traveling through a set of doubled ``3-handles" $\{\D\mathsf{B}_i\}$ can always be deformed into the homologous representative $\dot{\as}_3$ (3-spine) which only travels through the same number of ``3-handles" $\{\mathsf{B}_i\}$ contained in $\{\D\mathsf{B}_i\}=\{N_i \times S^5\}$ but not their dual ``handles" $\{\mathsf{B}^*_j\}$.  This can be understood as deforming the dressed core $N^3_i$ along the thickened direction $S^5$, by going from the ``south hemisphere" $\mathsf{B}^*_i$ to the ``north hemisphere" $\mathsf{B}_i$, as illustrated in  Fig.~\ref{fig:doubled_handle}(b).  Note that if the 3-cycle $\as_3$ travels through the ``south hemisphere" $\mathsf{B}^*_i$  in a collection of adjacent doubled ``3-handle" $\{\D\mathsf{B}_i\}$, one should collectively deform the collection of adjacent extended cores of  $\{\D\mathsf{B}_i\}$ to the ``north-hemisphere" $\mathsf{B}_i$.  This collective deformation will finally move the entire $\as_3$ to the 3-spine $\dot{\as}_3$ which are completely located on the ``north-hemisphere".

Now through the correspondence between $\dot{\as}_3$ and $\tilde{\as}$ in the CW  complex  $\L_c$ under deformation retraction,  we know that $\dot{\as}_3$ should travel through at least $d$ ``3-handles" corresponding to the set $\{\mathsf{B}_i\}$ contained in the set $\{\D\mathsf{B}_i\}$. Now when deforming $\dot{\as}_3$ to any homologous 3-cycle $\as_3$ in the same class $[\as_3]$, we know that $\as_3$ cannot go through a different set of doubled ``3-handles"  $\{\D\mathsf{B}'_j\}$.  Assuming it can, then $\as_3$ can always be deformed into a representative   $\as'_3$ that travels through the set of ``3-handles" $\{\mathsf{B}'_j\}$ that are different from the original set $\{\mathsf{B}_i\}$.  In that case, when applying deformation retraction, $\as'_3$ will be mapped to a different 3-cycle $\tilde{\as}' \neq \tilde{\as}$, this suggests that $\as'_3$ is in a different class of $\as_3$, i.e., $[\as'_3] \neq [\as_3]$, which leads to a contradiction.  Therefore, same as the special representative $\dot{\as}_3$, any other representative $\as_3$ should also travel through at least $d$ doubled ``3-handles".  

Now consider the triangulation $\L^r$, any representative $\as_3$ must travel through at least $d$ simplices.  We hence obtain 
\be
sys_3(\M^r; \ZZ_2) = \min_{\as_3 \neq 0 \in H_3(\M^r; \ZZ_2)} \{|\as_3|\}= \Omega(d),
\ee
(see Definition \ref{def:systole}).
In a completely analogous manner, we can prove that any non-trivial 2-cocycle $\bs^2$ must travel through at least $d^T$ doubled ``2-handles" $\D\mathsf{B}_i$ according to its correspondence to the 2-cocycle $\tilde{\bs}^2$ in the CW  complex $\L_c$ under deformation retraction.  We hence obtain 
$d$ simplices.  We hence obtain 
\be
sys^2(\M^r; \ZZ_2) = \min_{\bs^2 \neq 0 \in H^2(\M^r; \ZZ_2)} \{|\bs^2|\}= \Omega(d^T).
\ee
(see Definition \ref{def:cosystole}).

\end{proof}
Note the above way to show that the minimal (co)cycle length $\as_3$ and $\bs^2$ in the manifold is proportional to the cycle in the skeleton classical code $\bar{\as}_1$ and $\bar{\bs}^0$ can be straightforwardly adapted to the situation of the lower-dimensional 4-manifold construction in Sec.~\ref{sec:4-manifold}, and we will not repeat this argument again later.

Based on the results in Theorem \ref{theorem:manifold_scaling}, we can obtain the following theorem about the classical code defined on $\M^r$:
\begin{theorem}\label{theorem:good_code_on_manifold}
From the Tanner graph of any input skeleton classical LDPC code $\bar{\C}=\text{Ker} (\bar{\Hs})$ that both itself and its transposed code $\bar{\C}^T=\text{Ker} (\bar{\Hs}^T)$ are asymptotically good, one can define a classical LDPC code $\C = H_3(\L^r; \ZZ_2)$ on the triangulation $\L^r$ of the manifold $\M^r$ that both itself and its transpose $\C^T$ are asymptotically good.  
\end{theorem}

\begin{proof}
This is a direct application of Theorem \ref{theorem:manifold_scaling}.  The existence of such skeleton classical code $\bar{\C}$ is proven by the explicit construction in Lemma \ref{lemma:classical_code}.
\end{proof}
Note that a special case of the above theorem is that the skeleton classical code $\bar{\C}$ is asymptotically good but not true for its transposed code $\bar{\C}^T$. 

Theorem \ref{theorem:good_code_on_manifold} has accomplished the construction of an asymptotically good classical code on a high-dimensional simplicial complex without local codes (no sheaf structure), which is also a high-dimensional expander.   This improves on the results in Refs.~\cite{10.13069/jacodesmath.617235, 10.1007/s00454-017-9926-3, 10.48550/arxiv.1010.1400, 10.1007/s00493-006-0027-9, 10.48550/arxiv.math/0609773, 10.48550/arxiv.math/0609773} which give rise to good classical codes on high-dimensional simplicial complexes but are not LDPC.

\subsection{Lower-dimensional construction of 4-manifolds and the corresponding thickened classical codes}\label{sec:4-manifold}

\subsubsection{The modified construction}

\begin{figure*}[t]
\includegraphics[width=1\textwidth]{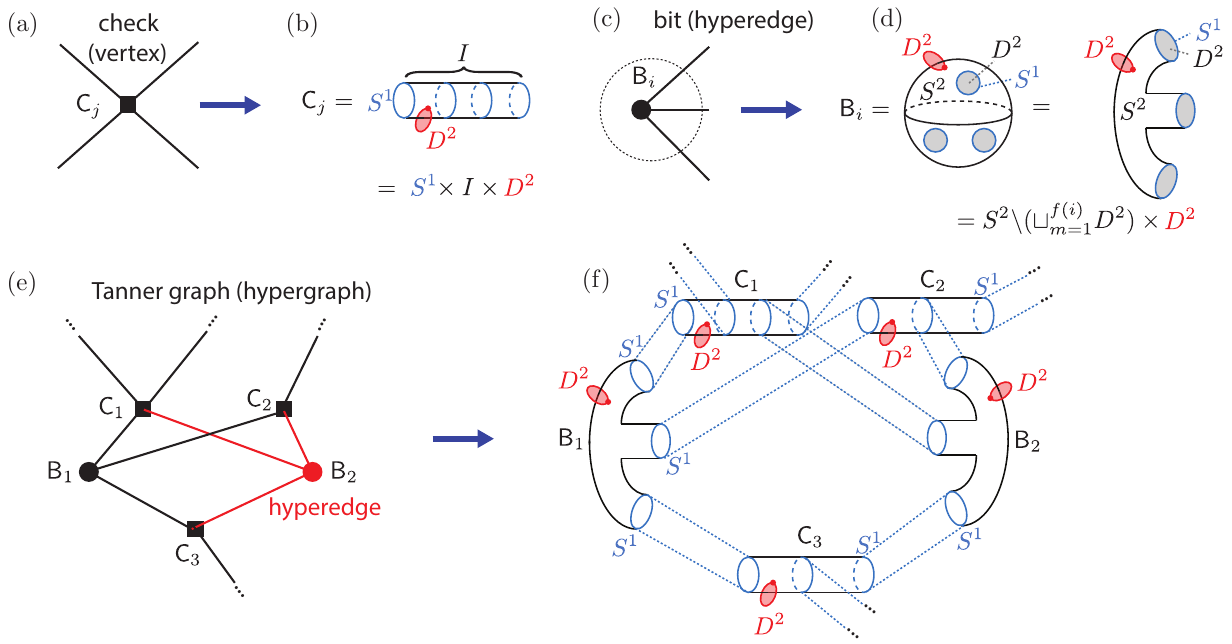}
\caption{(a,b)  The check $\mathsf{C}_j$ is mapped to a ``1-handle" $\mathsf{C}_j=S^1 \times I \times D^2$.  (c,d)  The bit  $\mathsf{B}_i$ is mapped to a ``2-handle" corresponding to a 2-sphere $S^4$ with multiple 2-disks $D^2$ being removed and then further thickened by $D^2$.  (e) A Tanner graph or equivalently a hypergraph. (f) The Tanner graph is thickened to a handlebody where the ``2-handles" are attached to the adjacent ``1-handles" according to the boundary map in the skeleton Tanner graph and via gluing them along the attaching region $S^1 \times D^2$ of the ``2-handles". }\label{fig:dictionary_modified}
\end{figure*}

\begin{figure*}[t]
\includegraphics[width=1\textwidth]{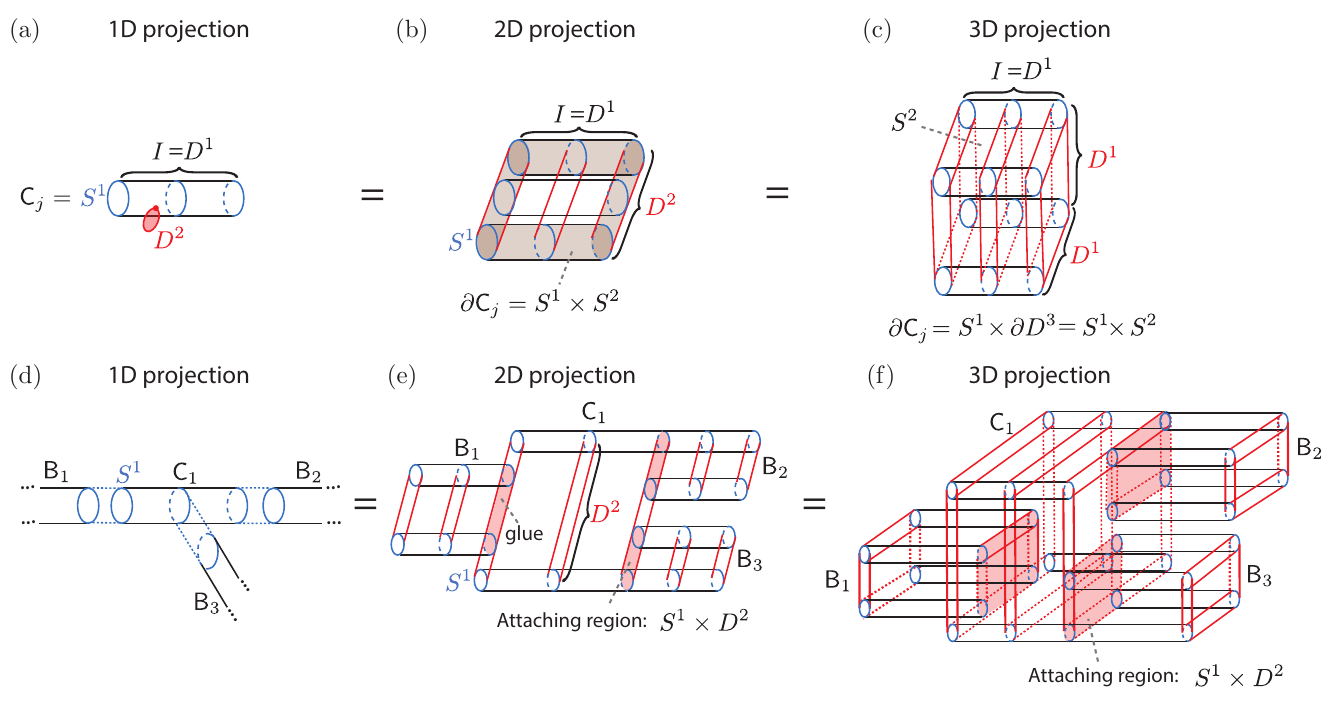}
\caption{(a) The ``1-handle'' $\mathsf{C}_j$ represented in the 1D projection, where both $D^2$ and $S^1$ should be viewed as along extra dimensions. (b) 2D projection: $S^1$ can be displaced along $I=D^1$ or $D^2$. (c) 3D projection:  $S^1$ can be  displaced along $D^1$ in three different directions. One can view the ``1-handle'' as a 3D cube thickened along $S^1$ in the extra dimension. The boundary of the cube is a 2-sphere $S^2$, while the boundary of the ``1-handle" is a 2-sphere thickened along $S^1$, i.e., $S^1 \times S^2$.  (d)  Illustration of three ``2-handles" $\mathsf{B}_1$,  $\mathsf{B}_2$ and $\mathsf{B}_3$ attached to a ``1-handle" $\mathsf{C}_1$ in the 1D projection.  (e) 2D projection: the attaching regions $S^1 \times D^2$ (highlighted) of three ``2-handles" are glued to the boundary of the ``1-handle". (f) 3D projection: one can visualize the attaching region (highlighted) as a thickened 2-disk: $S^1 \times D^2$.     
}\label{fig:branching_illustration}
\end{figure*}

\begin{figure*}[hbt]
\includegraphics[width=1\textwidth]{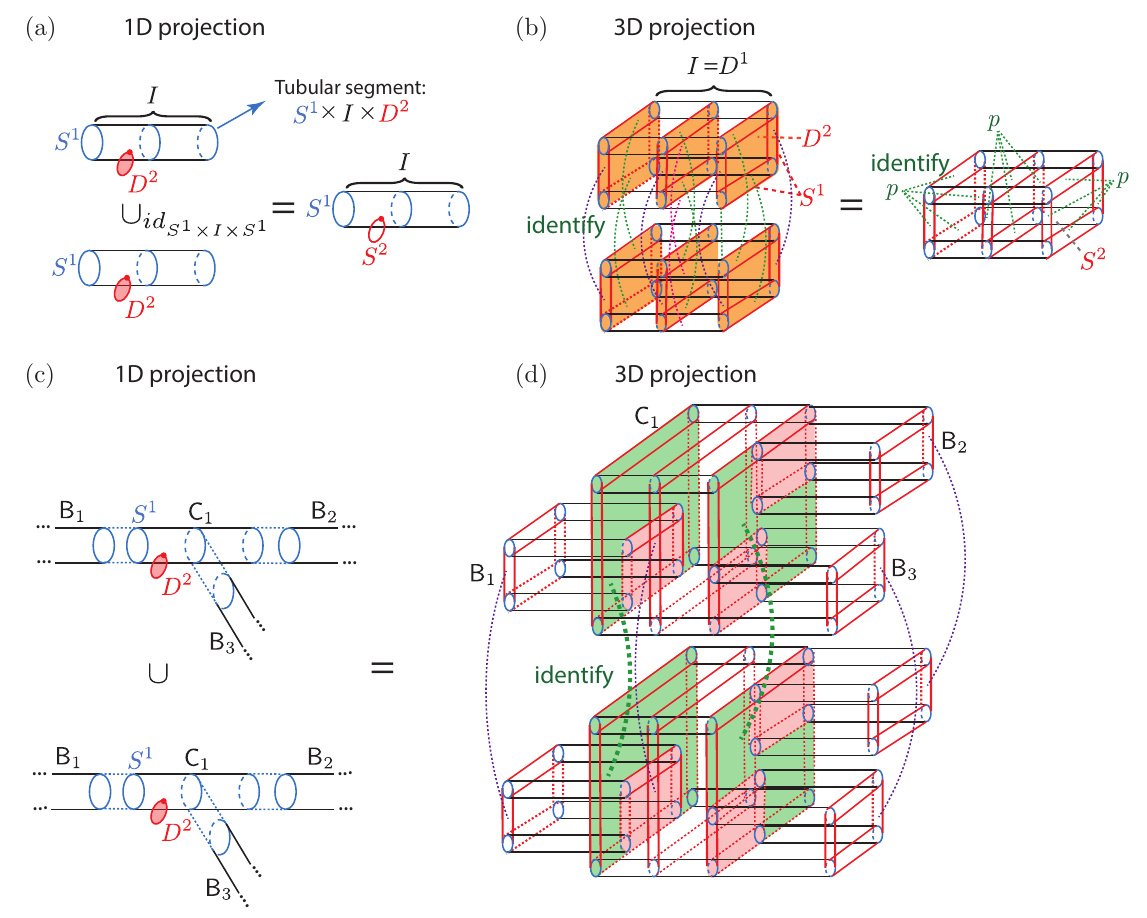}
\caption{(a) Illustration of the double of a tubular segment $S^1 \times I \times D^2$ in the 1D projection. Two $D^2$'s are glued along their common $S^1$ boundary and form an $S^2$.  The double hence becomes $S^1 \times I \times S^2$.   (b) Illustration of the double of a tubular segment in the 3D projection. Note that the vertical 2-disks $D^2$ in the two copies are highlighted and glued in pair along the common $S^1$ boundary. This becomes equivalent to a single copy where the boundary of each 2-disk $D^2$ are identified to a single point and hence becomes $S^2$.  (c)  Illustration of the double of two junction regions where three ``2-handles" are attached to a ``1-handle". (d) Visualizing the junction region in the 3D projection: the 2-disks $D^2$ including the attaching regions (highlighted in red) in the two copies are identified in the same way as the tubular segmenet in (b). The boundary regions of the 1-handle $\mathsf{C_j}$ outside the attaching regions (highlighted in green) in each copy are identified with each other like a wormhole.  }\label{fig:Double_illustration}
\end{figure*}

  In order to lower the dimensions of the constructed manifold for practical purpose (although we will see in Sec.~\ref{sec:deformation_retraction} that the dimension may not be major concern in the context of CW complex), we will now modify the manifold construction in Ref.~\cite{freedman:2020_manifold_from_code} to a 4-manifold $\M^4$ where the bits are now placed on 2-simplices (2-cells). From the perspective of systolic geometry, the lowering of the manifold dimension can lead to spurious cycles/cocycles with $O(1)$ size which makes the systole only $O(1)$, as has been discussed in Ref.~\cite{freedman:2020_manifold_from_code}.  However, as will be elaborated later, one can use a subsystem code idea that chooses proper cycle/cocycle basis to encode the logical information, such that the distance in the subsystem code is still large.

The modified construction still follows the handle construction as discussed in Sec.~\ref{sec:handle_construction}.    We choose the ``1-handles" corresponding to the checks as 
\be
\mathsf{C}_j=(S^1 \times D^3)_j = (S^1 \times I \times D^2)_j, 
\ee
as shown in Fig.~\ref{fig:dictionary_modified}(a,b).  Similarly, we choose the ``2-handles'' corresponding to the bits as the thickened  punctured 2-spheres 
\be
\mathsf{B}_i=\left(S^2 \backslash (\sqcup_{m=1}^{f(i)}D^2) \times D^2 \right)_i, 
\ee 
as illustrated in Fig.~\ref{fig:dictionary_modified}(c,d).  Similar to the 8-manifold construction above, we attach the handle according to the lifted boundary map $\hat{\partial}$ and Eq.~\eqref{eq:attaching_rule}, which is also encoded in the Tanner graph in Fig.~\ref{fig:dictionary_modified}(c).  The attaching map from the ``2-handles" $\textsf{B}_i$ to the ``1-handles" $\textsf{C}_j$ corresponds to disjoint embedding of the attaching regions:
\be\label{eq:attaching_map_modified}
(S^1 \times D^2)_{i, m}  \hookrightarrow \partial(S^1 \times D^3)_j= (S^1 \times S^2)_j, 
\ee
 as illustrated abstractly in Fig.~\ref{fig:dictionary_modified}(f) and concretely in Fig.~\ref{fig:branching_illustration}.   The $\mathsf{C}_j$-``handle", which has dimension 4, is illustrated in Fig.~\ref{fig:dictionary_modified}(a) with a 1D projection, can also be shown with a 2D and 3D projection respectively as in Fig.~\ref{fig:branching_illustration}(b,c).  In the 2D projection in (b), the cylinder $S^1 \times I$ is displaced along the $D^2$ direction, where the shaded region represents the boundary of the $\mathsf{C}_j$-``handle": $\partial \mathsf{C}_j =S^1 \times \partial D^3 = S^1 \times S^2$.  In the 3D projection in (c), the cylinder is displaced along two $D^1$ directions, which gives rise to a thickened cube, i.e., cube $D^3$ times a circle $S^1$ in the extra dimension. Now the boundary of this cube is $\partial D^3=S^2$, which gives rise to the boundary of the $\mathsf{C}_j$-``handle" $S^1 \times S^2$.   We then visualize the attaching map in Eq.~\eqref{eq:attaching_map_modified} abstractly illustrated in Fig.~\ref{fig:branching_illustration}(d) with the concrete 2D projection in Fig.~\ref{fig:branching_illustration}(e), where the $\mathsf{B}_1$, $\mathsf{B}_2$ and $\mathsf{B}_3$ ``handles" are attached (glued) to the boundary of the $\mathsf{C}_1$-handle.   The attaching regions $S^1 \times D^2$ are highlighted in red. A  more concrete illustration is shown via the 3D projection in Fig.~\ref{fig:branching_illustration}(f),  where one can clearly see that the $\mathsf{B}_i$-``handles" are attached to the boundary of the thickened cube $S^1 \times D^3$ (as a thickened 2-sphere $S^1 \times S^2$) with the attaching regions $S^1 \times D^2$ being highlighted.

 So far, we have built the desired 2-handlebody $H$ from the classical code, which is essentially a thickened Tanner graph.  We then take the double of $H$ to obtain the closed 4-manifold $\M^{4}=\mathcal{D}H$ similar to the previous construction of the 8-manifold.  Consider a tubular segment $S^1 \times I \times D^2$ (i.e., a tubular neighborhood of $S^1 \times I$) either inside the $\mathsf{C}_j$- or the $\mathsf{B}_i$-handles, the double produces 
 \be
(S^1 \times I \times D^2) \cup_{id_{S^1\times I \times S^1}} (S^1 \times I \times D^2) =S^1 \times I \times S^2
 \ee
 since two disks glued along their common $S^1$ boundary form a sphere, i.e.,  $D^2 \cup_{id_{\partial D_2=S^1}} D^2 = S^2$, as shown in Fig.~\ref{fig:Double_illustration}(a).  More concrete  illustration with 3D projection is shown in Fig.~\ref{fig:Double_illustration}(b), where each square region (orange) $D^2$ has a boundary $S^1$ which is identified to the $S^1$ boundary in the other identical copy.  At the junction where multiple $\mathsf{B}_i$-``handles" are attached to a single $\mathsf{C}_j$-``handles" as illustrated abstractly in Fig.~\ref{fig:Double_illustration}(c) and more concretely with the 3D projection in Fig.~\ref{fig:Double_illustration}(d),  the only additional care that needs to be taken is on the left and right boundaries of the $\mathsf{C}_j$-``handle" where we attach the $\mathsf{B}_i$ ``handles". Note that the regions (green) outside the attaching regions (red), i.e., $(S^1 \times D^2) \backslash \sqcup_m {(S^1 \times D^2)_m} $ are boundaries in the 2-handlebody $H$, which are hence identified with the same regions in the identical copy as shown in Fig.~\ref{fig:Double_illustration}(d).   The rest of the regions are all tubular segments $S^1 \times I \times D^2$ which can be identified using the rules in Fig.~\ref{fig:Double_illustration}(a,b) as described above.

We now reach the following lemma: 
\begin{lemma}\label{lemma:thickened_code}
The $r$-dimensional manifold $\M^r$ $(r \ge 4)$ built from handle construction with the input of the Tanner graph $G_T$ of a classical LDPC code $\bar{\C}$ with $n$ bits satisfies the following properties:
\begin{enumerate}
\item
$\M^r$ has a bounded local geometry, i.e., each vertex in its triangulation is adjacent to $O(1)$ simplices.  
\item 
$\M^r$ contains $\Theta(n)$ total number of $k$-simplices $(0 \le k \le r)$, i.e., $\text{dim}(C_k)=\Theta(n)$, where $C_k$ represents the $k^\text{th}$ chain group. 
\end{enumerate}
\end{lemma}
The proof is the same as the proof for Theorem \ref{theorem:manifold_scaling}.


\subsubsection{Cycle and cocycle mapping}\label{sec:cycle_mapping}

\begin{figure*}[t]
\includegraphics[width=2\columnwidth]{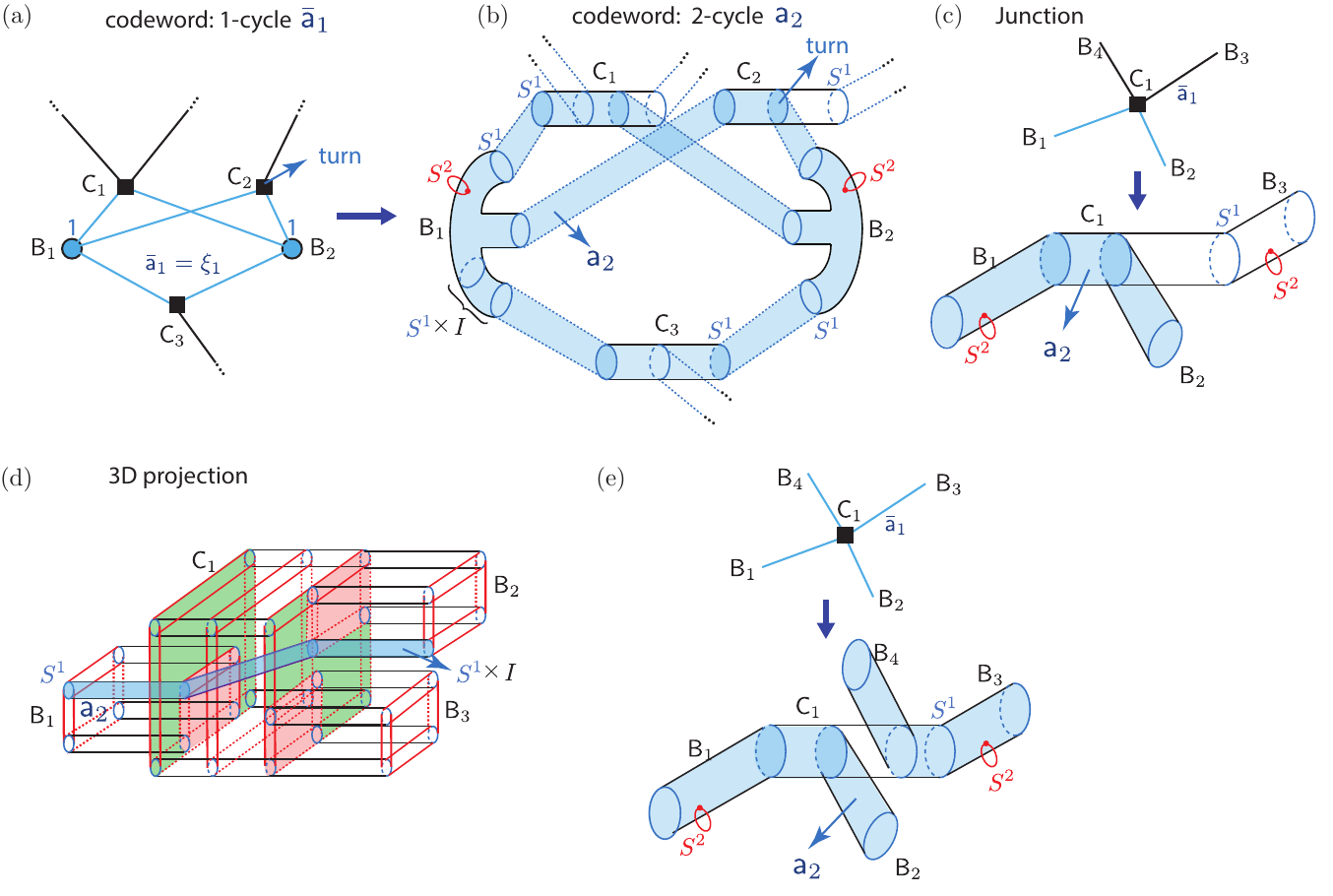}
\caption{(a) The codeword as a 1-cycle $\bar{\as}_1$ in the skeleton hypergraph. The 1-cycle can turn at the check (square vertex).  (b) Visulization of the 4-manifold, where each tubular segment is thickened along $S^2$. The codeword is thickened into a 2-cycle $\as_2$ in the manifold.   (c) The thickened 2-cycle codeword can turn at the junction where multiple ``2-handles" are attached to a single ``1-handle". (d)  Geometric understanding of the turning of 4-cycle in the 3D projection.  The 2-cycle can be viewed as the worldsheet (trajectory) of $S^1$. The $S^1$ can be transported along the extra dimensions.  (e) The codeword in the skeleton hypergraph always occupy even number of hyperedges (bits) connected to a vertex (check).  When there are more than two hyperedges being occupied in the skeleton hypergraph, one can always resolve the corresponding thickened 2-cycles in the manifolds via the equivalence relation in homology.   }\label{fig:cycle_simplicial}
\end{figure*}

There is a following mapping between the basis cycles/cocycles in the skeleton classical code and those in the thickened simplicial LDPC code defined on the triangulation $\L$ of the 4-manifold $\M^4$:
\begin{align}
\bar{\textsf{a}}_1 \rightarrow \mathsf{a}_2, \quad 
\bar{\textsf{b}}_0 \rightarrow \mathsf{b}_1,  \quad
\bar{\textsf{a}}^1 \rightarrow \mathsf{a}^2,  \quad  \bar{\textsf{b}}^0 \rightarrow \mathsf{b}^1,  
\end{align}
which belong to the cycle/cocycle basis $\{ \mathsf{a}_2\}$,  $\{ \mathsf{b}_1\}$, $\{ \mathsf{a}^2\}$ and $\{ \mathsf{b}^1\}$.
Due to the introduction of Poincar\'e duality $H_k(\L; \ZZ_2)$$\cong$$H^{4-k}(\L^*; \ZZ_2)$ from the double of the 4-handlebody,  the above cycles and cocycles all have their Poincar\'e dual cocycles and cycles on the dual triangulation $\L^*$: 
\be\label{eq:combined_isomorphism}
\mathsf{a}_2 \sim \mathsf{a}^{*2}, \quad 
\mathsf{b}_1 \sim \mathsf{b}^{*3},  \quad
\mathsf{a}^2 \sim \mathsf{a}^*_2,  \quad  
\mathsf{b}^1 \sim \mathsf{b}^*_3. 
\ee
These new cyclce/cocycle classes in the dual triangulation $\L^*$ also have corresponding classes in the original triangulation $\L$ due to the isomorphism $H_k(\L; \ZZ_2) \cong H_k(\L^*; \ZZ_2) \cong H_k(\M^4; \ZZ_2)$ and $H^k(\L; \ZZ_2) \cong H^k(\L^*; \ZZ_2) \cong H^k(\M^4; \ZZ_2)$, which have essentially the same support in the continuous picture of the manifold $\M$.  For simplicity, we also use the same set of notations to represent these corresponding cycles/cocycles in the original triangulation $\L$, i.e., $\mathsf{a}^{*2}, \mathsf{b}^{*3}, \mathsf{a}^*_2$ and $\mathsf{b}^*_3$. 

Besides the Poincar\'e duality isomorphism, since we are considering $\ZZ_2$ homology, there is an additional isomorphism between the $k^\text{th}$ $\ZZ_2$-homology and cohomology in the same triangulation $\L$, i.e., $H_k(\L; \ZZ_2)$$\cong$$H^{k}(\L; \ZZ_2)$ due to the universal coefficient theorem \cite{Hatcher:2001ut}.  One can consider the $k^\text{th}$ homology group $H_k$ and $k^\text{th}$ cohomology group $H^k$ as $\ZZ_2$ vector spaces.  The above isomorphism corresponds to the following pairing (inner product) of the basis vector ($k$-cycle $\mathsf{c}_k$) and the dual basis vector ($k$-cocycle $\mathsf{c'}^k$):
\be\label{eq:conjugate_relation}
\int_{\M^4} (\mathsf{c}_k)(\mathsf{c'}^k) \equiv \int_{\mathsf{c}_k} \mathsf{c'}^k = |\mathsf{c}_k \cap \mathsf{c'}^k| = \delta_{\mathsf{c}, \mathsf{c'}},
\ee
where in the second expression we sum the $\ZZ_2$-coefficients of the cocycle $c'^k$ over the cycle $\mathsf{c}_k$ and is equivalent to the number of overlap $k$-simplices between the support of the cycle-cocycle pair.   We call the cycle and cocycle with the same label $\mathsf{c}$, i.e.,  $\mathsf{c}_k$ and  $\mathsf{c}^k$ (and hence overlapping on single $k$-simplex) a conjugate pair, since later they will be used to compose the logical-$Z$ and logical-$X$ operators respectively (which are conjugate variables) in the product construction.   In our example, we have the following conjugate pairs:  $\mathsf{a}_2 \leftrightarrow \mathsf{a}^2$ and $\mathsf{b}_1 \leftrightarrow \mathsf{b}^1$.

Now combining the above isomorphism with Poincar\'e duality, we have 
\be
H_k(\L; \ZZ_2) \cong H^{4-k}(\L^*; \ZZ_2)\cong  H_{4-k}(\L^*; \ZZ_2)\cong H_{4-k}(\L; \ZZ_2).
\ee
This leads to a pair of dual cycles on $\L$:  $\mathsf{c}_k$ and $\mathsf{c}^*_{4-k}$, which has a non-trivial $\ZZ_2$ intersection with each other, namely
\be
|\mathsf{c}_k \cap \mathsf{c}^*_{4-k}| = 1.
\ee
More generally, for any pair of basis cycles, one has
\be
|\mathsf{c}_k \cap \mathsf{c'}^*_{4-k}| = \delta_{\cs, \cs'},
\ee
which is equivalent to the overlapping relation of the conjugated cycle-cocycle pair in Eq.~\eqref{eq:conjugate_relation}.
Note that the above intersection condition can be re-written as the cup product sum of their Poincar\'e dual cocycles:
\be
\int_{\M^4} \mathsf{c}^{*4-k} \cup \mathsf{c'}^{k}= \delta_{\cs, \cs'},
\ee
which is related to the isomorphism $H^k(\L^*; \ZZ_2) \cong H^{4-k}(\L^*; \ZZ_2)$.   Note the above discussion about the various isomorphism also apply to general dimension $r$, including the 8D manifold $(r=8)$ studied in Sec.~\ref{sec:handle_8-manifold}.

In the following, we analyze the properties of each type of cycles or cocycles as well as their Poincar\'e duals, along with illustrations of the underlying geometry.  Note that these analysis and illustration can be easily generalized to higher dimension $r$ by proper increasing of the dimension of the core and co-core of the handles, such as the $r=8$ case in Sec.~\ref{sec:handle_8-manifold}. The properties for $r\ge 8$ are similar to the 4-manifold case studied in this subsection.
\begin{enumerate}
\item 
\textit{2-cycles:}

We first investigate the codeword of the classical code $\bar{\C}$ associated to 1-cycle $\bar{\textsf{a}}_1$, which is mapped to the 2-cycle $\mathsf{a}_2$  in the triangulation $\L$ of the manifold $\M^4$, as illustrated in Fig.~\ref{fig:cycle_simplicial}(a,b).  The 1-cycle $\bar{\textsf{a}}_1 = \xi_1$ can be considered as a sub-hypergraph of the hypergraph $G_H$, which forms the skeleton of the 2-cycle $\mathsf{a}_2$ in the thickened code. More concretely, this means the thickened 2-cycle $\mathsf{a}_2$ is the boundary of the tubular neighborhood of $\xi_1$, which has the form 
\be
\mathsf{a}_2|_\tau = \xi_1|_{\tau} \times \partial D^2=\xi_1|_{\tau} \times S^1= S^1\times I
\ee
in a local tubular segment $\tau=S^1 \times I \times S^2$, where $\mathsf{a}_2|_\tau$ and $\xi_1|_{\tau}=I$ dnotes the segment of $\mathsf{a}_2$ and $\xi_1$ supported within $\tau$.  Since each $\mathsf{B}_i$ ``4-handle" corresponds to a single hyperedge in the hypergraph, the 2-cycle needs to completely experience the entire $\mathsf{B}_i$ ``2-handle" it goes through, as shown in Fig.~\ref{fig:cycle_simplicial}(b).   We emphasize that although $\bar{\textsf{a}}_1 = \xi_1$ is a non-geometric 1-cycle defined on a hypergraph which is not a simplicial complex, the thickened 2-cycle $\as_2$ is a geometric cycle defined on a simplicial complex.
We have hence turned the classical code into a geometric object by shifting to higher dimensions.

Another property of the 1-cycle codeword of the classical code is that it can \textit{turn} at the check vertex $\mathsf{C}_j$, meaning that the cycle can occupy only a portion of the hyperedges connected to the check vertex, as illustrated in Fig.~\ref{fig:cycle_simplicial}(a).  Therefore, the thickened 2-cycle codeword should also be able to turn at a $\mathsf{C}_j$-``handle", as illustrated in Fig.~\ref{fig:cycle_simplicial}(b, c).  The turning location corresponds to the junction where more than two $\mathsf{C}_j$-``handles" are attached to a $\mathsf{B}_i$-``handle".  The turning is more concretely illustrated with the 3D projection in Fig.~\ref{fig:cycle_simplicial}(d), where the 2-cycle $\as_2$ can occupy a tubular region $S^1 \times I$ in this segment which goes directly to one of the other $\mathsf{B}_i$-``handles". 

Now for the 2-cycle $\as_2$ in our construction,  it can only occupy  even number of attached $\mathsf{B}_i$-``handles" at each junction to satisfy the cycle condition $\partial_2 \as_2 = 0$.  When there are more than two $\mathsf{B}_i$-``handles" being occupied, one can resolve the 2-cycle by pairing up the occupied $\mathsf{B}_i$-``handles" in an arbitrary way, as illustrated in Fig.~\ref{fig:cycle_simplicial}(e).

\item 
\textit{1-cocycles:}

\begin{figure*}[t]
\includegraphics[width=2\columnwidth] {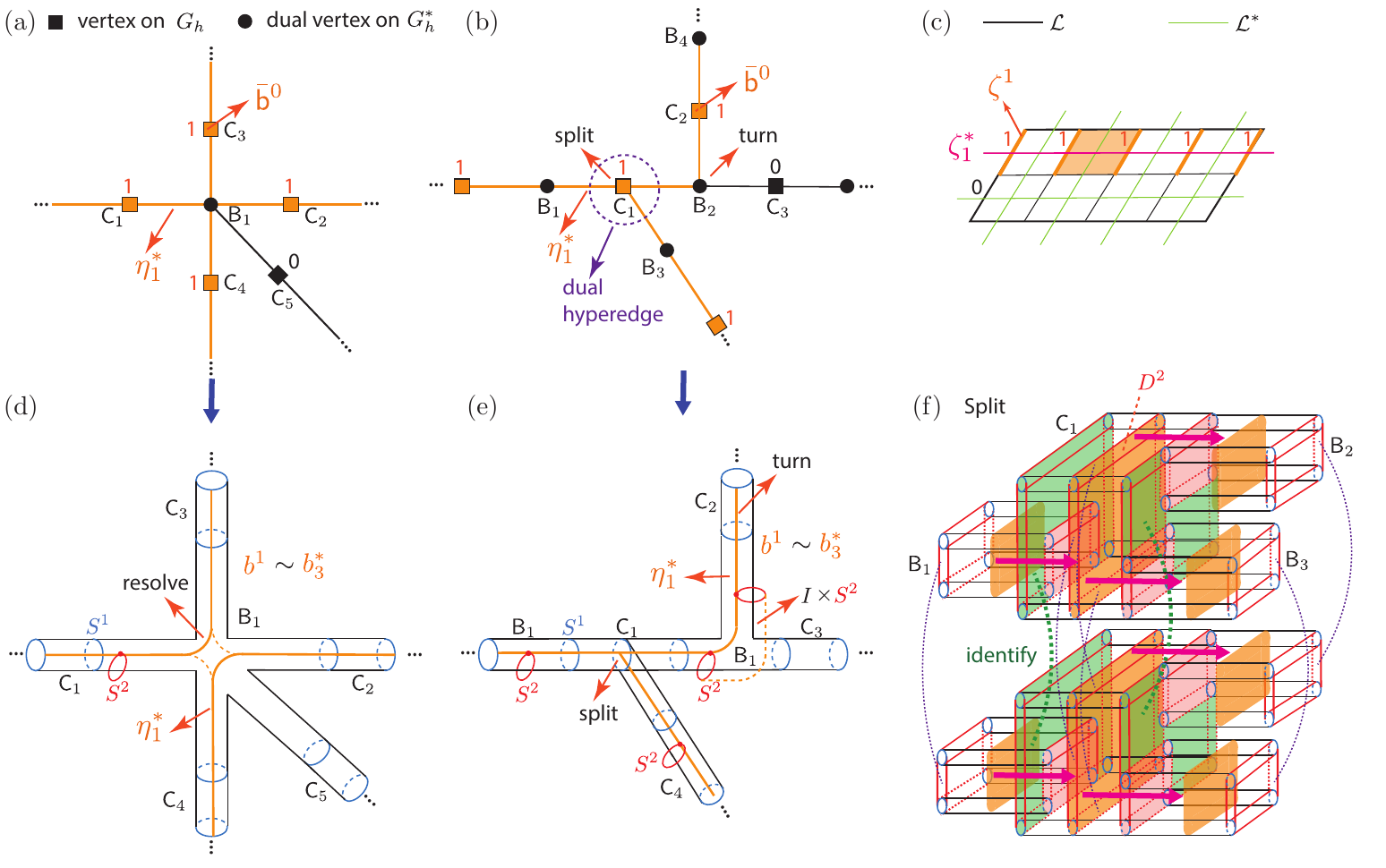}
\caption{(a) Illustration of the 0-cocycle $\bar{\bs}^0$ in the skeleton hypergraph $G_h$ occupying vertices (square) and its dual 1-cycle $\eta^*_1$ occupying the hyperedges in the dual hypergraph $G^*_h$.  (b) The 0-cocycle can turn at a hyperedge (bit) and has to split into all branches at a vertex (check). (c) Illustration of a cocycle and its Poincar\'e dual cycle on a 2D square complex. (d) In the 4-manifold, the thickened 1-cocycle and its dual 3-cycle locally look like $I \times S^2$ in a tubular segmenet, and can be resolved using the equivalence relation of cocycles when going into more than two legs in a $\mathsf{B}_i$-``handle.   (e)  The turning of the 1-cocycle in a $\mathsf{B}_i$-``handle" and splitting of the 1-cocycle in the junction region of the $\mathsf{C}_j$-``handle". (f) Understanding the geometry of the splitting of 1-cocycle in the 3D projection.  The 1-cocycle and its dual 3-cycle can be viewed as the worldsheet (trajectory) of a pair of vertical 2-disks $D^2$ glued along their common $S^1$ boundary.  When moving across the boundary of the $\mathsf{C}_1$-``handle", a pair of 2-disks can be created or annihilated at the ``wormhole" regions (highlighted in green) where the boundaries of the two copies are identified. During this moving process, a smaller vertical 2-disk hence first becomes a larger 2-disk ane then splits into two smaller 2-disks. }\label{fig:cocycle_property}
\end{figure*}

We then consider the codeword of the transposed classical code $\bar{\C}^T$ associated with the 0-cocycle $\bar{\bs}^0$, which is mapped to the 1-cocycle  $\bs^1$  in the thickened code defined on the triangulation $\L$ of the manifold $\M^4$, as illustrated in Fig.~\ref{fig:cocycle_property}.  

The 0-cocycle $\bar{\bs}^0$ occupies vertices $v$ (orange squares) on the hypergraph $G_h$. In the dual hypergraph $G_h^*$ where the vertex and hyperedge is interchanged, $\bar{\bs}^0$ corresponds to a dual 1-cycle $\eta^*_1 \in H_1(G^*_H; \ZZ_2)$ occupying the dual hyperedges $e_h^*$, as illustrated by orange lines in  Fig.~\ref{fig:cocycle_property}(a).   

An important property is that the 0-cocycle $\bar{\bs}^0$ can turn at a bit $\mathsf{B}_i$ (circle), as illustrated in  Fig.~\ref{fig:cocycle_property}(b). In the dual hypergraph picture, the dual 1-cycle   $\eta^*_1$ turns at the dual vertex  $\mathsf{B}_i$.  Note that in order for $\bar{\bs}^0$ to be a valid codeword of  $\bar{\C}^T$ or equivalently a valid 0-cocycle,  at each bit $\mathsf{B}_i$, there has to be even number of adjacent checks being occupied by $\bar{\bs}^0$ [see Fig.~\ref{fig:cocycle_property}(a,b)] in order to satisfy the parity constraint in $\bar{\C}^T$ or equivalently the zero coboundary condition at the corresponding hyperedge $e_{h,i}$, i.e., $d \bar{\bs}^0 (e_{h, i})$$=$$0$.   In the dual hypergraph picture, this means a single dual vertex (circle) has to be adjacent to even number of occupied dual hyperedges.

\begin{figure*}[t]
\includegraphics[width=1.3\columnwidth]{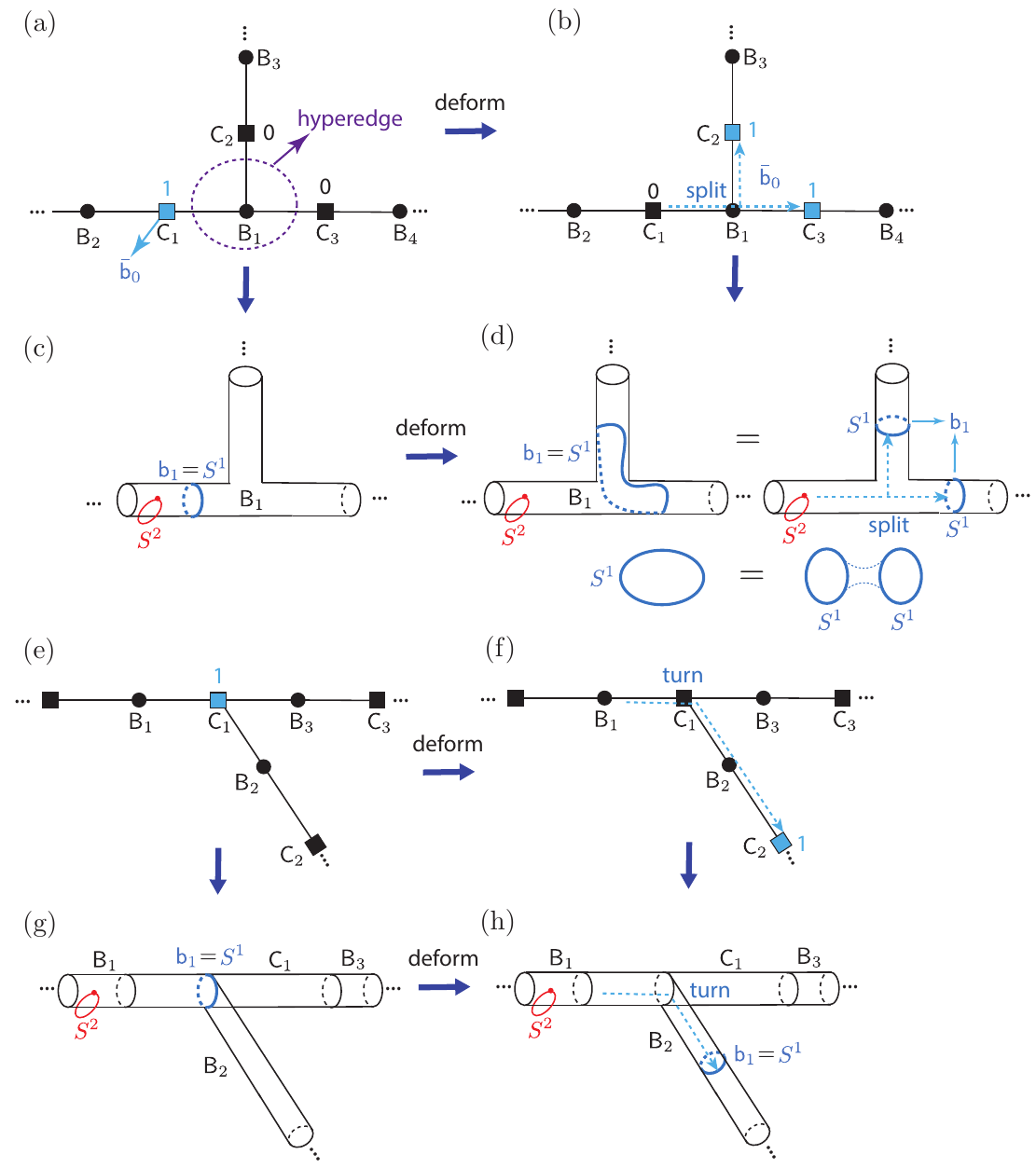}
\caption{(a, b) A 0-cycle occupying a single vertex (check) in the skeleton hypergraph can be deformed into a representative occupying the rest of vertices connected to the same hyperedge $\mathsf{B}_1$. (c,d) In the corresponding manifold, the thickened 1-cycle $S^1$ can be deformed and split into two $S^1$ when going through a $\mathsf{B}_i$ handles with more than two legs according to the homology relation shown below. (e,f) The 0-cycle representative can turn at the check location (square vertex). (g,h) The thickened 1-cycle can also turn at the junction region, with geometric interpretation essentially illustrated in Fig.~\ref{fig:cycle_simplicial}(d). }\label{fig:three_cycle_property}
\end{figure*}

Another important property is that the 0-cocycle $\bar{\bs}^0$ (codeword of $\bar{\C}^T$) has to split into all directions at each check $\mathsf{C}_j$ to satisfy the even-parity condition at the neighboring $\mathsf{B}_i$, which plays the role of a parity check in the transposed code $\bar{\C}^T$.    This is also equivalent to satisfying the 0-coboundary condition at the corresponding hyperedge $e_{h,i}$, i.e., $d \bar{\bs}^0 (e_{h, i})$$=$$0$.  Note that this is consistent with the dual hypergraph description, since in $G_h^*$, the check $\mathsf{C}_j$ plays the role of a dual hyperedge $e_{h,j}^*$.  Therefore, all legs of this dual hyperedge has to be occupied by the dual 1-cycle $\eta^*_1$. 

We then investigate the corresponding 1-cocycle in the thickened code.  It is more convenient to describe these 1-cocycles geometrically with their Poincaré dual cycles on the dual triangulation $\L^*$. 
An example of a 2D cellular (square) complex $\L$ (black) and its dual complex $\L^*$ (green) is illustrated in Fig.~\ref{fig:cocycle_property}(c).

\begin{figure*}[hbt]
\includegraphics[width=1.5\columnwidth]{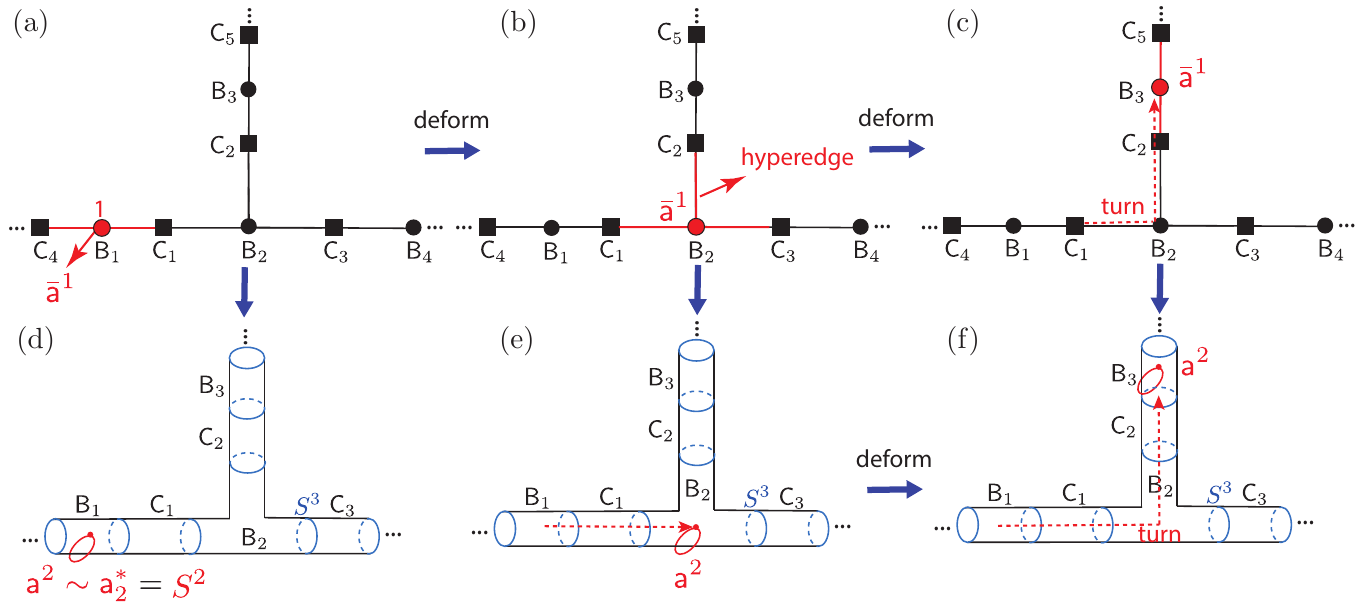}
\caption{(a,b,c) A 1-cocycle representative occupying a single hyperedge in the skeleton hypergraph can turn at a hyperedge (bit). (d,e,f) In the corresponding manifold, the thickened 2-cocycle $S^2$ can also turn in a $\mathsf{B}_i$-``handle" when there are more then two legs. }\label{fig:four_cocycle_property}
\end{figure*}

\begin{figure}[hbt]
\includegraphics[width=1\columnwidth]{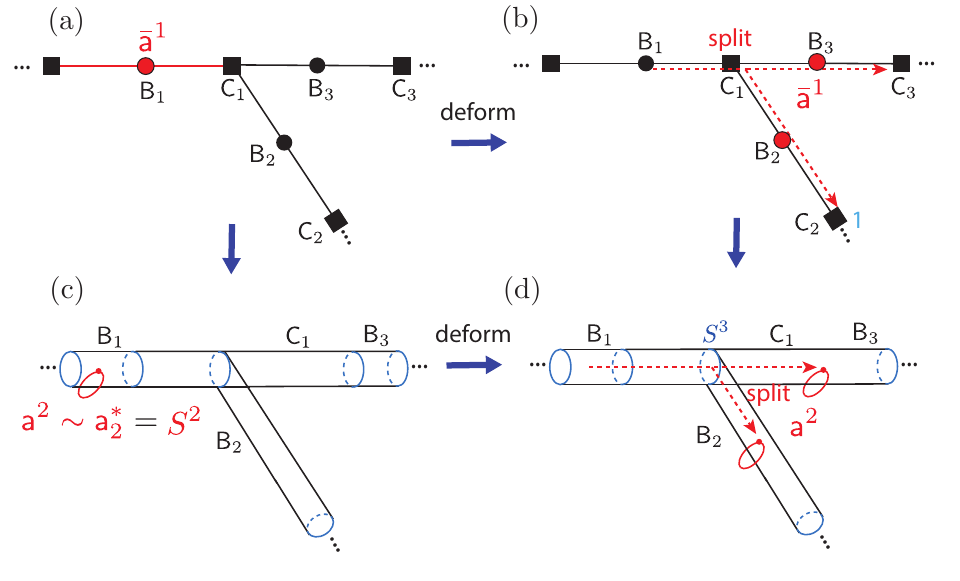}
\caption{(a,b) When deforming a 1-cycle representative on a single hyperedge across the check (square vertex), it has two splits and occupying the rest of the bits (hyperedges) connected to the check. (c, d) In the corresponding manifold, the thickened 2-cocycle also has to split at the junction region.  The geometric interpretation has essentially been illustrated in Fig.~\ref{fig:cocycle_property}(f).  }\label{fig:four_cocycle_interface}
\end{figure}

As shown in Fig.~\ref{fig:cocycle_property}(d), the dual 1-cycle $\eta^*_1$ in the skeleton classical code forms the skeleton of the 1-cocycle $\bs^1$ and its dual 3-cycle $\bs^*_3$, which are the boundary of the tubular neighborhood of $\eta^*_1$ (i.e., locally as $\eta^*_1|_\tau \times D^3$). In any tubular segment  $\tau = S^1 \times I \times S^2$, the 1-cocycle and its dual 3-cycle is supported on the thickened cycle:
\be
\bs^1|_\tau \sim \bs^*_3|_\tau =\eta^*_1|_{\tau} \times S^2=I \times S^2, 
\ee
as illustrated in Fig.~\ref{fig:cocycle_property}(d).
On the $\mathsf{B}_i$ with multiple legs, there can be crossing of $\eta^*_1$ as has been shown in Fig.~\ref{fig:cocycle_property}(a).   Since there has to be an even number of branches being occupied, we can always resolve $\eta^*_1$ using the homology equivalence relation.   As illustrated in Fig.~\ref{fig:cocycle_property}(d), one can either connect the top to the left and the bottom to the right  (solid line), or instead the top to the right and the bootom to the left (dashed line).   The thickened cocycle $\bs^1$ and its dual cycle $\bs^*_3$ can be reconnected the same way. 

Same as its skeleton $\eta^*_1$,  the 1-cocycle $\bs^1$ and the dual 3-cycle $\bs^*_3$ can turn at the $\mathsf{B}_i$-``handle" since locally (in a small tubular segment $\tau$) they are just a 1-cycle thickened along the 2-sphere $S^2$, namely $\eta^*_1|_{\tau} \times S^2$.  They also have to split into all branches at the $\mathsf{C}_j$-``handle" with a junction, as illustrated in Fig.~\ref{fig:cocycle_property}(e). The splitting of 1-cocycle $\bs^1$ and the dual 3-cycle $\bs^*_3$ is illustrated with the doubled picture and 3D projection in  Fig.~\ref{fig:cocycle_property}(f).   The dual 3-cycle $\bs^*_3$ can be interpreted as the world-volume (moving trajectory) of the two sphere $S^2$, which is allowed to split.  In the doubled picture, two disks $D^2$ (orange) are identified along their common $S^1$ boundary to form the 2-sphere $S^2$.  We can see that when we move the two glued disks across the attaching regions (red) of the $\mathsf{B}_1$-``handle", it becomes two larger glued disks.  This is because on the two identified boundaries (green) of the 2-handlebody $H$ which behave like wormholes, a pair of glued disks (in the green regions) can be created together which compensates the difference between the smaller disks on the left and larger disks on the right.  When we further move these two large glued disks towards the right and across the attaching region (red) of the $\mathsf{B}_2$- and $\mathsf{B}_3$-``handles", we see that the membrane in the green regions get annihilated due to the identification, and the larger glued disks split into  two pairs of smaller glued disks in the $\mathsf{B}_2$- and $\mathsf{B}_3$-``handles" repsectively.  The 1-cocycle and its dual 3-cycle $\bs^*_3$ hence splits at the junction.

\item
\textit{1-cycles:}

We now investigate the 0-cycle $\bar{\bs}_0$ in the classical code $\bar{\C}$ which is mapped to the 1-cycle $\bs_1$ in the thickened code, as illustrated in Fig.~\ref{fig:three_cycle_property}.    In contrast to $\bar{\as}_1$ and $\bar{\as}^0$ which are codeword of $\bar{\C}=\bar{\C}^T=\text{Ker}(\Hs^T\Hs)$ and hence have minimum size $\Omega(n)$,  $\bar{\bs}_0$ can have only $O(1)$ size.  The corresponding 1-cycle $\bs_1$ can hence also just have $O(1)$ minimum size.   Since $\bar{\bs}_0 \in {H}_0(G_h)= \text{Ker}({\partial}_0)/\text{Img}({\partial}_1) $ and $\partial_0$ maps all 0-chain to 0,  any single vertex $v \in G_h$ is a valid 0-cycle.  Due to the isomorphism $H_0(G_h; \ZZ_2) \cong H^0(G_h; \ZZ_2)$ from the universal coefficient theorem, we have the $0^\text{th}$ Betti number equaling the linear dimension of the transposed code $\bar{\C}^T$, i.e.,
\begin{align}
\text{dim}(H_0(G_h; \ZZ_2))=\text{dim}(H^0(G_h; \ZZ_2))=\bar{k}^T=\bar{k}=\Theta(n).
\end{align}
Therefore, there are $\Theta(n)$ equivalence classes of $\bar{b}_0$.  Since $\bar{b}_0$ is mapped to $\bs_1$, there are also $\Theta(n)$ equivalence classes of $\bs_1$. 

Distinct from the case of $\bar{\as}_1$ and $\bar{\as}^0$  which only has a unique representative,  in each equivalence class of $\bar{\bs}_0$, there are equivalent representatives differing by a 1-boundary, i.e., 
\be
\bar{\bs}_0=\bar{\bs}_0+{\partial}_1 {\chi}_1
\ee
(${\chi}_1$ is any 1-chain),  due to the definition of $\bar{H}_0(G_h)$ above which mods out $\text{Img}({\partial}_1)$.   As illustrated in Fig.~\ref{fig:three_cycle_property}(a,b) a 0-cycle $\bar{\bs}_0$ (blue) on the check $\mathsf{C}_1$ (vertex) is deformed to an equivalent representative occupying two vertices $\mathsf{C}_2$ and $\mathsf{C}_3$ by adding a 1-boundary $\bar{\partial}_1 e_{h, 1}=v_1 + v_2 + v_3$, where $e_{h, 1}$ represents the hyperedge associated with bit $\mathsf{B}_1$, and $v_j$ the vertex corresponds to check $\mathsf{C}_j$.  The deformation trajectory needs to split to the complementary branches at a bit $\mathsf{B}_i$ (hyperedge).   When mapped to the manifold $\M^4$,  the corresponding thickened 1-cycle
has the following equivalence relation up to a adding a 2-boundary:
\be
{\bs}_1={\bs}_1+ \partial_2 \chi_2
\ee
($\chi_2$ is any 2-chain).  Therefore, the 1-cycle 
$\mathsf{b}_1 = S^1$ can be deformed to a larger $\mathsf{b}_1 = S^1$ going into all the remaining legs when the $\mathsf{B}_1$ bifurcate, and then being split into two $S^3$, as shown in Fig.~\ref{fig:three_cycle_property}(c,d).   This splitting is due to the recoupling relation in homology, which has been illustrated in the lower panel of (d).  
Since the 2-cycle $\as_2$ discussed above can be considered as the worldsheet (moving trajectory) of the 1-cycle $\mathsf{b}_1 = S^1$ here, the above splitting picture of $\mathsf{b}_1$ is consistent with the splitting of $\as_2$ into all legs in the $\mathsf{B}_1$-``handle" as illustrated in Fig.~\ref{fig:cycle_simplicial}(b).  

Moreover, at the check $\mathsf{C_1}$, the 0-cycle $\bar{\bs}_0$ can be deformed and turn to any branch as shown in  Fig.~\ref{fig:three_cycle_property}(e,f).   In the illustrated example, the deformation is achieved by adding a 1-boundary ${\partial}_1 e_2=v_1+v_2$,  where $e_2$ is the edge corresponding to bit $\mathsf{B}_2$, and $v_1$ and $v_2$ are the vertices associated with checks $\mathsf{C}_1$ and $\mathsf{C}_2$. The  0-cycle $\bar{\bs}_0$ hence turns to the lower branch.  Similarly, the  0-cycle can also be deformed to the right branch.    When mapped to the manifold $\M^4$, the 1-cycle $\mathsf{b}_1 = S^1$ can be deformed and turn at the junction of the $\C_j$-``handle" to any branch, as illustrated in Fig.~\ref{fig:three_cycle_property}(g,h). Since the 2-cycle $\as_2=S^1$ can be considered as the worldsheet (moving trajectory) of the 1-cycle $\mathsf{b}_1$  as mentioned above,  we can understand this turning with the more  concrete 3D projection picture in Fig.~\ref{fig:cycle_simplicial}(d), where the worldsheet of $S^1$ turns at the junction.

\item
\textit{2-cocycles:} 

We now consider the 1-cocycle $\bar{\as}^1$ in the classical code $\bar{\C}$ which is mapped to the 2-cocycle $\as^2$ in the thickened code defined on $\M^4$, as illustrated in Fig.~\ref{fig:four_cocycle_property} and \ref{fig:four_cocycle_interface}. Similar to the case of 0-cycle $\bar{\bs}_0$ and the corresponding 1-cycle $\bs_1$, the 1-cocycle $\bar{\as}^1$ and the corresponding 2-cocycle $\as^2$ also have $O(1)$ minimum size and $\Theta(n)$ equivalence classes.

Note that since $\bar{\as}^1$$=$$H^1(G_h; \ZZ_2)$$=$$\text{Ker}({d}^1)/\text{Img}({d}^0)$ and $\bar{d}^1$ acts trivially, any single hyperedge in $G_h$ is a valid 1-cocycle $\bar{\as}^1=\bar{e}_{h,i}$ ($\bar{e}_{h,i}$ is the indicator 1-cochain at hyperedge ${e}_{h,i}$), as shown in Fig.~\ref{fig:four_cocycle_property}(a). There are equiavlent representatives differing by the 0-coboundary:
\be
\bar{\as}^1 = \bar{\as}^1 + d^0 {\chi}^0,
\ee
where ${\chi}^0$ is any 0-cocycle.  As illustrated in Fig.~\ref{fig:four_cocycle_property}(a,b), a representative of $\bar{\as}^1$ occupying the bit $\mathsf{B}_1$ (edge $e_1$) can be moved to the neighboring bit $\mathsf{B}_2$ (hyperedge $e_{h,2}$) by adding a 0-coboundary of the indicator cochain $\bar{v}_1$ at $\mathsf{C}_1$ (vertex):  $d^0 \bar{v}_1= \bar{e}_{1}+\bar{e}_{h,2}$.  Furthermore, $\bar{\as}^1$ can turn at $\mathsf{B}_2$ to the upper branch and moves to $\mathsf{B}_3$ by adding a 0-coboundary of the indicator cochain  at $\bar{v}_2$ at $\mathsf{C}_2$: $d^0 \bar{v}_2= \bar{e}_{h,2}+\bar{e}_{3}$, as illustrated in Fig.~\ref{fig:four_cocycle_property}(c). When mapped to the 4-manifold $\M^4$, the corresponding 2-cocycle and its dual 2-cycle $\as^2 \sim \as^*_2 = S^2$ can also turn at the $\mathsf{B}_2$-``handle" since it can be freely deformed to any point on the $\mathsf{B}_2$-``handle",  as illustrated in Fig.~\ref{fig:four_cocycle_property}(d-f).   Since the 1-cocycle $\bs^1$ and its dual 3-cycle $\bs^*_3$ discussed previously can be considered as the worldsheet (moving trajectory) of the 2-cocycle and its dual 2-cycle  $\as^2 \sim \as^*_2 = S^2$, this turning is consistent with the turning of the 1-cocycle $\bs^1 \sim \bs^*_3$ illustrated in Fig.~\ref{fig:cocycle_property}(e).

Moreover, when moving through a check $\mathsf{C}_1$, the 1-cocycle $\bar{\as}^1$ needs to be split into all the remaining branches by adding a 0-coboundary of an indicator 0-cochain at $\mathsf{C}_1$ (vertex), i.e., $d^0 \bar{v}_1$, as illustrated in Fig.~\ref{fig:four_cocycle_interface}(a,b).   The corresponding 2-cocycle and its dual 2-cycle $\as^2 \sim \as^*_2=S^2$ also splits at the junction in the $\mathsf{C}_1$-``handle" into all remaining branches, as shown in Fig.~\ref{fig:four_cocycle_interface}(c,d).   As mentioned above, the worldsheet  of 2-cocycle and its dual 2-cycle $\as^2 \sim \as^*_2 = S^2$ corresponds to the 1-cocycle and its dual 3-cycle $\bs^1 \sim \bs^*_3$ discussed before.  Hence the splitting of  $S^2$ has already been illustrated concretely in the 3D projection picture in Fig.~\ref{fig:cocycle_property}(f).

\item
\textit{0-cocycles, 4-cycles and 0-cycles:}

Besides the doubling of existing cyclces and cocycles in the skeleton classical code due to the introduction of Poincar\'e duality, the construction of the 4-manifold $\M^4$ out of the skeleton classical code also introduces new emergent cycles and cocycles which were not present in the classical code.
One of them is the 0-cocycle $\cs^0$ and its Poincar\'e dual 4-cycle $\cs^*_4$ which are supported on the entire manifold $\M^4$: the 0-cocycle $\cs^0$ occupies all the vertices in $\M^4$, while the 4-cycle $\cs^*_4$ occupies all the 4-simplexes in $\M^4$.

The 0-cocycle $\cs^0$ has a conjugate 0-cycle $\cs_0$.   Any 0-cycle occupying a single vertex $v$ in $\M^4$ is a valid representative of $\cs_0$.  According to Eq.~\eqref{eq:conjugate_relation}, the conjugate pair overlap at a single vertex $v$ as:
\be
\int_{\mathsf{c}_0} \mathsf{c}^0 = |\mathsf{c}_0 \cap \mathsf{c}^0| = 1,
\ee
which also leads to the following non-trivial cup product between $\cs^0$ and the Poincar\'e dual of $\cs_0$, i.e., $\cs^{*4}$:
\be
\int_{\M^4} \cs^0 \cup \cs^{*4} = |\cs^*_4 \cap \cs_0|=1,
\ee
where the first equality shows its equivalence to the $\ZZ_2$ intersection between the Poincar\'e dual cycles.

\item
\textit{Spurious cycles and cocycles:}

There also exist spurious cycles in the manifold which were not present in the skeleton classical code. For each class of 1-cycle $\bs_1=S^1 \in H_1(\L; \ZZ_2)$, one can thicken it along the $S^2$ direction and get $\mathsf{f}_3 =\bs_1 \times S^2=S^1 \times S^2$ which also has $O(1)$ size.   This means there also exist $\Theta(n)$ equivalence classes in the 3-cycle basis $\{\mathsf{f}_3\}$.   

Due to the isomorphism in Eq.~\eqref{eq:combined_isomorphism},  we also have the dual 1-cycle $\mathsf{f}^*_1 \in H_1(\L; \ZZ_2)$ which intersects with the 3-cycle $\mathsf{f}_3$ at a single point, i.e., $|\mathsf{f}_3 \cap \mathsf{f}^*_1|=1$.
The size of $\mathsf{f}^*_1$ is unclear. Nevertheless, since it is a 1-cycle on a 4D simplicial complex, it may have an upper bound of $O(\log n)$ size obtained from the case of random simplicial complex. 

Beside these pairs of dual cycles, there are also their Poincar\'e dual cocycles $\fs^{*1}$ and $\fs^3$, which have the same size and number of equivalence classes.  

\end{enumerate}

\subsection{Thickened 3D hypergraph-product code and its code parameter scaling}

When taking a homological product of three identical copies of classical codes defined on the triangulation of the 4-manifolds constructed above, we get a thickened 3D hypergraph product defined on   
the product simplicial complex $\tilde{\L}=\L \otimes \L' \otimes \L''$ which forms the triangulation of the product manifold $\widetilde{\M}^{12}=\M^4 \times \M^{4'} \times \M^{4''}$. Here, $\L$, $\L'$ and $\L''$ represent the simplicial complex (triangulation) associated with the three factor 4-manifolds $\M^4$, $\M^{4'}$ and $\M^{4''}$.   In particular, we place the qubits on the 4-simplexes of the triangulation $\tilde{\L}$. Therefore, the logical-$Z$ operators are hence associated with the 4-cycles, while the logical-$X$ operator are associated with the 4-cocycles or equivalently their Poincar\'e dual 8-cycles in the dual triangulation $\L^*$. 

Due to the presence of spurious cyclces/cocycles as discussed in the last subsection, the 4-systole and 4-cosystle can end up being small, which can give rise to $O(1)$ code distance.  However, instead of using the conventional subspace code to encode the logical information into the manifold $\widetilde{\M}^{12}$, one can use the more flexible subsystem code idea to encode logical information only into a subset of homology/cohomology classes.  

One can get more intuition from the following example illustrated in Fig.~\ref{fig:sub_system_encoding}.  We start with a torus $T^2$ having 1-systole of size $L$, i.e., $\text{min}(|\alpha_1|)=\text{min}(|\beta_1|)=L$, where $\alpha_1$ and $\beta_1$ represent the longitudinal and meridian 1-cycles.  We then remove two small disks on $T^2$ and glue a small handle on that to form a genus-2 surface $\Sigma_2$, which introduces two additional cycles $\alpha'_1$ and $\beta'_1$ both with $O(1)$ size. We can choose the following homology basis for this surface: $\{\alpha_1, \beta_1, \alpha'_1, \beta'_1\}$. This gives rise to four pairs of conjugate logical operators encoding four logical qubits as illustrated in Fig.~\ref{fig:sub_system_encoding}, where the logical-$X$ and logical-$Z$ operators of the same logical qubit are defined on a pair of dual basis cycles intersecting with each other, such as $\lo{X}_{\alpha_1}$ and $\lo{Z}_{\beta_1}$ where $|\alpha_1 \cap \beta_1|=1$. Since the shortest 1-cycle has only $O(1)$ size, the 1-systole is hence $O(1)$.  If one defines a conventional subspace code with the code space being $\C= 2^{|H_1(\Sigma_2;\ZZ_2)|}$, the code distance is only $d=O(1)$.   However, we can treat the logical qubits associated with short cycles as gauge qubits (which do not store any information), and choose only a subset of basis cycles to encode the information, e.g., $\{\alpha_1, \beta_1\}$ composed of only the long cycles and encoding two logical qubits corresponding to the dual pair of logical operators:  $\lo{X}_{\alpha_1}$ and $\lo{Z}_{\beta_1}$,   $\lo{Z}_{\alpha_1}$ and $\lo{X}_{\beta_1}$.    Note that any $Z$ or $X$ errors along the two short cycles $\alpha'_1$ and $\beta'_1$ does not intersect with the long cycles $\alpha_1$ and $\beta_1$ and hence cannot induce a logical error in the subsystem code. Therefore, the distance of the subsystem code is still determined by the shortest length of the large cycles, i.e., $d=L$.

\begin{figure}[hbt]
\includegraphics[width=1\columnwidth]{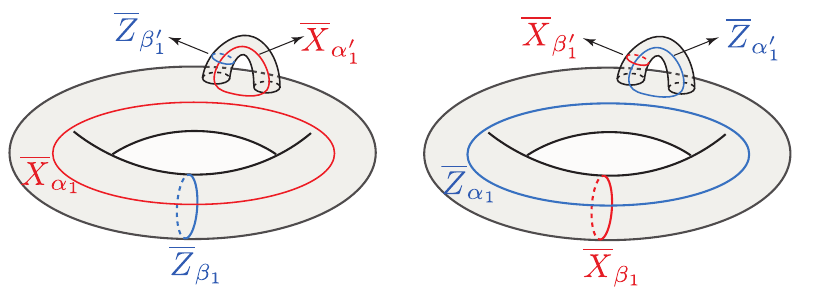}
\caption{Illustration of the subsystem encoding of the homological code defined on a manifold. With a chosen homology basis, two logical qubits are encoded into a dual pair of mutually intersecting cycles $\alpha_1$ and $\beta_1$ with minimum size $O(L)$, while the other two logical qubits are encoded into a dual pair of $\alpha'_1$ and $\beta'_1$ cycles with minimum size $O(1)$. One can then set the logical qubits supported on the short cycles as gauge qubits and hence do not store quantum information in them.   For the rest of the two logical qubits, there subsystem code distance is still large, i.e., $O(L)$, despite the systole in this manifold is only $O(1)$ due to the presence of short cycles.}\label{fig:sub_system_encoding}
\end{figure}

For general situations, we have the following lemma:
\begin{lemma}\label{lemma:subsystem}
    For a homological quantum code defined on the triangulation of a  $k$-manifold $\M^k$, one can define a subsystem code by associating the logical-$Z$ operators with a subset of an $i^\text{th}$ homology basis $\{\alpha_i\}$ and the conjugate logical-$X$  operators on the dual subset of $(k-i)^\text{th}$ homology basis $\{\beta^*_{k-i}\}$ satisfying the intersection relation $|\alpha_i \cap \beta^*_{k-i}|= \delta_{\alpha, \beta}$.  The distance of the subsystem code is hence $d=\text{min}(\text{min}\{|\alpha_i|\}, \text{min}\{|\beta^*_{k-i}|\})$.
\end{lemma}
\begin{proof}
Due to the intersection relation $|\alpha_i \cap \beta^*_{k-i}|= \delta_{\alpha, \beta}$, any logical-$Z$ operator supported on the basis cycle ${\alpha_i}$ only anticommutes with the logical-$X$ operator supported on the dual cycle $\alpha^*_{k-i}$, namely 
\be
\lo{Z}_{\alpha_i} \lo{X}_{\alpha^*_{k-i}}=-\lo{X}_{\alpha^*_{k-i}} \lo{Z}_{\alpha_i}.
\ee
Therefore, only the $X$-errors wrapped around the dual cycle $\alpha^*_{k-i}$ can flip the eigenvalue of the logical operator $\lo{Z}_{\alpha_i}$.   Similarly, on the $Z$-errors wrapped around the cycle $\alpha_i$ can flip the eigenvalue of the logical operator  $\lo{X}_{\alpha^*_{k-i}}$.

Therefore, the $Z$-distance of this subsystem code is the smallest size of all basis $i$-cycles in the subset $\{\alpha_i\}$,  i.e., $d_Z=\text{min}\{|\alpha_i|\}$ (which takes the minimum size among all representatives of the basis $i$-cycles in the subset $\{\alpha_i\}$), while the $X$-distance is the smallest size of all basis $(k-i)$-cycles in the subset $\{\beta^*_{k-i}\}$, i.e., $d_X=\text{min}\{|\beta^*_{k-i}|\}$. We hence obtain the overall code distance as
\be
d= \text{min}(d_X, d_Z)=\text{min}(\text{min}\{|\alpha_i|\}, \text{min}\{|\beta^*_{k-i}|\}).
\ee
\end{proof}

Based on the above lemma and Lemma \ref{lemma:thickened_code}, we can reach the following theorem:
\begin{theorem}\label{theorem:3Dhypergraph}
Given a skeleton classical code $\bar{\C}=\text{Ker}(\Hs^T \Hs)$, where $\Hs$ is a full-rank parity check matrix of a good classical LDPC code with prameters $[n,\Theta(n),\Omega(n)]$, the homological product of three identical copies of 4-manifolds $\M^4$ obtained from the handle construction with the input of $\bar{\C}$ gives rise to an $[[N, \Theta(N), \Omega(N^{1/3})]]$ qLDPC code with constant stabilizer weight $\omega=O(1)$. 
\end{theorem}
\begin{proof}
According to the Künneth formula Eq.~\eqref{eq:Kunneth_general}, the homology/cohomology groups of the triple product manifold $\widetilde{\M}^{12}=\M^4 \times \M'^4 \times \M''^4$ can be decomposed as a direct sum of the homology/cohomology groups of the factor manifold:
\begin{align}\label{eq:Kunneth_manifold}
\tilde{H}_4=& (H_2 \otimes H'_1 \otimes H''_1) \oplus (H_2 \otimes H'_2 \otimes H''_0) \cr 
&\oplus (H_3 \otimes H'_1 \otimes H''_0) \oplus \text{perm.}\cr
\tilde{H}^4=& (H^2 \otimes H'^1 \otimes H''^1) \oplus (H^2 \otimes H'^2 \otimes H''^0) \cr 
&\oplus (H^3 \otimes H'^1 \otimes H''^0) \oplus \text{perm.},
\end{align}
where ``perm." stands for permutations.   From Poincar\'e duality, the 4th cohomology group is isomorphic to the 8th homology group, i.e.,
\begin{align}
\tilde{H}^4 \cong \tilde{H}_8=& (H_2 \otimes H'_3 \otimes H''_3) \oplus (H_2 \otimes H'_2 \otimes H''_4) \cr 
&\oplus (H_1 \otimes H'_3 \otimes H''_4) \oplus \text{perm.},
\end{align}

From the above expression, we see that the contributions to the homology/cohomology can be divided into three groups $(2,1,1)$, $(2,2,0)$ and $(3,1,0)$ (and their permutations).  Since the 1-cycle and 0-cycle of the skeleton classical code $\bar{\C}$ (corresponding to chain complex $X$) is mapped to 2-cycle and 1-cycle in $\M^4$, we have the $\ZZ_2$ Betti number:
\begin{align}
\tilde{b}_2=&\text{dim}(H_2(\M^4; \ZZ_2))=b_1=\text{dim}(H_1(X; \ZZ_2))=\Theta(n), \cr
\tilde{b}_1=&\text{dim}(H_1(\M^4; \ZZ_2))=b_0=\text{dim}(H_0(X; \ZZ_2))=\Theta(n),
\end{align}
where we have used ${b}_1={b}_0=\bar{k}=\Theta(n)$ from Lemma~\ref{lemma:classical_code}. 
We then apply the Künneth theorem to the Betti number and obtain the total number of logical qubits $K$:
\begin{align}
K=&\tilde{b}_4 = b_2\cdot b'_1\cdot b''_1 + b_2\cdot b'_2 \cdot b''_0 + b_3 \cdot b'_1 \cdot b''_0 +\text{perm.} \cr
=& \Theta(n)\cdot \Theta(n) \cdot \Theta(n) + \Theta(n)\cdot \Theta(n) \cdot 1 + \Theta(n)\cdot \Theta(n) \cdot 1  \cr
=& \Theta(n^3)= \Theta(N).
\end{align}
Here, $N=\Theta (n) \cdot \Theta (n) \cdot \Theta (n)=\Theta (n^3)$ are the total number of qubits (4-simplices) in the qLDPC codes defined on the triangulation of $\tilde{M}^{12}$, which is given by condition~2 in Lemma \ref{lemma:thickened_code}.

Following Lemma \ref{lemma:subsystem}, now we select a subset of 4-cycle basis and its dual 8-cycle basis to form a subsystem code. In particular, we compose these basis cycles in $\tilde{\M}^{12}$ with all the basis cycles of $\M^4$ introduced in Sec.~\ref{sec:cycle_mapping} except those spurious cycles/cocycles. According to the Künneth formula in Eq.~\eqref{eq:Kunneth_manifold} the basis 4-cycles in the chosen subset can be divided into three groups:   
\begin{align}
\mathsf{A}_4=&\as_2 \otimes \bs'_1 \otimes \bs''_1  \quad \text{(and perm.)}   \cr
\mathsf{B}_4=&\as_2 \otimes \as'_2 \otimes \cs''_0  \quad \text{(and perm.)}   \cr
\mathsf{C}_4=&\bs^*_3 \otimes \bs'_1 \otimes \cs''_0  \quad \text{(and perm.)}. 
\end{align}
Their conjugate 4-cocyles are
\begin{align}\label{eq:conjugate_basis_cycle}
\mathsf{A}^4=&\as^2 \otimes {\bs'}^1 \otimes {\bs''}^1  \quad \text{(and perm.)}   \cr
\mathsf{B}^4=&\as^2 \otimes {\as'}^2 \otimes {\cs''}^0  \quad \text{(and perm.)}   \cr
\mathsf{C}^4=&\bs^{*3} \otimes {\bs'}^1 \otimes {\cs''}^0  \quad \text{(and perm.)}.  
\end{align}
The corresponding dual basis 8-cycles are given by Poincar\'e duality as
\begin{align}\label{eq:dual_basis_cycle}
\mathsf{A}^*_8=& \as^*_2 \otimes {\bs^*_3}' \otimes {\bs^*_3}''  \quad \text{(and perm.)}   \cr
\mathsf{B}^*_8=& \as^*_2 \otimes {\as^*_2}' \otimes {\cs^*_4}''  \quad \text{(and perm.)}   \cr
\mathsf{C}^*_8=& \bs_1 \otimes {\bs^*_3}' \otimes {\cs^*_4}''  \quad \text{(and perm.)}.  
\end{align}
For any basis cycle in the first group $\{\mathsf{A}_4\}$, we have 
\begin{align}
\text{min}(|\mathsf{A}_4|)=& \text{min}(|\as_2|) \cdot \text{min}(| \bs'_1|) \cdot \text{min}(|\bs''_1|) \cr
\ge& \text{min}(|\bar{\as}_1|) \cdot \text{min}(| \bar{\bs}'_0|) \cdot \text{min}(|\bar{\bs}''_0|)  \cr
=& \Omega(n)\cdot \Omega(1) \cdot \Omega(1) \cr 
=& \Omega(n)=\Omega(N^{\frac{1}{3}}).
\end{align}
Here we have used the fact that the 2-cycle $\as_2$ is a thickened version of the 1-cycle skeleton $\bar{\as}_1=\xi_1$ in the classical code $\bar{\C}$ (see Sec.~\ref{sec:cycle_mapping} ), which has a lower bound in size $\text{min}|\as_2| \ge \text{min}|\bar{\as}_1|= \Omega(n)$. This can be proven in the same way as the proof of Theorem \ref{theorem:manifold_scaling}.  Moreover, the 1-cycle $\bs_1$ is a thickened version of the 0-cycle $\bar{\bs}_0$ in the skeleton classical code, and has size $O(1)$.
For any basis cycle in the second group, we have 
\begin{align}
\text{min}(|\mathsf{B}_4|)=& \text{min}(|\as_2|) \cdot \text{min}(| \as'_2|) \cdot \text{min}(|\cs''_0|) \cr
\ge& \text{min}(|\bar{\as}_1|) \cdot \text{min}(| \bar{\as}'_1|) \cdot 1  \cr
=& \Omega(n)\cdot \Omega(n)  \cr 
=& \Omega(n^2)=\Omega(N^{\frac{2}{3}}),
\end{align}
where we have used the fact that $\cs_0''$ is only a single vertex in the triangulation of 4-manifold ${\M^4}''$ and hence $|\cs_0''|$=1.
For any basis cycle in the third group, we have 
\begin{align}
\text{min}(|\mathsf{C}_4|)=& \text{min}(|\bs^*_3|) \cdot \text{min}(| \bs'_1|) \cdot \text{min}(|\cs''_0|) \cr
=& \text{min}(|\bs^1|) \cdot \text{min}(| \bs'_1|) \cdot 1 \cr
\ge& \text{min}(|\bar{\bs}^0|) \cdot \text{min}(| \bar{\bs}'_0|)   \cr
=& \Omega(n)\cdot \Omega(1) \cr 
=& \Omega(n)=\Omega(N^{\frac{1}{3}}),
\end{align}
where we have used the fact that $\bs_3^*$ is the Poincar\'e dual of $\bs^1$ and hence has the same size $\text{min}(|\bs^*_3|)=\text{min}(|\bs_1|)$.
Now combining the results from the above three groups, we obtain the logical-$Z$ distance according to Lemma \ref{lemma:subsystem}:
\be
d_Z= \text{min}(\{|\textsf{A}_4|\}, \{|\textsf{B}_4|\}, \{|\textsf{C}_4|\}) = \Omega(N^{\frac{1}{3}}).
\ee

We then consider the size of the dual basis 8-cycles in Eq.~\eqref{eq:dual_basis_cycle}, or equivalently that of the conjugate 4-cocycles in Eq.~\eqref{eq:conjugate_basis_cycle}.  For any basis cocycle in the first group, we have 
\begin{align}
\text{min}(|\mathsf{A}^*_8|)=&\text{min}(|\mathsf{A}^4|) = \text{min}(|\as^2|) \cdot \text{min}(| {\bs'}^1 |) \cdot \text{min}(|{\bs''}^1|) \cr
\ge& \text{min}(|\bar{\as}^1|) \cdot \text{min}(| {{\overline{\bs}}'}^0|) \cdot \text{min}(|{\overline{\bs}''}^0|)  \cr
=& \Omega(1)\cdot \Omega(n) \cdot \Omega(n) \cr 
=& \Omega(n^2)=\Omega(N^{\frac{2}{3}}),
\end{align}
where we have used fact that $\as^2$ and $\bs^1$ are  thickened versions of the $O(1)$ size skeleton 1-cocycle $\bar{\as}^1$ and the $\Omega(n)$-size skeleton 0-cocycle $\bar{\bs}^0$ in the classical code $\bar{\C}$ respectively.
For any basis cocycle in  the second group, we have
\begin{align}
\text{min}(|\mathsf{B}_8|)=\text{min}(|\mathsf{B}^4|)=& \text{min}(|\as^2|) \cdot \text{min}(| \as'^2|) \cdot \text{min}(|\cs''^0|) \cr
\ge& \text{min}(|\bar{\as}^1|) \cdot \text{min}(| \bar{\as}'^1|) \cdot \text{min}(|\cs''^0|)  \cr
=& \Omega(1)\cdot \Omega(1) \cdot \Omega(n)  \cr 
=& \Omega(n)=\Omega(N^{\frac{1}{3}}),
\end{align}
where we have used the fact that $\cs''^0 \sim {\cs^*}''_4$ is supported on the entire 4-manifold $\M^4$ and hence has the size $|\cs''^0| = |{\cs^*}''_4|= \Omega(n)$.
For any basis cocycle in the third group, we have
\begin{align}
\text{min}(|\mathsf{C}^*_8|)=& \text{min}(|\mathsf{C}^4|)= \text{min}(|\bs^{*3}|) \cdot \text{min}(| \bs'^1|) \cdot \text{min}(|\cs''^0|) \cr
=& \text{min}(|\bs_{1}|) \cdot \text{min}(| \bs'^1|) \cdot \text{min}(|\cs''^0|) \cr
\ge& \text{min}(|\bar{\bs}_0|) \cdot \text{min}(| \bar{\bs}'^0|) \cdot \text{min}(|\cs''^0|) \cr
=& \Omega(1)\cdot \Omega(n) \cdot \Omega(n) \cr 
=& \Omega(n^2)=\Omega(N^{\frac{2}{3}}),
\end{align}
where we have used the fact that the 3-cocycle and its dual 1-cycle $\bs^{*3} \sim \bs_1$ is a thickened version of the $O(1)$-size skeleton 0-cycle $\bar{\bs}_0$ in the classical code $\bar{\C}$.
Combining the results from the above three groups, we obtain the logical-$X$ distance according to Lemma \ref{lemma:subsystem}:
\be
d_X= \text{min}(\{|\textsf{A}^*_8|\}, \{|\textsf{B}^*_8|\}, \{|\textsf{C}^*_8|\}) = \Omega(N^{\frac{1}{3}}).
\ee
We then get the overall distance of the code as
\be
d=\text{min}(d_X, d_Z)=\Omega(N^{\frac{1}{3}}).
\ee
Finally the constant stabilizer weight $\omega = O(1)$ is equivalent to the bounded local geometry of the product manifold $\tilde{\M}^{12}=\M^4 \times \M'^4 \times \M''^4$, which is in turn given by the bounded local geometry of the factor manifolds $\M^4$,  $\M'^4$ and  $\M''^4$ according to Lemma \ref{lemma:thickened_code}.

\end{proof}

\subsection{Triple cup product and logical non-Clifford gates on the thickened 3D hypergraph product codes}

We can now evaluate cup products in three copies of identical codes defined on manifold $\tilde{\M}^{12}$ and obtain the following theorem:
\begin{theorem}\label{theorem:triple_intersection_12D}
There exist a family of thickened 3D hypergraph product codes $\tilde{\C}$ defined on the triangulation of a 12-manifold  $\tilde{\M}^{12}$ with rate $K=\Theta(N)$, subsystem-code distance $D=\Omega(N^{1/3})$ and constant stabilizer weight $w=O(1)$,  such that a constant-depth circuit implementing the cohomology operation of a triple cup product on three identical copies of $\tilde{\C}$ give rise to $\Theta(N)$ non-Clifford logical CCZ gates. 
\end{theorem}

\begin{proof}
We consider the case of the logical CCZ gate, where a triple cup product of operator-valued 4-cochains from three copies of thickened 3D hypergraph product codes are summed over the product simplicial complex $\tilde{\L}=\L \otimes \L' \otimes \L''$ forming the triangulation of the product manifold $\widetilde{\M}^{12}=\M^4 \times \M^{4'} \times \M^{4''}$. According to Lemma \ref{lemma_gate_3} and Eq.~\eqref{eq:logical_CCZ_higher},  we have the following unitary implementing the logical gate:
\begin{align}\label{eq:unitary_12D_simplicial}
&U= (-1)^{\int_{\widetilde{\M}^{12}} a^{4}_{(1)} \cup a^{4}_{(2)} \cup a^{4}_{(3)}}    \cr
=& \prod_{\alpha^{4}, \beta^{4}, \gamma^{4}} \overline{\text{CCZ}}[(\alpha^{4}; 1), (\beta^{4}; 2),(\gamma^{4}; 3)]^{\int_{\widetilde{\M}^{12}} {\alpha^{4} \cup \beta^{4} \cup \gamma^{4}}}, \cr
\end{align}
where $a^{4}_{(i)}$ represents an operator-valued 4-cocycles in the $i^\text{th}$ copy of thickened 3D hypergraph product code, and $\alpha^{4}, \beta^{4}, \gamma^{4}$ are 4-cocycles from a cohomology basis $\{\alpha^{4}\}$, $\{\beta^{4}\}$ and $\{\gamma^{4}\}$  for copy 1, 2 and 3 respectively.  

Now it is clear that $U$ is a logical gate since it maps the code back to itself, as proven in Sec.~\ref{sec:gates_simplicial_complex}. Nevertheless, for the logical gate to be non-trivial, i.e., not a logical identity, we need to make sure the triple cup product sum  $\int_{\widetilde{\M}^{12}} {\alpha^{4} \cup \beta^{4} \cup \gamma^{4}}$ in the exponent of Eq.~\eqref{eq:unitary_12D_simplicial} evaluates non-trivially. There are two types of choices of triplets of cocycles satisfy this condition.   

We first consider the following cohomology classes using the Künneth theorem:
\begin{align}\label{eq:cocycle_class1}
\alpha^{4} =& \as^2 \otimes \cs'^0 \otimes {{\as^{*}}''}^2    \cr
\beta^{4} =& \as^{*2} \otimes \as'^2  \otimes \cs''^0   \cr
\gamma^{4} =& \cs^0 \otimes {{\as^{*}}'}^2 \otimes \as''^2 .  
\end{align}
Here, $\as^2$, $\as'^2 $ and $\as''^2$ are three arbitrary 2-cocycles from the 2nd cohomology basis $\{\as^2\}$, $\{\as'^2\}$ and $\{\as''^2\}$ in the factor 4-manifolds $\M^{4}$, $\M'^{4}$ and $\M''^{4}$. Due to Poincar\'e duality, the basis 2-cocyle $\as^2$ has a unique dual basis 2-cocycle denoted by $\as^{*2}$, satisfying the intersection condition:
\be
\int_{\M^4} \as^2 \cup \as^{*2} \equiv \int_{\L} \as^2 \cup \as^{*2} = |\as^*_2 \cap \as_2|=1,
\ee
where the sum over the 4-manifold $\M^4$ corresponds to the sum over its triangulation $\L$ in the discrete description.
Similarly, ${{\as^{*}}'}^2$ and ${{\as^{*}}''}^2$ are the unique dual basis 2-cocycles of $\as^2$ and $\as'^2 $ respectively.   In addition, we have used the unique 0-cocyle class $\cs^0$,  $\cs'^0$ and $\cs''^0$ from each 4-manifold whose Poincar\'e dual $\cs^*_4$, $\cs'^*_4$ and $\cs''^*_4$ are the cocycles enclosing the entire 4-manifolds $\M^4$, $\M'^4$ and $\M''^4$.

Within the factor 4-manifold $\M^4$, there is a non-trivial triple cup product and the triple intersection structure:
\be\label{eq:non-trivial_triple_class}
\as^2 \cup \as^{*2} \cup \cs^0 \neq 0 \in H^4(\M^4; \ZZ_2),
\ee
and
\begin{align}\label{eq:non-trivial_triple}
&\int_{\M^4} \as^2 \cup \as^{*2} \cup \cs^0 \equiv \int_{\L} \as^2 \cup \as^{*2} \cup \cs^0 \cr
&= |\as^*_2 \cap \as_2 \cap \cs^*_4|=1.
\end{align}
The interpretation of the triple intersection is that the Poincar\'e dual pair $\as^*_2$ and $\as_2$ intersect at a single point which in term intersects with the 4-cycle $\cs^*_4$ wrapping around the entire 4-manifold $\M^4$ at a single point.  Similar triple cup product sturcture also occurs in the second and third factor manifold $\M'^4$ and $\M''^4$.

 Using the Künneth theorem for cup product, we can decompose the triple cup product into its tensor component as:
 \begin{align}
 & {\alpha^{4} \cup \beta^{4} \cup \gamma^{4}} \cr
 =& (\as^2 \otimes \cs'^0 \otimes {{\as^{*}}''}^2) \cup (\as^{*2} \otimes \as'^2  \otimes \cs''^0 ) 
\cup ( \cs^0 \otimes {{\as^{*}}'}^2 \otimes \as''^2)  \cr
=&  ( \as^2 \cup \as^{*2} \cup \cs^0) \otimes (\cs'^0 \cup \as'^2  \cup {{\as^*}'}^2 )  \otimes  ({{\as^{*}}''}^2 \cup \cs''^0 \cup \as''^2) \cr
\neq & 0 \in H^{12}(\widetilde{M}^{12}; \ZZ_2).
 \end{align}
The non-triviality of the 12-cocycle $\alpha^{4} \cup \beta^{4} \cup \gamma^{4} \neq 0 $ comes from the fact that each tensor component of the triple cup product in the second equality, e.g., $\as^2 \cup \as^{*2} \cup \cs^0$, is in a non-trivial class  according to Eq.~\eqref{eq:non-trivial_triple_class} for the factor manifold $\M^4$ and the similar expression for the second and third factor manifold $\M'^4$ and $\M''^4$. 
   
Therefore,  we can hence re-express this triple cup product sum in the exponent of Eq.~\eqref{eq:unitary_12D_simplicial} as:
\begin{align}\label{eq:simplicial_Kunneth}
& \int_{\widetilde{\M}^{12}} {\alpha^{4} \cup \beta^{4} \cup \gamma^{4}} \cr
=&  \int_{\widetilde{\M}^{12}} (\as^2 \otimes \cs'^0 \otimes {{\as^{*}}''}^2) \cup (\as^{*2} \otimes \as'^2  \otimes \cs''^0 ) \cr
&\cup ( \cs^0 \otimes {{\as^{*}}'}^2 \otimes \as''^2) \cr
=&  \int_{\M^4} ( \as^2 \cup \as^{*2} \cup \cs^0) \cdot \int_{\M'^4} (\cs'^0 \cup \as'^2  \cup {{\as^*}'}^2 ) \cr
& \cdot \int_{\M''^4} ({{\as^{*}}''}^2 \cup \cs''^0 \cup \as''^2)\cr
=& 1\cdot 1 \cdot 1 =1,
\end{align}
where in the third equality we have used Eq.~\eqref{eq:non-trivial_triple} and its analogy for the other two factor manifolds to give the non-trivial triple cup product within each factor manifold $\M^4$, $\M'^4$ and $\M''^4$.   

There exist another set of cohomology classes which have non-trivial cup products:
\begin{align}\label{eq:cocycle_class2}
\tilde{\alpha}^{4} =& {\bs}^1  \otimes \cs'^0 \otimes {{\bs^*}''}^3    \cr
\tilde{\beta}^{4} =&  \bs^{*3} \otimes {\bs'}^1 \otimes {\cs''}^0   \cr
\tilde{\gamma}^{4} =&  \cs^0 \otimes {{\bs^*}'}^3 \otimes {\bs''}^1,  
\end{align}
where $\bs^1$ and $\bs^{*3}$ are a pair of dual basis cocycles of the manifold in $\M^4$ satisfying the intersection condition:
\be
\int_{\M^4} \bs^1 \cup \bs^{*3}  = |\bs^*_3 \cap \bs_1|=1.
\ee
Similarly, there exist the other two pairs of dual basis cocycles $({\bs'}^1, {{\bs^*}'}^3)$ and $({\bs''}^1, {{\bs^*}''}^3)$ in $\M'^4$ and $\M''^4$ respectively. We hence have the following non-trivial triple cup product sum using the Künneth theorem:
\begin{align}\label{eq:simplicial_Kunneth_second}
& \int_{\widetilde{\M}^{12}} {\tilde{\alpha}^{4} \cup \tilde{\beta}^{4} \cup \tilde{\gamma}^{4}} \cr
=&  \int_{\M^4} ({\bs}^1  \cup \bs^{*3} \cup \cs^0 ) \cdot \int_{\M'^4} (\cs'^0 \cup {\bs'}^1 \cup {{\bs^*}'}^3 ) \cr
& \cdot \int_{\M''^4} ({{\bs^*}''}^3 \cup \cs''^0 \cup {\bs''}^1)\cr
=& 1.
\end{align}

Based on the above two types of non-trivial triple cup product sum, the logical gate $U$ in Eq.~\eqref{eq:unitary_12D_simplicial} implements non-trivial collective logical CCZ gates.   
\end{proof}

\subsection{Magic rate and logical gate structure: counting the number of logical CCZ's from the number of triple intersection points }\label{sec:magic_rate_hypergraph}

As has been shown in Ref.~\cite{zhu2023non} and also above, the logical gate structure is completely determined by the triple intersection (cup product) structure in the underlying manifold $\tilde{\M}^{12}$. In particular, the number of logical CCZs is determined by the number of $\ZZ_2$ triple intersection points in $\tilde{\M}^{12}$.

Now we have two types of contributions to the $\ZZ_2$ triple intersection points given by Eq.~\eqref{eq:cocycle_class1} and Eq.~\eqref{eq:cocycle_class2} respectively.  For the first type  [Eq.~\eqref{eq:cocycle_class1}] and the basis 4-cocycle $\alpha^4=\as^2 \otimes \cs'^0 \otimes {{\as^{*}}''}^2$, we recall that the number of basis cocycle classes in the set $\{\as^2\}$ and $\{{{\as^{*}}''}^2\}$ are both $\Theta(n)$, while there is only a unique class of $\cs'^0$.   Therefore, the number of basis cocycle classes in the subset $\{\alpha^4\}$ is $\Theta(n^2)=\Theta(N^{\frac{2}{3}})$.  Similarly, there are $\Theta(N^{\frac{2}{3}})$ basis cocycle classes in the subset $\{\beta^4\}$ and $\{\gamma^4\}$. For the second type  [Eq.~\eqref{eq:cocycle_class2}], there are also $\Theta(N^{\frac{2}{3}})$ cocycle classes in the subsets $\{\tilde{\alpha}^4\}$, $\{\tilde{\beta}^4\}$, and $\{\tilde{\gamma}^4\}$   from similar reasoning.  In sum, the logical non-Clifford gate $U$ addresses $\Theta(N^{\frac{2}{3}})$ logical qubits in the qLDPC code, with the addressing rate being $r_A=\Theta(N^{\frac{2}{3}})/N=\Theta(1/N^{\frac{1}{3}})$.  

Now we further investigate the logical gate connectivity and total  number of logical CCZ's.  For the first type [Eq.~\eqref{eq:cocycle_class1}] and a given basis cocycle $\alpha^4$, i.e., with fixed choice of $\as^2$ and ${\as^*}''^2$ ($\Theta(N^{2/3})$ possible choices), there are $\Theta(n)=\Theta(N^{\frac{1}{3}})$ possible choices of  basis cocycles  $\beta^4$ since $\as^{*2}$ is fixed by the dual of $\as^2$ due to Poincar\'e duality while there are $\Theta(n)$ possible choices of $\as'^2$.   Now for a given pair of $\alpha^4$ and $\beta^4$,   there is a unique choice of basis cocycle $\gamma^4$ since both ${{\as^{*}}'}^2$ and $\as''^2$ are fixed to be the dual cocycles of $\as'^2$ and  ${\as^*}''^2$.  In other words, each $\alpha^4$ is coupled to $\Theta(n)=\Theta(N^{\frac{1}{3}})$ pairs of $\beta^4$ and $\gamma^4$ via logical CCZ's.  By symmetry, each $\beta^4$ is also coupled to  $\Theta(N^{\frac{1}{3}})$ pairs of $\gamma^4$ and $\alpha^4$, while each $\gamma^4$ is also coupled to  $\Theta(N^{\frac{1}{3}})$ pairs of $\alpha^4$ and $\beta^4$.  For the second type [Eq.~\eqref{eq:cocycle_class2}] involving cocycle classes $\tilde{\alpha}^4$, $\tilde{\beta}^4$, and $\tilde{\gamma}^4$, the logical gate structure is completely the same as those for $\alpha^4$, $\beta^4$ and $\gamma^4$ in the first type based on similar reasoning. 
Now the total number of logical CCZ's implemented by $U$ is hence $n_\text{CCZ}=\Theta(N^{\frac{2}{3}}) \cdot \Theta(N^{\frac{1}{3}}) = \Theta(N)$, which gives rise to the magic rate $r_M=n_\text{CCZ}/N = O(1)$.

We see that one drawback of the current scheme is that it only addresses a fraction of the logical qubits, i.e., $\Theta(N^{\frac{2}{3}})$.   This issue will be cured in the alternative scheme in the Sec.~\ref{sec:good_product} which takes a homological product of a good quantum LDPC code and a good classical LDPC code.  

We also note that although there exists $\Theta(N)$ CCZ gates, one can only extract  
$\Theta(N^{1/3})$ non-overlapping CCZ gates, which hence allows fault-tolerantly preparing $\Theta(N^{1/3})$ non-overlapping CCZ magic states in a single shot by turning off some of the logical qubit (set those in the $\lo{\ket{0}}$  state while others in the $\lo{\ket{+}}$ state).  This type of protocol is called a `\textit{magic state fountain}' \cite{zhu2023non} and will be elaborated in Sec.~\ref{sec:fountain} for the thickened homological product code.
We summarize the above result in the following corollary of Theorem \ref{theorem:triple_intersection_12D}:
\begin{corollary}\label{corollary:hypergraph}
	For the family of thickened 3D hypergraph product codes defined on the triangulation of a 12-manifold from Theorem \ref{theorem:triple_intersection_12D}, one can use a constant-depth circuit to fault-tolerantly prepare $\Theta(N^{1/3})$ non-overlapping logical CCZ magic states with distance $\Omega(N^{1/3})$ in a single shot without disttilation. 
\end{corollary}

\subsection{Alternative construction with the product of 8-manifolds}

From Sec.~\ref{sec:4-manifold} to  \ref{sec:magic_rate_hypergraph} we have adhered to the product construction based on lower-dimensional 4-manifolds.   Here, we present alternative constructions using the 8-manifolds introduced in Sec.~\ref{sec:handle_construction} which has a separation between the logical cycles/cocycles coming from the input classical code and the spurious cycles/cocycles following the original style in Ref.~\cite{freedman:2020_manifold_from_code}.  The purpose of this section is to show that the logical non-Clifford gate properties for this family of constructions can work at various dimensions with flexibility.  

The thickened 3D hypergraph product in this case is built on the triangulation of a 21-manifold:  $\widetilde{\M}^{24}=\M^8 \times \M'^8 \times \M''^8$.   For each factor 8-manifold, there is the following mapping between the cycles and cocycls in the skeleton classical code $\bar{\C}$ and the cycles and cocycles in the 8-manifold:
\be
\bar{\textsf{a}}_1 \rightarrow \mathsf{a}_3, \quad 
\bar{\textsf{b}}_0 \rightarrow \mathsf{b}_2,  \quad
\bar{\textsf{a}}^1 \rightarrow \mathsf{a}^3,  \quad  \bar{\textsf{b}}^0 \rightarrow \mathsf{b}^2.  
\ee
Note that $\mathsf{a}_3$ and $\mathsf{b}^2$ correspond to the codeword of the input skeleton classical code $\bar{\C}$ and its transpose code $\bar{\C}^T$ respectively.   Meanwhile, due to the doubling, there also exist there Poincar\'e duals:
\be
\mathsf{a}_3 \sim {\mathsf{a}^*}^5, \quad \mathsf{b}_2 \sim {\mathsf{b}^*}^6, \quad \mathsf{a}^3 \sim \mathsf{a}^*_5, \quad   \mathsf{b}^2 \sim \mathsf{b}^*_6.
\ee

We now state the following theorem in analogy to Theorem \ref{theorem:triple_intersection_12D}:
\begin{theorem}\label{theorem:triple_intersection_24D}
There exist a family of thickened 3D hypergraph product codes $\tilde{\C}$ defined on the triangulation of a 24-manifold  $\tilde{\M}^{24}$ with rate $K=\Theta(N)$, subsystem-code distance $D=\Omega(N^{1/3})$ and constant stabilizer weight $w=O(1)$,  such that a constant-depth circuit implementing the cohomology operation of a triple cup product on three identical copies of $\tilde{\C}$ give rise to a non-Clifford logical gate. 
\end{theorem}

\begin{proof}
In the product manifold $\widetilde{M}^{24}$, we place the qubits on the 8-cells. The logical-$Z$ and -$X$ operators hence correspond to 8-cycles and 8-cocycles (equivanetly the dual 16-cycles) respectively.
According to the Kunneth formula, we have
\be
\tilde{H}_8= \bigoplus_{i+j+k=8} (H_i \otimes H'_j \otimes H''_k).
\ee
Therefore, the number of logical qubits which equals 8th Betti number scales as 
\be\label{eq:Betti-8}
K=\tilde{b}_8 \ge b_3 \cdot b'_3 \cdot b'_2 =\Theta(n^3)=\Theta(N),
\ee
where we have get a lower bound with a combination of the dimensions of $\bar{\C}$ and $\bar{\C}^T$. 

Now due to the presence of spurious cycles $\ds_1$ and cocyles $\ds^1$, as well as their Poincar\'e dual ${\ds^*}^7$ and $\ds^*_7$, their product can form spurious 8-cycles which could have short support such as $O(1)$. For example, the following 8-cycles can be short:
\be
\non  \ds_1 \otimes {\ds_7^*}' \otimes \cs''_0, \quad  \ds_1 \otimes {\ds_1}' \otimes \bs_6^*  
\ee
where $\cs_0$ is the 0-cycle.  To resolve this issue, we can again use the subsystem encoding following Lemma \ref{lemma:subsystem}.   In particular, we select the following basis cycles  with large size, which can be divided into three groups:   
\begin{align}
\mathsf{A}_8=&\as_3 \otimes \as'_3 \otimes \bs''_2  \quad \text{(and perm.)}   \cr
\mathsf{B}_8=&\as_3 \otimes {\as_5^*}' \otimes \cs''_0  \quad \text{(and perm.)}   \cr
\mathsf{C}_8=&\bs_2 \otimes {\bs_6^*}' \otimes \cs''_0  \quad \text{(and perm.)}, 
\end{align}
where the first group $\mathsf{A}_8$ contributes to the lower bound on the 8th Betti number in Eq.~\eqref{eq:Betti-8}.  Similarly,  we have the following basis cocycles with large size, 
\begin{align}
\mathsf{A}^8=&\as^3 \otimes {\as'}^3 \otimes {\bs''}^2  \quad \text{(and perm.)}   \cr
\mathsf{B}^8=&\as^3 \otimes {{\as^*}'}^5 \otimes {\cs''}^0  \quad \text{(and perm.)}   \cr
\mathsf{C}^8=&\bs^2 \otimes {{\bs^*}'}^6 \otimes {\cs''}^0  \quad \text{(and perm.)}. 
\end{align}

We can bound the $Z$-distance as follows:
\begin{align}
\text{min}(|\mathsf{A}_8|)=& \text{min}(|\as_3|) \cdot \text{min}(| \as'_3|) \cdot \text{min}(|\bs''_2|) \cr
\ge& \text{min}(|\bar{\as}_1|) \cdot \text{min}(| \bar{\as}'_1|) \cdot \text{min}(|\bar{\bs}''_0|)  \cr
=& \Omega(n)\cdot \Omega(n) \cdot \Omega(1) = \Omega(N^{\frac{2}{3}}),
\end{align}

\begin{align}
\text{min}(|\mathsf{B}_8|)=& \text{min}(|\as_3|) \cdot \text{min}(| {\as_5^*}'|) \cdot \text{min}(|\cs''_0|) \cr
=&  \text{min}(|\as_3|) \cdot \text{min}(| \as'^3| \cdot 1 \cr
\ge& \text{min}(|\bar{\as}_1|) \cdot \text{min}(| \bar{\as}'^1|)   \cr
=& \Omega(n) \cdot 1 =\Omega(N^{\frac{1}{3}}),
\end{align}
\begin{align}
\text{min}(|\mathsf{C}_8|)=& \text{min}(|\bs_2|) \cdot \text{min}(| {\bs_6^*}'|) \cdot \text{min}(|\cs''_0|) \cr
=& \text{min}(|\bs_2|) \cdot \text{min}(| \bs'^2|) \cdot 1 \cr
\ge& \text{min}(|\bar{\bs}_0|) \cdot \text{min}(| {{\overline{\bs}}'}^0|)   \cr
=&  \Omega(1) \cdot  \Omega(n) = \Omega(N^{\frac{1}{3}}),
\end{align}
\be
d_Z= \text{min}(\{|\textsf{A}_8|\}, \{|\textsf{B}_8|\}, \{|\textsf{C}_8|\}) = \Omega(N^{\frac{1}{3}}).
\ee
Similarly, we can bound the $X$-distance as follows:
\begin{align}
\text{min}(|\mathsf{A}^8|) =& \text{min}(|\as^3|) \cdot \text{min}(| {\bs'}^3 |) \cdot \text{min}(|{\bs''}^2|) \cr
\ge& \text{min}(|\bar{\as}^1|) \cdot \text{min}(| {{\overline{\as}}'}^1|) \cdot \text{min}(|{\overline{\bs}''}^0|)  \cr
=& \Omega(1)\cdot \Omega(1) \cdot \Omega(n) =\Omega(N^{\frac{2}{3}}),  \cr
\end{align}

\begin{align}
\text{min}(|\mathsf{B}^8|)=& \text{min}(|\as^3|) \cdot \text{min}(| {\as^*}'^5|) \cdot \text{min}(|\cs''^0|) \cr
=& \text{min}(|\as^3|) \cdot \text{min}(| {\as}'_3|) \cdot \text{min}(|\cs''^0|) \cr
\ge& \text{min}(|\bar{\as}^1|) \cdot \text{min}(| \bar{\as}'_1|) \cdot \text{min}(|\cs''^0|)  \cr
=& \Omega(1)\cdot \Omega(n) \cdot \Omega(n) =\Omega(N^{\frac{2}{3}}),
\end{align}

\begin{align}
 \text{min}(|\mathsf{C}^8|) =& \text{min}(|\bs^{2}|) \cdot \text{min}(| {\bs^*}'^6|) \cdot \text{min}(|\cs''^0|) \cr
=& \text{min}(|\bs^{2}|) \cdot \text{min}(| \bs'_2|) \cdot \text{min}(|\cs''^0|) \cr
\ge& \text{min}(|\bar{\bs}^0|) \cdot \text{min}(| \bar{\bs}'_0|) \cdot \text{min}(|\cs''^0|) \cr
=& \Omega(n) \cdot \Omega(1)  \cdot \Omega(n) =\Omega(N^{\frac{2}{3}}).
\end{align}
\be
d_X= \text{min}(\{|\textsf{A}^8|\}, \{|\textsf{B}^8|\}, \{|\textsf{C}^8|\}) = \Omega(N^{\frac{1}{3}}).
\ee
We hence obtain the overall distance bound:
\be
d= \text{min}(d_Z, d_X) = \Omega(N^{\frac{1}{3}}).
\ee

We now investigate the logical CCZ gates.
According to Lemma \ref{lemma_gate_3} and Eq.~\eqref{eq:logical_CCZ_higher},  we have the following unitary implementing the logical gate:
\begin{align}\label{eq:unitary_24D_simplicial}
&U= (-1)^{\int_{\widetilde{\M}^{24}} a^{8}_{(1)} \cup a^{8}_{(2)} \cup a^{8}_{(3)}}    \cr
=& \prod_{\alpha^{8}, \beta^{8}, \gamma^{8}} \overline{\text{CCZ}}[(\alpha^{8}; 1), (\beta^{8}; 2),(\gamma^{8}; 3)]^{\int_{\widetilde{\M}^{24}} {\alpha^{8} \cup \beta^{8} \cup \gamma^{8}}}, \cr
\end{align}
where $a^{8}_{(i)}$ represents an operator-valued 8-cocycles in the $i^\text{th}$ copy of thickened 3D hypergraph product code, and $\alpha^{8}, \beta^{8}, \gamma^{8}$ are 8-cocycles from cohomology basis $\{\alpha^{8}\}$, $\{\beta^{8}\}$ and $\{\gamma^{8}\}$  for copy 1, 2 and 3 respectively.  

We first consider the following basis cocycles:
\begin{align}\label{eq:cocycle_class1}
\alpha^{8} =& {\as^*}^5 \otimes \cs'^0 \otimes \as''^3    \cr
\beta^{8} =& \as^{3} \otimes {\as^*}'^5  \otimes \cs''^0   \cr
\gamma^{8} =& \cs^0 \otimes \as'^3 \otimes {{\as^{*}}''}^5 .  
\end{align}
According to the Kunneth formula for cup product, we can decompose the triple cup product to its tensor component as:
\begin{align}
& {\alpha^{8} \cup \beta^{8} \cup \gamma^{8}} \cr
=& ({\as^*}^5 \otimes \cs'^0 \otimes \as''^3) \cup (\as^{3} \otimes {\as^*}'^5  \otimes \cs''^0) \cup (\cs^0 \otimes \as'^3 \otimes {{\as^{*}}''}^5) \cr
=& ( {\as^*}^5 \cup \as^{3} \cup \cs^0) \otimes (\cs'^0   \cup {{\as^*}'}^5 \cup \as'^3 ) \otimes
 (\as''^3    \cup \cs''^0 \cup {{\as^{*}}''}^5 ) \cr
\neq & 0 \in H^{24}(\widetilde{M}^{24}; \ZZ_2).
\end{align}
We hence have the non-trivial triple intersection:
\begin{align}
& \int_{\widetilde{\M}^{24}} {\alpha^{8} \cup \beta^{8} \cup \gamma^{8}} \cr
=&  \int_{\M^8} ( {\as^*}^5 \cup \as^{3} \cup \cs^0) \cdot \int_{\M'^8} (\cs'^0   \cup {{\as^*}'}^5 \cup \as'^3 ) \cr
& \cdot \int_{\M''^8} (\as''^3    \cup \cs''^0 \cup {{\as^{*}}''}^5 )\cr
=& 1\cdot 1 \cdot 1 =1.
\end{align}
In addition, we also consider the following basis cocycles:
\begin{align}
\tilde{\alpha}^{8} =& \bs^{*6}  \otimes \cs'^0 \otimes {{\bs}''}^2    \cr
\tilde{\beta}^{8} =&  \bs^{2} \otimes  {{\bs^*}'}^6 \otimes {\cs''}^0   \cr
\tilde{\gamma}^{8} =&  \cs^0 \otimes {{\bs}'}^2 \otimes {{\bs^*}''}^6,  
\end{align}
which gives rise to the following non-trivial triple intersection:
\begin{align}
& \int_{\widetilde{\M}^{24}} {\tilde{\alpha}^{8} \cup \tilde{\beta}^{8} \cup \tilde{\gamma}^{8}} \cr
=&  \int_{\M^8} ( {\bs^*}^6 \cup \bs^{2} \cup \cs^0) \cdot \int_{\M'^8} (\cs'^0   \cup {{\bs^*}'}^6 \cup \bs'^2 ) \cr
& \cdot \int_{\M''^8} (\bs''^2    \cup \cs''^0 \cup {{\bs^{*}}''}^6 )\cr
=& 1\cdot 1 \cdot 1 =1.
\end{align}
Both types of non-trivial triple intersections give rise to non-trivial logical CCZ gates according to Eq.~\eqref{eq:unitary_24D_simplicial}.  
\end{proof}

 Note that the logical gate structure is completely isomorphic to the construction with products of three 4-manifolds stated in Sec.~\ref{sec:magic_rate_hypergraph}, which differs only by shifting the (co)cycle dimensions.

\section{Parallelizable non-Clifford logical  gates and magic state fountain on constant-rate 3D homological product codes with $\Omega(\sqrt{N})$-distance }\label{sec:good_product}

\subsection{Constructing the thickened 3D  homological product codes and logical non-Clifford gates}

We first consider a 3D homological product code that is the homological product \cite{Bravyi:2014bq} of a good classical LDPC code $\bar{\C}$ with $n$ bits and a good quantum LDPC code $\bar{\C}'$ by Panteleev and Kalachev \cite{pkldpc22} with $m=\Theta(n)$ qubits. Here, 3D just means the underlying chain complex is a 3D (4-term) chain complex. We also note that in this construction we can use any good classical code, either from random bipartite expander graphs or the Sipser-Spielman construction \cite{Sipser_Spielman}, and there is no need to pick a symmetric parity-check matrix $\mathsf{H}^T\mathsf{H}$. We then construct thickened homological product codes by mapping the good classical LDPC code $\bar{\C}$ into a 4-manifold $\M^4$ (as discussed in Sec.~\ref{sec:4-manifold}) and the good qLDPC code $\bar{\C}'$ into an 11-manifold $\M'^{11}$  using the Freedman-Hastings mapping \cite{freedman:2020_manifold_from_code}.  The thickened homological product codes are then defined on the triangulation of the 15D product manifold $\tilde{\M}^{15}=\M^4 \times \M'^{11}$.

The details of the Freedman-Hastings mapping from the quantum code to an 11-manifold can be found in Ref.~\cite{freedman:2020_manifold_from_code}.   Besides the procedure introduced in Sec.~\ref{sec:handle_construction} about attaching ``4-handles" to the ``3-handles" according to the boundary map obtained from the $X$-checks, one needs to further attach ``5-handles" to the ``4-handles" according to the boundary map obtained from the $Z$-checks. One then obtains a 5-handlebody $H$, and taking the double of two identical copies of $H$ will produce the 11-manifold $\M'^{11}=\mathcal{D}H$.  The logical-$Z$ operators are supported on the 4-cycles, which has been visualized in Fig.~\ref{fig:dictionary}(h) inside the handlebody with the replacement of $D^6$ to $D^7$.  The logical-$X$ operators are supported on the dual 7-cycles due to Poincar\'e duality, and can be visualized in the same way as the 4-cycles in  Fig.~\ref{fig:dictionary}(h) with the additional trick of using the dual handlebody, i.e., the other upsidedown copy of handlebody in the double construction built from handles with dual indices similar to the situation in Eq.~\eqref{eq:long_chain_dual}.   

We now introduce the following theorem essentially obtained from Ref.~\cite{freedman:2020_manifold_from_code} (Theorem 1.2.1) with the additional input from Ref.~\cite{pkldpc22},  which has also been used to construct the 3D local code with optimal code parameters in Ref.~\cite{portnoy2023local} (restated as Theorem~5) \footnote{\label{footnote:freedman}As has been clarified below Theorem 5 in Ref.~\cite{portnoy2023local}, the original Theorem 1.2.1 in Ref.~\cite{freedman:2020_manifold_from_code} has a $polylog(m)$ reduction in the rate and distance due to the additional requirement that the underlying manifold is simply connected for the interest of systolic geometry.  When dropping this additional requirement which is unnecessary for the present paper, the proof in Ref.~\cite{freedman:2020_manifold_from_code} gives the optimal parameters without the $polylog(m)$ reduction.}: 
\begin{theorem}\label{theorem:FH} $($Freedman and Hastings \cite{freedman:2020_manifold_from_code}$)$
Given the good qLDPC code $\bar{\C'}$ from Ref.~\cite{pkldpc22} with the parameters $[[m, \Theta(m), \Theta(m)]]$ as an input, it can be mapped to  a good homological qLDPC code with the parameters $[[\Theta(m), \Theta(m), \Theta(m)]]$ defined on the triangulation of an 11-manifold $\M'^{11}$ with bounded local geometry and with its 4-systole and 7-systole corresponding to the logical-$Z$ and -$X$ distance respectively and both having  size $\Theta(m)$.
\end{theorem}

We note that although Ref.~\cite{portnoy2023local} uses the good LDPC code from Ref.~\cite{pkldpc22} as the input for the mapping, it has not justified the sparse liftability of this code. It is known that the good LDPC code in Ref.~\cite{pkldpc22} is a balanced product code from a pair of classical codes.  In Appendix \ref{app:sparse-liftability}, we prove that all balanced product codes of classical codes have sparse liftability.

We now consider the three following cohomology classes in the product manifold $\widetilde{\M}^{15}=\M^4 \times \M'^{11}$:
\begin{align}\label{eq:cocycle_list}
\alpha^{6} =& \as^{*2} \otimes \fs'^4,     \cr
\beta^{2} =& \as^{2} \otimes  \gs'^0,    \cr
\gamma^{7} =& \cs^0 \otimes {{\fs^*}'}^7,  
\end{align}
where $\as^{*2} \sim \as_2$ and $\as^2$  are a pair of dual cocycles with minimum size  $\Omega(n)$ and $\Omega(1)$ respectively in $\M^4$; $\fs'^4$ and ${{\fs^*}'}^7 \sim \fs'_4$ are a pair of dual cocycles  in $\M'^{11}$  which corresponds to the logical-$Z$ and -$X$ operators in the qLDPC code $\bar{\C}'$. Their minimum size corresponds to the 4-systoles and 7-systoles and both having size $\Omega(m)=\Omega(n)$; $c^0$ and $g'^0$ are the unique 0-cocycle class in $\M^4$ and $\M'^{11}$ respectively and both have size $\Omega(n)$.  We have the following non-trivial triple cup product using the Künneth theorem:
\begin{align}
 & \alpha^6 \cup \beta^2 \cup \gamma^7  \cr
 =& (\as^{*2} \otimes \fs'^4) \cup (\as^{2} \otimes  \gs'^0) \cup (\cs^0 \otimes {{\fs^*}'}^7) \cr
 =& (\as^2 \cup \as^{*2} \cup \cs^0) \otimes (\fs'^4 \cup \gs'^0 \cup {{\fs^*}'}^7) \neq 0 \in H^{15}(\widetilde{M}^{15}; \ZZ_2), \cr    
\end{align}
and hence the non-trivial triple intersection
\begin{align}\label{eq:triple_intersection}
& \int_{\widetilde{\M}^{15}} \alpha^6 \cup \beta^2 \cup \gamma^7 \cr
=&     \int_{\M^4}(\as^2 \cup \as^{*2} \cup \cs^0) \cdot \int_{\M'^{11}}(\fs'^4 \cup \gs'^0 \cup {{\fs^*}'}^7)        \cr
=&1                    
\end{align}

Based on the above triple intersection property we can construct three non-identical copies of qLDPC codes $\C_{(1)}$, $\C_{(2)}$, and $\C_{(3)}$ all defined on the same triangulation of 15-manifold $\widetilde{\M}^{15}$ with the qubits placed on 6-simplices, 3-simplices, and 7-simplices respectively. We also define a total code involving all the three copies of codes as $\tilde{\C}=\C_{(1)} \otimes \C_{(2)} \otimes \C_{(3)}$. Therefore, the basis cocycle classes $\alpha^6$, $\beta^3$ and $\gamma^7$ correspond to the support of logical-$X$ operators in the three copies of codes respectively, while their conjugate basis cycle classes $\alpha_6$, $\beta_3$ and $\gamma_7$ correspond to the support of the logical-$Z$ operators.    Since there exist spurious cycles in both the manifold $\M^4$ and $\M'^{11}$ \footnote{For the quantum code defined on the 11-manifold $\M'^{11}$, the spurious 1-cycles/cocycles and the dual 11-cycles/cocycles are separated from the 4-cycles and 7-cycles where logical $Z$ and $X$ are encoded respectively.} constructed from the classical and quantum codes respectively, we also choose the subsystem-code encoding, and only encode the logical-X operators into the cohomology basis subset $\{\alpha^6\}$,  $\{\beta^2\}$, and $\{\gamma^7\}$ in the three copies repsectively, with their conjugate logical-$Z$ operators encoded into the conjugate homology basis subset $\{\alpha_6\}$,  $\{\beta_2\}$, and $\{\gamma_7\}$ respectively, which are decomposed by the Künneth theorem as:
\begin{align}\label{eq:cycle_list}
\alpha_{6} =& \as^*_{2} \otimes \fs'_4,     \cr
\beta_{2} =& \as_{2} \otimes  \gs'_0,    \cr
\gamma_{7} =& \cs_0 \otimes {\fs^*_7}'.  
\end{align}

We hence introduce the following theorem:
\begin{theorem}\label{theorem:homological_product}
There exist a family of thickened 3D homological  product codes $\tilde{\C}$ defined on the triangulation of a 15-manifold  $\tilde{\M}^{15}$ with encoding rate $K=\Theta(N)$, subsystem-code distance $D=\Omega(\sqrt{N})$ and constant stabilizer weight $w=O(1)$,  such that a constant-depth circuit implementing the cohomology operation of a triple cup product give rise to $\Theta(N)$ non-Clifford logical CCZ gates on $\tilde{\C}$.     
\end{theorem}

\begin{proof}
We first estimate the total number of qubits in each code copy, which equals the total number of 6-simplices, 2-simplices and 7-simplices on $\tilde{\M}^{15}$ respectively.  Due to Lemma \ref{lemma:thickened_code} and  Theorem \ref{theorem:FH}, the factor manifold $\M^4$ and $\M'^{11}$ both have bounded local geometry, i.e., each vertex in its triangulation is adjacent to $O(1)$ $k$-simplices, the product manifold $\tilde{\M}^{15}=\M^4 \times \M'^{11}$ also has bounded local geometry.  Therefore, the number of $k$-simplices in these manifolds are all proportional to the number of vertices and the volume of the manifold.  Using the fact that $\tilde{\M}^{15}$ is the product of $\M^4$ and $\M'^{11}$, we know the number of $k$-simplices in $\tilde{\M}^{15}$, i.e., the dimension of the $k$-chain groups,  all scale as 
\be
\text{dim}(C_k) = \Theta(n) \cdot \Theta(m) \cdot \Theta(n^2) = \Theta(N), 
\ee
where $N$ is the total number of vertices in $\tilde{M}^{15}$.

Now since both cocycles basis $\{\as^2\}$ and $\{\fs'^4\}$ have dimension $\Theta(n)$, which come from the code dimension of the input classical code $\bar{C}$ and quantum code $\bar{C}'$, the cocycle basis subset $\{\alpha^6\}$ has dimension $\Theta(n^2)$ according to the Künneth theorem.  Therefore, the first copy of qLDPC code $\C_{(1)}$ with qubits placed on 6-simplices has linear dimension  $K^{(1)}=\Theta(n^2)=\Theta(N)$, i.e., constant encoding rate.  For the second and third copies of qLDPC codes $\C_{(2)}$ and $\C_{(3)}$, since both the cocycle basis $\{\as^{*2}\}$ and $\{{{\fs^{*}}'}^7\}$ have dimension $\Theta(n)$, which again come from the code dimension of the classical and quantum codes, the code dimension is hence $K^{(2)}=K^{(3)}=\Theta(n)=\Theta(\sqrt{N})$ for both the second and third copies. When summing up all the logical qubits in these three code blocks, the overall encoding rate of the total code $\tilde{C}=\C_{(1)} \otimes \C_{(2)} \otimes \C_{(3)}$ is still constant, i.e.,
\be
r_E=\frac{K}{\Theta(N)}= \frac{K^{(1)}+K^{(2)}+K^{(3)}}{\Theta(N)}= O(1).
\ee
We hence call the first qLDPC code copy the memory register which is used for information storage with a constant encoding rate, and the second and third qLDPC code copies the ancilla registers which are used to assist the memory register for doing logical non-Clifford gates. 

We then bound the distance of the subsystem-code encoding using Lemma~\ref{lemma:subsystem}.  For the first code copy $\C_{(1)}$,   the logical-$Z$ distance is determined by the minimum size of any basis 6-cycle in the basis subset $\{\alpha_6\}$ defined in Eq.~\eqref{eq:cycle_list}, i.e.,
\begin{align}
 \text{min}(|\alpha_6|) =& \text{min}(|\as^*_2|) \cdot \text{min}(|\fs'_4|) \cr
=& \Omega(1) \cdot \Omega(m) = \Omega(n) = \Omega(\sqrt{N}),
\end{align}
where we have used the bound of the cycle/cocycle length $|\as^*_2|=|\as^2|=\Omega(1)$ in the thickened good classical LDPC code  as discussed in the proof of Theorem~\ref{theorem:3Dhypergraph} and the $Z$-distance bound (4-systole) of the thickened good qLDPC code for $|\fs'_4|=\Omega(m)$ according to Theorem~\ref{theorem:FH}. We hence have
\be
d_Z^{(1)}=\text{min}(\{ \alpha_6 \}) = \Omega(\sqrt{N}).
\ee
Meanwhile, the $X$-distance is determined by the minimum size of any conjugate basis 6-cocycle in the basis subset $\{\alpha^6\}$ defined in Eq.~\eqref{eq:cocycle_list},  or equivalently that of its Poincar\'e dual basis 9-cycle in the basis subset $\{\alpha^*_9\}$:
\begin{align}
 \text{min}(|\alpha^6|)=&  \text{min}(|\alpha^*_9|) = \text{min}(|\as^{*2}|) \cdot \text{min}(|\fs'^4|) \cr
=& \text{min}(|\as_{2}|) \cdot \text{min}(|{\fs_7^{*}}'|) =\Omega(n) \cdot \Omega(m) \cr
=& \Omega(n^2) = \Omega(\sqrt{N}),
\end{align}
where we have used the distance bound of the thickened classical code, i.e., $|\as_2|=|{\as^*}^2|=\Omega(n)$, as well as the $X$-distance bound (7-systole) of the thickend good qLDPC code for $|f'^4|=|{\fs_7^{*}}'|=\Omega(m)$.
We hence have 
\be
d_X^{(1)}=\text{min}(\{ |\alpha^6| \})=\text{min}(\{|\alpha^*_9|\}) = \Omega(N),
\ee
and the overall code distance:
\be
d^{(1)}= \text{min}(d^{(1)}_X, d^{(1)}_Z) = \Omega(\sqrt{N}).
\ee

For the second code copy $\C_{(2)}$, the $Z$-distance is determined by the minimum size of any basis 2-cycle in the basis subset $\{\beta_2\}$: 
\begin{align}
 \text{min}(|\beta_2|) =& \text{min}(|\as_2|) \cdot \text{min}(|\gs'_0|) \cr
=& \Omega(n) \cdot 1 =  \Omega(\sqrt{N}),
\end{align}
where we have used the fact that $\gs'_0$ is a single vertex in the triangulation of the 11-manifold $\M^{11}$.  We hence have
\be
d^{(2)}_Z= \text{min}(\{\beta_2\}) = \Omega(\sqrt{N}).
\ee
Meanwhile, the $X$-distance is determined by the minimum size of any conjugate basis 2-cocycle in the basis subset $\{\beta^2\}$, or equivalently that of its Poincar\'e dual basis 13-cycle in the basis subset $\{\beta^*_{13}\}$
\begin{align}
 \text{min}(|\beta^2|)=&  \text{min}(|\beta^*_13|) = \text{min}(|\as^{2}|) \cdot \text{min}(|\gs'^0|) \cr
=&  \Omega(1) \cdot \Omega(m) = \Omega(n) = \Omega(\sqrt{N}),
\end{align}
where we have used the fact that $|\as^2|=\Omega(1)$ in the thickened classical code, and $|g'^0|=|{g^*_{11}}'|=\Omega(m)$ occupies the entire 11-manifold $\M'^{11}$.  We hence have
\be
d_X^{(2)}=\text{min}(\{ |\alpha^2| \})=\text{min}(\{|\alpha^*_{13}|\}) = \Omega(\sqrt{N}),
\ee
and the overall distance is hence
\be
d^{(2)}= \text{min}(d^{(2)}_X, d^{(2)}_Z) =  \Omega(\sqrt{N}).
\ee

For the third copy $\C_{(3)}$,  the 0-cycle $\cs_0$ is a single vertex in the thickened classical code and hence has $|\cs_0|=1$;  its conjugate 0-cocycle (Poincar\'e dual to a 4-cycle) $\cs^0 \sim \cs^*_4$  occupies the entire 4-manifold $\M^4$ which leads to $|\cs^0| = |\cs^*_4| = 1$.  Meanwhile, $|{f^*_7}'|$ and $|{f^*}'^7|=|f'_4|$ correspond to the 7-systole and 4-systole of the 11-manifold $\M'^{11}$.  This similarly gives rise to 
\begin{align}
d^{(3)}_Z =& \min{|\gamma_7|}= \Omega(\sqrt{N}), \cr
d^{(3)}_X =& \min{|\gamma^7|}=\min{|\gamma^*_8|} = \Omega(\sqrt{N}) \cr
d^{(3)}=& \text{min}(d^{(3)}_X, d^{(3)}_Z) =  \Omega(\sqrt{N}).
\cr
\end{align}
The overall distance of the total code $\tilde{\C}=\C_1 \otimes \C_2 \otimes \C_3$ is hence $D=\text{min}(d^{(1)}, d^{(2)}, d^{(3)})= \Omega(\sqrt{N})$.

Finally, we implement the following constant-depth circuits corresponding to the triple cup product of a higher gauge theory defined on the 15-manifold $\tilde{\M}^{15}$ according to Lemma \ref{lemma_gate_3} and  Eq.~\eqref{eq:logical_CCZ_higher}:
\begin{align}\label{eq:unitary_15D_simplicial}
&U= (-1)^{\int_{\widetilde{\M}^{15}} a^{6}_{(1)} \cup a^{2}_{(2)} \cup a^{7}_{(3)}}    \cr
=& \prod_{\alpha^{6}, \beta^{2}, \gamma^{7}} \overline{\text{CCZ}}[(\alpha^{6}; 1), (\beta^{2}; 2),(\gamma^{7}; 3)]^{\int_{\widetilde{\M}^{15}} {\alpha^{6} \cup \beta^{2} \cup \gamma^{7}}}. \cr
\end{align}
Since the triple cup product sum in the exponent is non-trivial, we obtain a non-Clifford logical gate. 

\end{proof}

We note that we can also have an alternative construction with the good classical code being mapped to the 8-manifold $M^8$, which gives rise to the product manifold $\widetilde{\M}^{19}= \M^4 \times \M'^{11}$.  The triple intersection and logical CCZ structure will be similar.

\subsection{Magic state fountain and logical gate structure}\label{sec:fountain}

\begin{figure}[t]
\includegraphics[width=1\columnwidth]{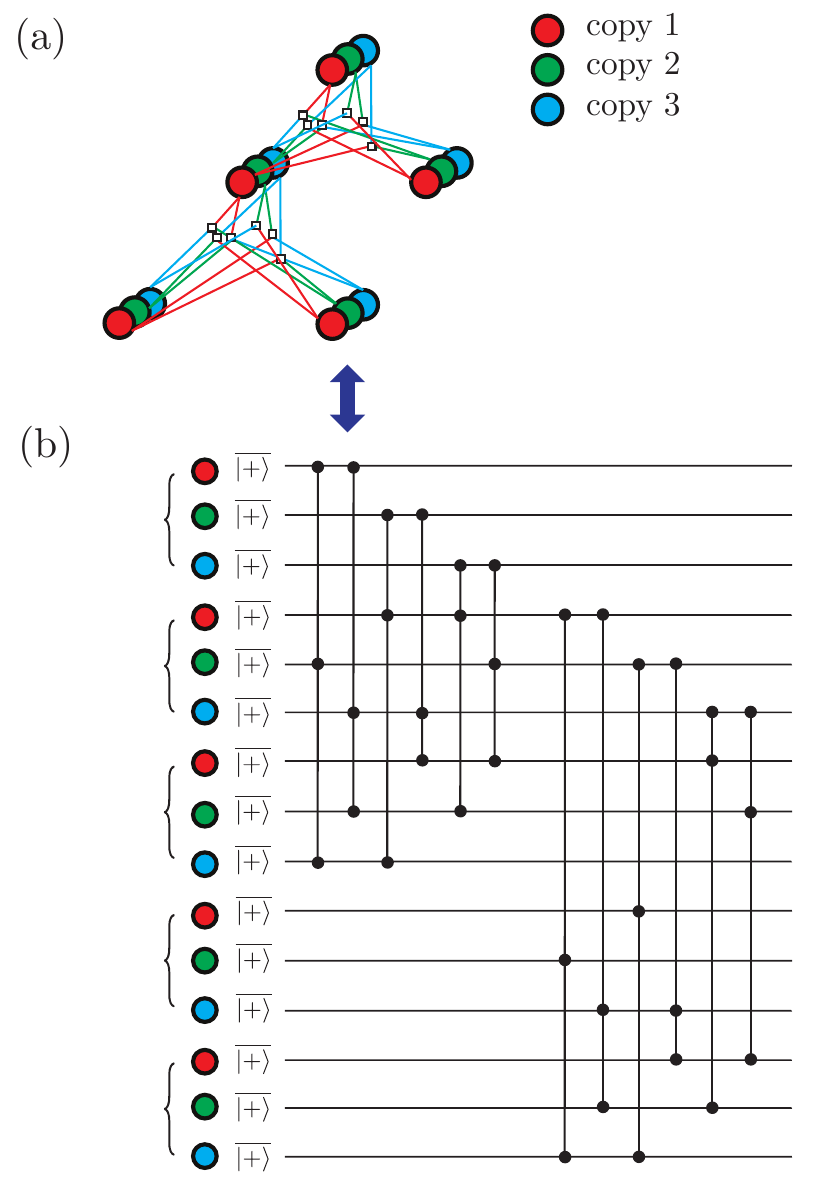}
\caption{ (a)  Illustration of an interaction hypergraph involving three identical copies of codes, where each vertex represents a logical qubit and each hyperedge (3-way junction) represents a logical CCZ gate acting on the code.  (b) When initializing all the logical qubits into the logical state $\lo{\ket{+}}$, the application of the logical gate $U$ effectively applies the logical CCZ's illustrated in the circuit and hence injects a hypergraph magic state. }\label{fig:hypergraph_magic_state}
\end{figure}

\begin{figure*}[t]
 \includegraphics[width=1.6\columnwidth]{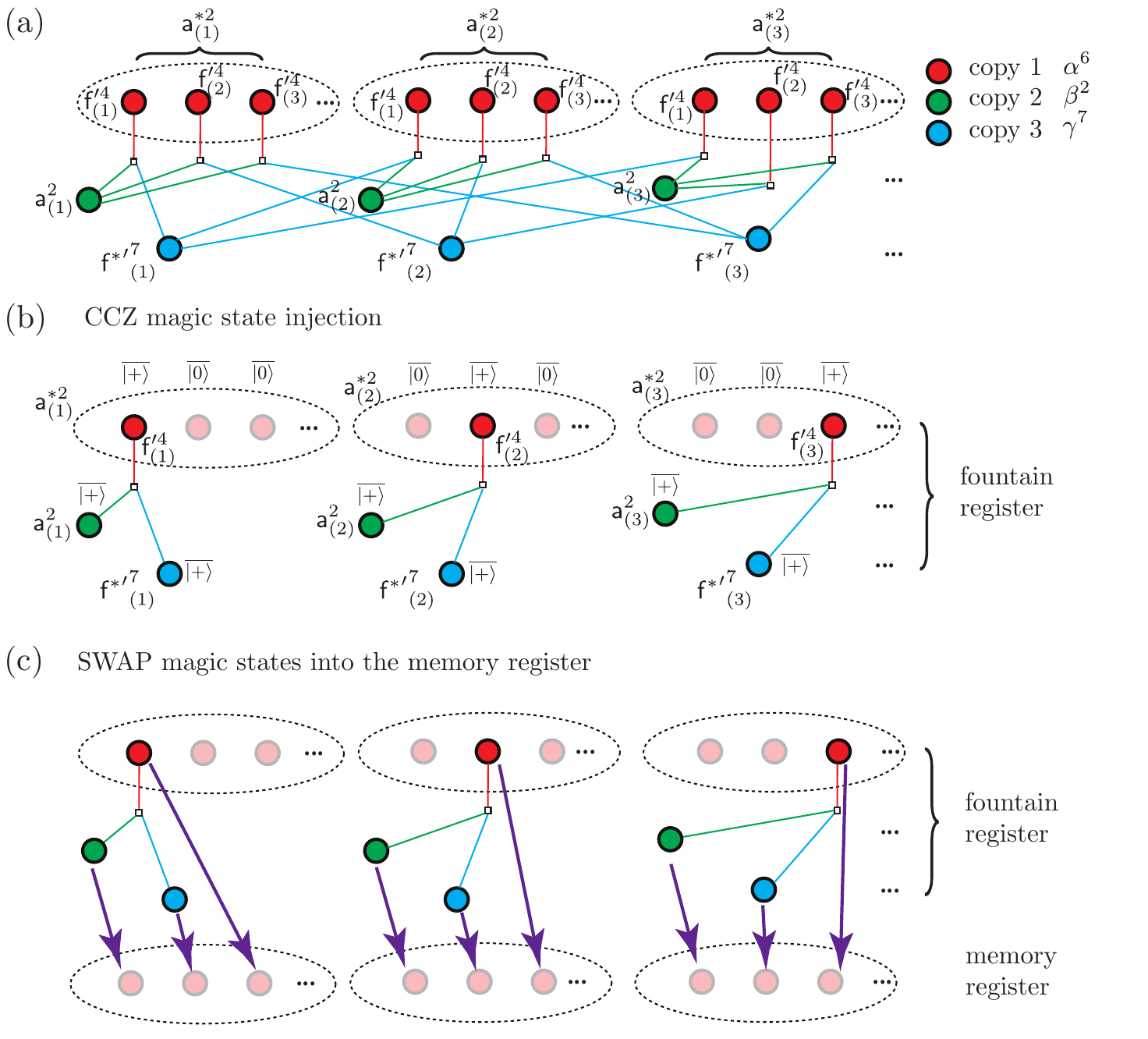}
\caption{ Illustration of the logical gate structure and magic state fountain in the homological product code.  (a) The interaction hypergraph involving three non-identical copies of codes.  The logical qubits in copy 1 has two cocycle labels $\as^{*2}_{(i)}$ and $\fs^{'4}_{(j)}$, while those in copy 2 and copy 3 has one cocycle label. (b) Illustration of the injection of non-overlapping CCZ magic states into the fountain register.  One deliberately sets a fraction of logical qubits in copy 1 at $\lo{\ket{0}}$ state which effectively turns off all the hyperedges (CCZ's) connected two these logical qubits.  The remaining $\Theta(\sqrt{N})$ CCZ magic states are hence non-overlapping. (c) After injecting the magic states into the fountain register, one further SWAPs these resource states into the memory register via logical Clifford gates and then use gate teleportation to turn them into parallelizable logical CCZ gates. }\label{fig:interaction-hypergraph}
\end{figure*}

As has been first introduced in Ref.~\cite{zhu2023non}, the collective logical CCZ gate structure is encoded into an \textit{interaction hypergraph} $G_h=(V, E_h)$.   The vertex $v \in V$ represents a logical qubit labeled by the cocycle label and copy number, such as $(\alpha^6; 1)$, $(\beta^2; 2)$ or $(\gamma^7; 3)$, while a hyperedge $e_h \in E_h$ coupling logical qubits corresponds to a triple intersection point or equivalently a non-trivial cup product such as  $\int_{\widetilde{\M}^{15}} {\alpha^{6} \cup \beta^{2} \cup \gamma^{7}}=1$ in the current code.   An example of the interaction hypergraph is given by Fig.~\ref{fig:hypergraph_magic_state}(a) where the three copies of codes are identical, as in the case of the thickened 3D hypergraph code in Sec.~\ref{sec:non-Clifford_hypergraph}. 

As has been pointed out in Ref.~\cite{zhu2023non}, the interaction hypergraph has a one-to-one correspondence with the \textit{quantum hypergraph states} \cite{Rossi_2013, chen2023magic, Takeuchi:2019_hypergraph} as the generalization of graph states. One can define a 3-uniform hypergraph state on a hypergraph $G_h=(V, E_h)$ as follows: one assigns to each vertex a qubit initialized in the $\ket{+}$ state; for each hyperedge, one perform a CCZ gate between the three connected qubits (vertices) labeled by $v_1, v_2, v_3$. We then obtain the following 3-uniform hypergraph state:
\be
\ket{g_3} = \prod_{\{v_1, v_2, v_3\} \in E_h} \text{CCZ}(v_1, v_2, v_3)\ket{+}^{\otimes m},
\ee
where $m=|V|$ represents the total number of qubits (vertices). We also call this state a hypergraph magic state, since it is beyond the stabilizer description  \cite{chen2023magic, zhu2023non}.  When considering our collective logical CCZ gate $U$ and its corresponding interaction hypergraph, we can first initialize all the logical qubits in the logical state $\overline{\ket{+}}$, and then apply the collective logical CCZ gate $U$ which hence produces the hypergraph magic state, as illustrated in Fig.~\ref{fig:hypergraph_magic_state}(b).  We call this type of scheme \textit{magic state fountain}, which directly injects high-fidelity magic states into an qLDPC code block instead of doing a state distillation with multiple rounds. This idea has already appeared in Ref.~\cite{zhu2023non}.

We now investigate the logical gate structure from the $\ZZ_2$ triple intersection structure in Eq.~\eqref{eq:cocycle_list} and Eq.~\eqref{eq:triple_intersection}, and then connect it to the magic state fountain scheme.

As has been discussed above, there are $\Theta(N)$ inequivalent choices of $\alpha^6 = \as^{*2} \otimes \fs'^4$ coming from the $\Theta(\sqrt{N})$ choices for both $\as^{*2}$ and $\fs'^4$.  For each choice of $\alpha^6$,  there is a unique pair of $\beta^2$ and $\gamma^7$ (involving the dual components $\as_2$ and ${\fs^*_7}'$) that has non-trivial triple cup product (intersection) with $\alpha^6$ due to Poincar\'e duality, although the different pairs can still share the same $\beta^2$ or $\gamma^7$. The total number of logical CCZ's implemented by the constant-depth circuit $U$ is hence $n_\text{CCZ} = \Theta(N)$ corresponding to a constant magic rate $r_M=n_\text{CCZ}/\Theta(N) = O(1)$. This magic rate also quantifies the complexity of the corresponding hypergraph magic state injected to the code when applying $U$ to the logical state $\lo{\ket{+}}^{\otimes K}$, which grows linearly with the number of qubits $N$.  All logical qubits in the total code $\tilde{\C} =\C_{(1)} \otimes \C_{(2)} \otimes \C_{(3)}$ participate in the logical gate $U$, therefore the addressing rate is $r_A $$=$$\Theta(N)/ \Theta(N) $$=$$ O(1)$.

We then delve more deeply into the logical gate structure by looking at the interaction hypergraph as shown in Fig.~\ref{fig:interaction-hypergraph}(a). Since $\alpha^6 = \as^{*2} \otimes \fs'^4$,  we can divide the logical qubits in copy 1 labeled by $\alpha^6$ (red circles) into $\Theta(\sqrt{N})$ groups (dashed ellipse) with different $\as^{*2}$ label, while in each group we choose all possible $\fs'^4$ labels and hence have $\Theta(\sqrt{N})$ logical qubits.  These $\Theta(N)$ logical qubits labeled by $\alpha^6$ in copy 1 are then coupled to $\Theta(\sqrt{N})$ logical qubits labeled by $\beta^2$ in copy 2 and $\gamma^7$ in copy 3 respectively via logical CCZ. Note that since $\beta^{2} = \as^{2} \otimes  \gs'^0$ and $\gamma^{7} = \cs^0 \otimes {{\fs^*}'}^7$, and $\gs'^0$ as well as ${{\fs^*}'}^7$ are unique choices, we can just use label $\as^2$ and ${{\fs^*}'}^7$  to label the qubits in copy 2 and 3 respectively. In particular, each logical qubit in copy 1 with cocycle label $\as^{*2}$ couples to a unique logical qubit in copy 2 with the dual cocycle label $\as^2$ according to Poincar\'e  duality.  Similarly,   each logical qubit in copy 1 with cocycle label $\fs'^4$ couples to a unique logical qubit in copy 2 with the dual cocycle label ${{\fs^*}'}^7$.   This completes the description of the structure of the collective logical CCZ gate and the corresponding hypergraph magic state, as  illustrated in Fig.~\ref{fig:interaction-hypergraph}(a). 

As pointed out before, the complexity of the hypergraph magic state here grows linearly with the system size $\Theta(N)$, and is expected to be hard for a classical computer to simulate.  So such a high-complexity state may be useful for demonstrate quantum advantage with a constant-depth circuit.  On the other hand, for the purpose of universal quantum computation with the magic state fountain scheme, we may want to reduce the complexity of this hypergraph magic state. 

\begin{figure*}[t]
\includegraphics[width=1.4\columnwidth]{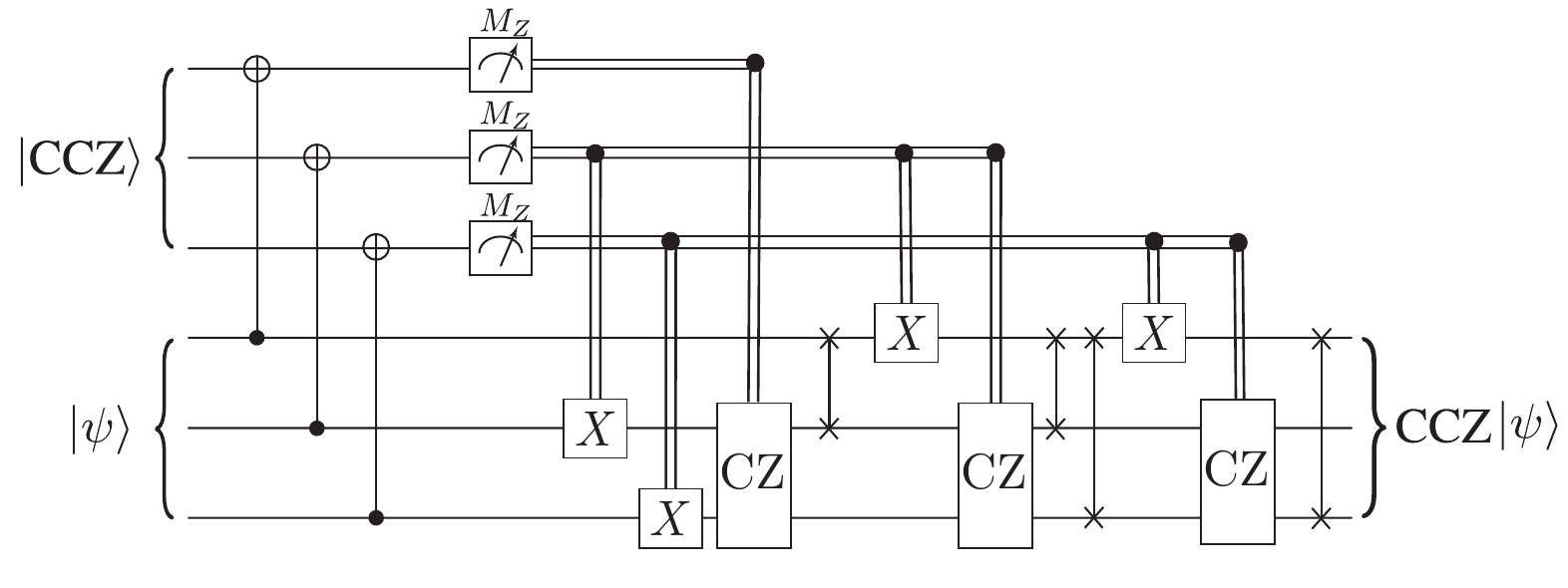}
\caption{The gate teleportation gadget from Ref. \cite{beverland2020lower} that converts a CCZ magic state into a CCZ gate acted on the input state $\ket{\psi}$. }\label{fig:CCZ_gadget}
\end{figure*}

We consider a \textit{fountain register} composed of all three different copies of qLDPC codes, as shown in Fig.~\ref{fig:interaction-hypergraph}(b).  We first initialize all the logical qubits in copy 2 and 3 (green and blue) into $\lo{\ket{+}}$ states.  We then initialize $\Theta(\sqrt{N})$ logical qubits in copy 1 (red) with different labels for both $ \as^{*2} $ and $\fs'^4$  into the logical state $\lo{\ket{+}}$, while ``turning off" the rest of the logical qubits  into state  $\lo{\ket{0}}$.  All the hyperedges coupled to the turned-off logical qubits in state  $\lo{\ket{0}}$ are effectively removed, so are the corresponding logical CCZ's.   The hypergraph magic state hence becomes a tensor product of $\Theta(\sqrt{N})$ non-overlapping  CCZ magic states, i.e., $\Motimes_{j=1}^{\Theta(\sqrt{N})}\overline{\ket{\text{CCZ}}}_j$, as illustrated in Fig.~\ref{fig:interaction-hypergraph}(b).

We then SWAP the $\Theta(\sqrt{N})$ CCZ  magic states into the memory register via logical Clifford gates, either using the combination of targeted logical CZ via a constant-depth circuit and logical-$X$ measurements which has been previously studied in Ref.~\cite{zhu2023non}, or using various lattice surgery schemes for logical Pauli measurements in qLDPC codes \cite{cohen22, huang2023homomorphic,  xu2024fast, cross2024improved, williamson2024low, swaroop2024universal}.  We now repeat this inject-and-SWAP process for $\Theta(\sqrt{N})$ rounds, which can then inject $\Theta(N)$ CCZ magic states into the memory register.  Now we can consume this $\Theta(N)$ CCZ magic states to implement logical CCZ gates via gate teleportation with the logical circuit from Ref. \cite{beverland2020lower} and shown in Fig.~\ref{fig:CCZ_gadget}.

As shown above, one can inject $\Theta(\sqrt{N})$ CCZ  magic states in parallel in a single round using $N$ physical qubits. The single-round state injection rate in the current scheme is hence 
\be
r_I=\Theta(\sqrt{N})/N=\Theta \left(\frac{1}{\sqrt{N}}\right) = \Theta \left(\frac{1}{D}\right). 
\ee
In contrast, when using the same number of $N$ physical qubits to form a single copy of 3D color code defined on a 3D cube with boundaries (equivalent to three copies of 3D surface codes) \cite{Kubica:2015br} with code parameters $[[N, 3, \Theta(N^{1/3})]]$, one can only inject a single CCZ magic states in a single round. More generally, when using $K$ copies of 3D color code with code parameter scaling $N=\Theta(KD^3)$, the corresponding  CCZ magic state injection rate is 
\be
r'_I=K/N=\Theta \left(\frac{1}{D^3}\right).
\ee
We hence know that the 3D thickened homological product codes has better single-round injection rate than the 3D color code with the same distance.   Now, the $\Theta(\sqrt{N})$-round state injection for the thickened homological product code has a constant injection rate, i.e.,  $\tilde{r}_I=O(1)$, while the $\Theta(\sqrt{N})$-round injection rate for the 3D color code still approaches to zero when $N \rightarrow \infty$ (equivalently $D \rightarrow \infty$).   

In terms of the quality of the magic states, with $\Theta(N)$ qubits, the thickened homological product codes can produce $\Theta(N)$ magic states with an effective distance $D=\Omega(\sqrt{N})$ in $\Theta(\sqrt{N})$ rounds. In contrast, even for a single copy of 3D color code producing a single magic state, the effective distance can only be $\Theta(N^{1/3})$. 

We hence summarize the above finding in the following corollary of Theorem \ref{theorem:homological_product}:
\begin{corollary}\label{corollary:homological}
	For the family of thickened 3D homological product codes defined on the triangulation of a 15-manifold from Theorem \ref{theorem:homological_product}, one can use a constant-depth circuit to fault-tolerantly prepare $\Theta(\sqrt{N})$ non-overlapping logical CCZ magic states with distance $\Omega(\sqrt{N})$ in a single shot without disttilation. 
\end{corollary}

We note that although the current construction can only parallelize $\Theta(\sqrt{N})$ logical CCZ's in a single round rather than $\Theta(N)$, it may not be a bottleneck for this computing scheme at the current stage since the state-of-the-art logical Clifford gates cannot be fully parallelized, i.e., achieving $\Theta(N)$ gates per round with a constant space overhead.  A naive estimate is as follows: the lattice-surgery-based logical measurements can address at most $O(\sqrt{N})$  non-overlapping logical operators for a distance-$O(\sqrt{N})$ code. Assuming one can only do logical measurements on non-overlapping logicals with a constant space overhead, one may only be able to parallelize $O(\sqrt{N})$ logical Clifford gates, which is just as good as the parallelizability of CCZ magic state injection in the current construction.  

We also note that we have not estimated the time overhead and prallelizability for the SWAP operation as well as the CCZ magic state consumption.   We will leave these for future study.

\section{Deformation retraction to the hidden CW complexes and low-overhead scheme for practical implementation}\label{sec:deformation_retraction}

In the previous sections, we have been focused on constructing homological codes on the triangulation $\L$ of the manifolds from the input classical or quantum qLDPC codes, and realize logical CCZ gates via triple cup products defined on the triangulation $\L$.
The manifold constructed in Secs.~\ref{sec:non-Clifford_hypergraph} and \ref{sec:good_product} have dimensions 12 and 15 respectively.  Cautious readers might question whether such constructions would be useful for practical implementation, especially in near-term realizations with $O(100)$ to $O(1000)$ qubits.  First of all, dimensionality should not be a major concern for practicality in the context of qLDPC codes, since most of these codes are not geometrically local and hence require long-range connectivity.  Therefore, connectivity rather than the dimensionality will instead be the major concern. As we can see in the previous manifold constructions, the manifold has inherited the combinatorial properties from the skeleton codes through handle attachment and each handle only contains $O(1)$ constant number of simplices. Therefore, the manifolds have essentially the same connectivity as its input code at coarse-grained scale.   For example, the 2D hypergraph product code has been shown to be embeddable into constant number of layers with long-range connection without a crossing within each layer \cite{tremblay2022constant}. Therefore a thickened 2D hypergraph product code defined on the manifold should inherit such connectivity.   A similar embeddability for a thickened good qLDPC code has also been hinted in Ref.~\cite{portnoy2023local}.  Similarly, a thickened 3D hypergraph (homological) product code on the manifold in Sec.~\ref{sec:non-Clifford_hypergraph} (Sec.~\ref{sec:good_product}) should have a similar connectivity as the skeleton 3D hypergraph (homological) product code. 

Some remaining concern would be the stabilizer weight on high dimensions as well as how large the constant overhead is when subdividing each handle into simplices.  To  address this, we will show in the following a more compact code realization via deformation retraction to a CW complex.

 \subsection{Deformation retraction}\label{sec:retraction_detail}
Here, let us present more details on deformation retraction introduced in Sec.~\ref{sec:handle_8-manifold}.  The deformation retraction maps the manifold described by the handle complex $\L_h$ (e.g. Eq.~\ref{eq:long_chain}) to a CW (cellular) complex $\L_c$. As pointed out in Ref.~\cite{freedman:2020_manifold_from_code, guemard2025lifting}, 
the handle complex and the CW complex is isomorphic: $\L_h \cong \L_c$. 

In $r$-dimension, the deformation retraction $\R$  retracts a $k$-handle to its core, i.e., a $k$-cell $D^k$:
\be\label{eq:retraction_k-handle}
\R: D^k \times D^{r-k} \rightarrow D^k. 
\ee
Similarly, $\R$ retracts the dressed ``$k$-handle" to its dressed core $N^k$, i.e., a dressed ``$k$-cell": 
\be\label{eq:retraction_dressed_handle}
\R: N^k \times D^{r-k} \rightarrow N^k.
\ee

\begin{figure*}
    \centering
    \includegraphics[width=2\columnwidth]{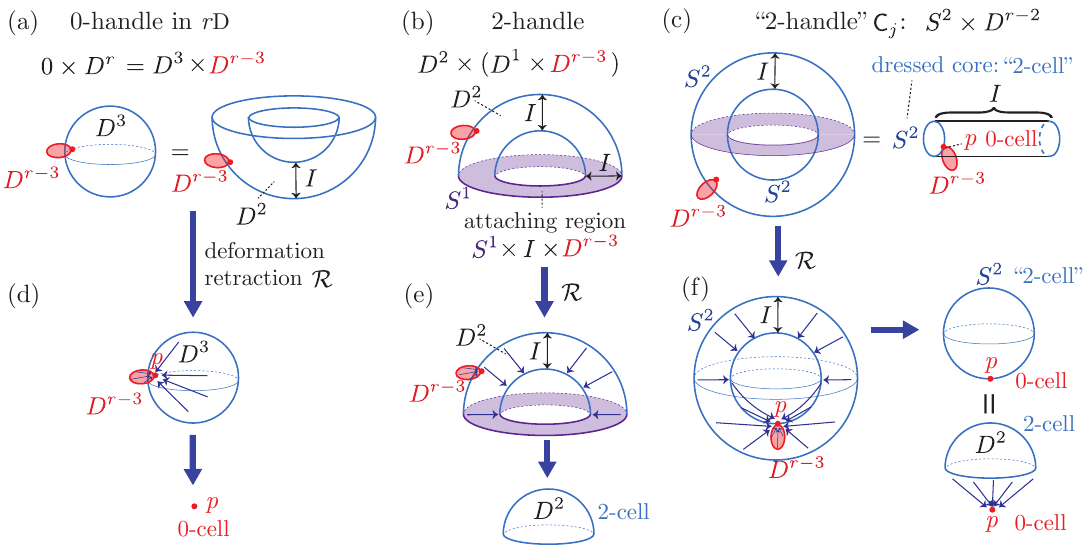}
    \caption{Anatomy of the dressed ``2-handle". (a) A 0-handle in $r$D can be deformed into a half-shell $D^2 \times I$ thickened by $D^{r-3}$ facing south. (b) A 2-handle deformed into a thickened half-shell facing north. The attaching region (purple) is an thickened annulus on the equator. (c) The 2-handle is attached to the 0-handle along the thickened annulus at the equator, which forms the ``2-handle" as a whole shell $S^2 \times I$ thickened by $D^{r-3}$. It is equivalent to the more abstract presentation for the ``2-handle" on the right, where the ``2-cell" $S^2$ forms the dressed core.  (d)  Deformation retraction of the 0-handle to a single point $p$ as the 0-cell.  (e) Retracting the 2-handle to a north hemisphere, i.e., a 2-cell $D^2$. (f) Retracting the dressed ``2-handle", where the north part is retracted to a north hemisphere and the south part is retracted to a point $p$, which becomes a ``2-cell" $S^2$ along with a point $p$ at the south pole.  This is equivalent to a 2-cell (disk $D^2$) attached to the 0-cell (vertex $p$).   }
    \label{fig:anatomy}
\end{figure*}

We can apply the deformation retraction to both the manifold built from the classical code in Sec.~\ref{sec:handle_construction} where bits are placed on 3-cells (with dimension $r\ge 8$)  and the construction in Sec.~\ref{sec:4-manifold} where bits are placed on 2-cells with no separation with the spurious (co)cycle dimension, as well as the Freedman-Hastings manifold built form the quantum code presented in Sec.~\ref{sec:good_product} \cite{freedman:2020_manifold_from_code}.   To demonstrate the scheme concretely, we focus on the manifold where bits are placed on 3-cells following Sec.~\ref{sec:handle_construction} and choose $r \ge 8$, such that there is a separation in dimensions between the spurious and logical (co)cycles\footnote{As mentioned before, although $r=7$ should also work, we choose $r \ge 8$ for conceptually simpler situation where the left and right portions are separated by trivial group $0$ and trivial boundary maps in the middle.}.  As we will see, in the CW complex approach we actually do not care much about the total dimension $r$, and the construction easier if the handles and their dual handles have different dimensions (unlike the 4-manifold construction where the dual of the 2-handles are still 2-handles).

\begin{figure}[t]
    \centering
    \includegraphics[width=1\linewidth]{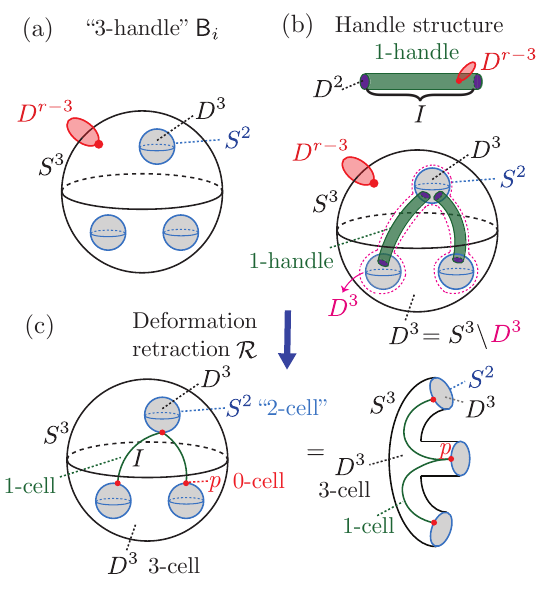}
    \caption{The anatomy of the dressed ``3-handle". (a) A realistic 3D presentation of the ``3-handle".  The 3-sphere should be viewed as the entire interior 3D space of a 3-ball with its boundary being identified to a single point. Three thickened 3-balls $D^3 \times D^{r-3}$ (grey) are removed, with the boundary being the thickened 2-sphere $S^2 \times D^{r-3}$.  (b) On the 3D dressed-core, the ``3-handle" contains two 1-handles  connecting the three thickened $S^2$ boundaries, together with the three removed $D^3$ it forms a single  3-ball $D^3$ (highlighted by the pink dashed lines). The complement in the 3-sphere is just a 3-ball $D^3=S^3 \backslash D^3$, which is the 3-handle when thickened by $D^{r-3}$. (c) Deformation retracting two 1-handles and one 3-handles into two 1-cells and one 3-cell, and the ``3-handle" into its dressed core.  It also connects to three ``2-cell" $S^2$ and three 0-cell $p$ on the neighboring ``2-handles".   One can deform the drawing to a more abstract representation on the right with the same style as Fig.~\ref{fig:dictionary}(d).     }
    \label{fig:retraction_3-handle}
\end{figure}

In order to understand the details of deformation retraction, we need to analyze the anatomy of the handle structure for the dressed ``2-handles" and ``3-handles" respectively.   The handle structure of the ``2-handles" $\mathsf{C}_j$ has been discussed briefly in Sec.~\ref{sec:handle_construction} (for $r=5$).  Here we give a more detailed geometric instruction.  The handle attachment scheme here can be considered as the higher-dimensional generalization of the attachment of the 1-handle to the 0-handle forming a (0,1)-handlebody in 2D illustrated in Fig.~\ref{fig:handle_introduction} in  Sec.~\ref{sec:handle_construction}.   We start with the 0-handle in $r$ dimensions, i.e., $h_0=(0 \times D^r, \varnothing)$, which can also be expressed as a 3-disk (3-ball) thickened by an $(r-3)$-disk, i.e., $D^3 \times D^{r-3}$, as represented in Fig.~\ref{fig:anatomy}(a).   The 3D part can further deform it into a half-shell, i.e., a thickened hemisphere (2-disk): $D^2 \times I$.  Now we introduce the 2-handle $h_2=(D^2 \times D^{r-2}, S^1 \times D^{r-2})$.  It will be convenient to also represent it as a thickened half-shell as $(D^2 \times I) \times D^{r-3}$, as shown in Fig.~\ref{fig:handle_introduction}(b) where we have flipped it to enclose the north hemisphere. The attaching region (purple) is a thickened annulus $(S^1\times I) \times D^{r-3}$.   We then attach the 2-handle $h_2$ to the 0-handle $h_0$ along the attaching region of $h_2$.  This can be visualized as attaching two thickened half-shells along the thickened annulus in the equator, which results in a thickened whole shell $(S^2 \times I) \times D^{r-3}$, as illustrated in Fig.~\ref{fig:anatomy}(c).  We then deform it into our familiar presentation in Fig.~\ref{fig:dictionary} where the 2-sphere is abstractly represented by a blue circle.

We now investigate the deformation retraction of the ``2-handle".  First, the 0-handle $0 \times D^r$ ($r$-dimensional ball), is retracted to its core, i.e., a single point (vertex) $p$ as a 0-cell, as illustrated in Fig.~\ref{fig:anatomy}(d).  Next, the 1-handle $D^2 \times D^{r-2}$ is retracted to its code $D^2$, which becomes a 2-cell. We can understand this in the thickened half-shell picture in Fig.~\ref{fig:anatomy}(e), where the thickness of the shell and annulus is reduced to zero while in the thickening dimensions the $(r-3)$-ball  $D^{r-3}$ is shrunk to a single point.   Finally, since the ``2-handle" is composed of one 0-handle and one 2-handle, we can now simultaneously retract its components [see Fig.~\ref{fig:anatomy}(f)].  On the north part, we retract the 2-handle into a north hemisphere (disk) $D^2$.  On the south part, the entire 0-handle is retracted to a point $p$ at the southpole, which effectively identify the equator of the north hemisphere into the southpole $p$. Therefore, the ``2-handle" is now retracted to its dressed core $S^2$, which is a dressed ``2-cell" which contains one 0-cell (vertex) $p$ at the southpole and one ``2-cell" $S^2$.  Here the quotation mark emphasizes that its not a genuine 2-cell which should be a 2-disk.  However, $S^2$ together with a vertex $p$ (0-cell) nested on it can just be interpreted as a genuine 2-cell $D^2$ attached to the vertex $p$, as illustrated in (f).   Therefore, we do end up with a genuine CW complex.

We now analyze the handle structure of the dressed ``3-handle" $\mathsf{B}_i$, as illustrated in Fig.~\ref{fig:retraction_3-handle}(a) with a more realistic 3D presentation of its dressed core $N_i$  instead of the more abstract presentation in Fig.~\ref{fig:dictionary}(d).  Here, one should think of $S^3$ being the entire 3D space inside a 3-ball $D^3$ with its boundary $S^2$ being identified to a single point.  Recall that the dressed core has the form of a punctured 3-sphere  $N^3_i=S^3 \backslash \sqcup_{m=1}^{f(i)}D^3_m$ with $f(i)$ 3-balls (grey shade) being removed, left with $f(i)$ disconnected $S^2$ boundaries ($f(i)=3$ in the illustration of Fig.~\ref{fig:retraction_3-handle}). The entire ``3-handle" is the thickening of its dressed core, i.e., $\mathsf{B}_i = N_i \times D^{r-3}$.  

To clarify its handle structure, we start with the $f(i)$ disconnected thickened $S^2$ boundaries ($S^2 \times D^{r-3}$) and attach handles with increasing indexes, as illustrated in Fig~\ref{fig:retraction_3-handle}(b). We first add $(f(i)-1)$ 1-handles 
\[ h_1=\bigg(I \times (D^2\times D^{r-3}), S^0 \times (D^2\times D^{r-3})\bigg),\]
where $S^0$ represents two endpoints of the interval $I$. The attaching region is hence two copies of thickened 2-disk $D^2 \times D^{r-3}$ (purple). Now we attach the   $(f(i)-1)$ 1-handles to connect the $f(i)$ disconnected $S^2$ boundaries. Here, each 1-handle is attached to two different thickened $S^2$ boundaries,  which can be regarded as attaching to the boundary of dressed ``2-handles" $\mathsf{C}_j$ connected to the ``3-handle" $\mathsf{B}_i$.   Now within the 3-sphere $S^3$ that contains the dressed core, the $(f(i)-1)$ 1-handles along with the $f(i)$ removed 3-balls $D^3$ (grey) form a single connected component homeomorphic to a 3-ball $D^3$ (highlighted by the pink dashed lines). When we remove this this connected component $D^3$ from the 3-sphere $S^3$, the complement becomes $S^3\backslash D^3 =D^3$, i.e., another 3-disk.  This complement, when thickened by $D^{r-3}$, just forms a three handle $h_3 $$=$$ (D^3 \times D^{r-3}, S^2 \times D^{r-3})$, with the attaching region being the thickened 2-sphere $S^2 \times D^{r-3}$ (the region highlighted by purple dashed lines).  We hence obtain the complete handle structure of the dressed ``3-handle" $\B_i$: one 3-handle, $(f(i)-1)$ 1-handles.

\begin{figure*}[hbt]
    \centering
    \includegraphics[width=2\columnwidth]{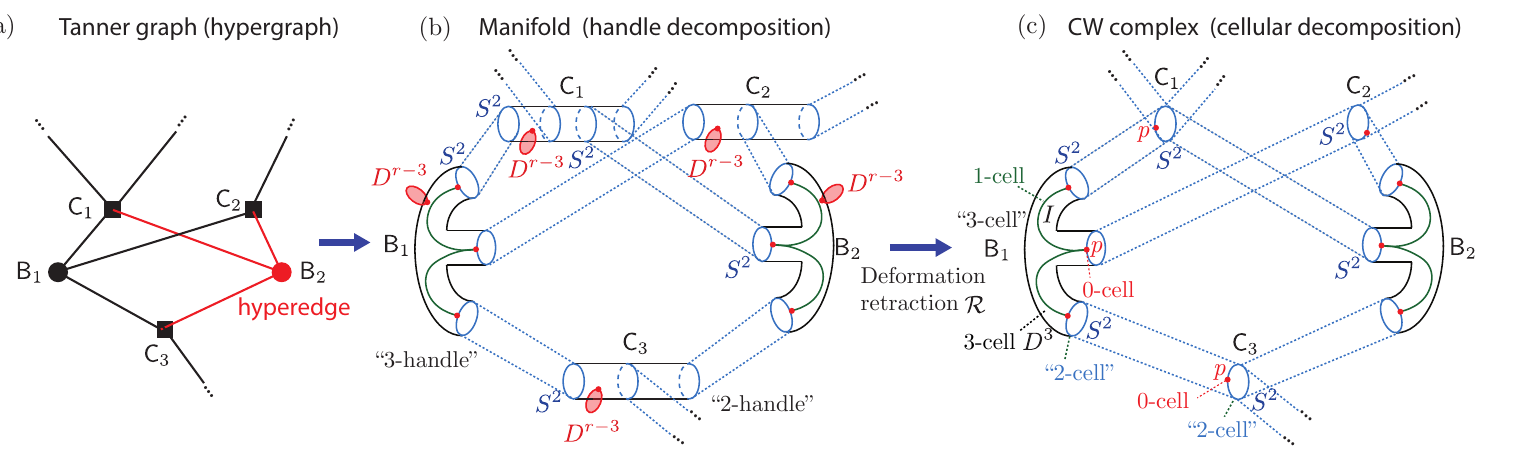}
    \caption{(a,b) Mapping the input Tanner graph to the handle decomposition for the 3-handlebody $H$ of a $r$-manifold. The cells are showed on the dressed core of the ``handles". (c) Deformation retraction to a CW complex. The ``3-handle" are retracted to its dressed core, a ``3-cell" containing one 3-cell, $(f(i)-1)$ 1-cells, $f(i)$ $S^2$ boundaries each with a 0 cell $p$.  Each ``2-cell" contains one $S^2$ and one 0-cell $p$, both identified with the $S^2$ boundaries and 0-cell $p$ on the neighboring ``3-cells".  We see that a single ``2-cell" can be attached to more than two 3-cell, which shows the CW complex is not the discretization of a manifold.}
    \label{fig:dictionary_retraction}
\end{figure*}

We now analyze the deformation retraction of the dressed ``3-handle".  First, the entire ``3-handle" is retracted to its dressed core $N^3_i=S^3 \backslash \sqcup_{m=1}^{f(i)}D^3_m$  according to Eq.~\eqref{eq:retraction_dressed_handle}, which is a dressed ``3-cell", as illustrated in Fig.~\ref{fig:retraction_3-handle}(c).  During this process, each 1-handle $h_1=I \times D^{r-1}$ is retracted to its core $I$ (green), which now becomes a single 1-cell (edge). Meanwhile, the 3-handle $h_3=D^3 \times D^{r-3}$ is retracted to a single 3-cell $D^3$, which is the complement of the $(f(i)-1)$ 1-cells and the $f$ 3-balls $D^3$ inside the 3-sphere $S^3$.    
Now we also investigate the interaction between a ``3-handle"  and its neighboring ``2-handles".
Each 1-handle should be attached to the boundary of the 0-handles in the neighboring ``2-handles". Recall that each ``2-handle" is composed of one 2-handle and one 0-handle, and the latter is exactly the 1-handle is attached to. Now under deformation retraction of the entire manifold, the 0-handle is retracted to a 0-cell $p$ (red vertex).  Therefore, each 1-cell $I$ in the dressed ``3-cell" is now attached to two 0-cells (red vertices) in the neighboring dressed ``2-cells".   Finally, the dressed ``2-handle" has three disconnected boundary components, i.e., the thickened 2-sphere $S^2 \times D^{r-3}$, which are attached to the single 2-handle $D^2 \times D^{r-3}$ in the neighboring dressed ``2-handle". Under deformation retraction, the 2-handle is retracted to the ``2-cell" $S^2$ which can be further decomposed as a 2-cell $D^2$ attached to a 0-cell $p$, as has been shown in Fig.~\ref{fig:anatomy}(f).  So the cell structure of the dressed ``3-cells" in the CW complex $\L_c$ is composed of one 3-cell $D^3$, $(f(i)-1)$ 1-cells $D^1=I$ and in addition $f(i)$ 0-cells $p$ and 2-cells $D^2$ on its boundary, which will be identified with the corresponding cells on the neighboring dressed ``2-cells".

From the above analysis, we can conclude the following fact:
\begin{fact}\label{fact:only_one_cell}
The dressed ``2-cells" and ``3-cells" only contain exactly one bare 2-cells and 3-cells respectively.  
\end{fact}

\subsection{Global structure of the CW complex and Poincar\'e duality}

Equipped with the handle anatomy and retraction procedure provided in the previous subsection, we now take a look at the global structure of the CW complex, as illustrated in Fig.~\ref{fig:dictionary_retraction}, where we use the same Tanner graph example as in Fig.~\ref{fig:dictionary}.  For any skeleton classical code $\bar{\C}=\text{Ker}(\bar{\Hs})$ and its associated Tanner graph $G_T$ [see Fig.~\ref{fig:dictionary_retraction}(a)], we can build a manifold $\M^r$ with the handle construction.  We illustrate the 3-handlebody $H$ of the manifold in Fig.~\ref{fig:dictionary_retraction}(b), whose double gives the entire manifold, i.e., $\M= \D H$. Now applying the deformation retraction $\R$ gives rise to the right portion of the following cellular (CW) chain complex $\L_c$ (set $r=8$) from $C_0$ to $C_3$:
\begin{align}\label{eq:long_chain_CW}
& C_8 \rightarrow C_7 \rightarrow C_{6} \xrightarrow[]{\hat{\partial}^T = \bar{\mathsf{H}}^T}  C_5 \rightarrow 0  \rightarrow C_3 \xrightarrow[]{\hat{\partial} = \bar{\mathsf{H}}} C_{2} \rightarrow C_1 \rightarrow C_0, \cr
&\qquad \qquad \qquad  \qquad \qquad \qquad \qquad \ \   \text{check}  \qquad \text{bit} 
\end{align}
which is completely isomorphic to the handle chain complex $\L_h$.  The corresponding CW complex from 0-cells up to 3 cells is shown in Fig.~\ref{fig:dictionary_retraction}(c). The higher cells corresponding to the left portion of the cellular (CW) chain complex in Eq.~\eqref{eq:long_chain_CW} are not shown in (c).      
The detailed cell structure of each ``2-cell" and ``3-cell" is shown explicitly.    Note that under deformation retraction,  each ``2-handle" is shrunk to a single ``2-cell" $S^2$ along with a 0-cell $p$, which is then attached to in general more than two ``3-cells" along their $S^2$ boundaries.  This leads to the following fact:
\begin{fact}
The CW complex $\L_h$ obtained from deformation retraction,  unlike the triangulation $\L$ of the manifold $\M^r$,  is in general no longer a discretization of the manifold $\M^r$ or equivalently a combinatorial manifold.   
\end{fact}


We now consider defining a classical code $\C$ on the CW complex $\L_c$.  Note that the dressed ``2-cells" and ``3-cells" have one to one correspondence to the checks $\mathsf{C}_j$ and bits $\mathsf{B}_i$ due to the sequence of mappings from the skeleton classical code $\bar{\C}=\text{Ker}(\bar{\Hs})$ to the manifold and then to the CW complex.  Now due to Fact \ref{fact:only_one_cell}, each ``2-cell" and ``3-cell" only contain one genuine 2-cell $D^2$ and 3-cell $D^3$ respectively.  We now place a bit on each 3-cell and a check on each 2-cell of the CW complex $\L_c$, and define the classical code as 
\be
\C=H_3(\L_c; \ZZ_2)=\text{Ker}(\partial_3)=\text{Ker}(\bar{\Hs}). 
\ee
Due to the isomorphism between a portion of the CW complex (level 2 and 3) to the chain complex $\mathcal{X}$ of the skeleton classical code,  this new code is exactly  the same as the input skelton code, i.e., $\C=\bar{\C}=\text{Ker}(\bar{\Hs})$, same as their transposed code $\C^T=\bar{\C}^T =\text{Ker}(\bar{\Hs}^T)$. 

Note that the above equivalence also applies to the Freedman-Hastings mapping from quantum code to the 11D manifold in Ref.~\cite{freedman:2020_manifold_from_code}. 
As mentioned before, the $X$-check, qubit, and $Z$-check are mapped to dressed ``3-handles", ``4-handles" and ``5-handles" respectively, and each of them only contain exactly one bare 3-handle, 4-handle and 5-handle respectively.  When deformation retracted to an 11D CW complex $\L_c$, the dressed ``handles" are retracted to dressed ``3-cells", ``4-cells" and ``5-cells" respectively, each of which contain exactly one bare 3-cell, 4-cell and 5-cell respectively. We define a new quantum code $\C= \mathbb{C}^{|H_4(\L_c; \ZZ_2)|}$ by placing one $X$-check, qubit and $Z$-check on each 3-cell, 4-cell and 5-cell respectively.    Since a portion of the CW complex $\L_c$ (from level 3 to 5) is isomorphic to the chain complex $\mathcal{X}$ of the skeleton quantum code $\bar{\C}$,  the new quantum code is exactly the same as the skeleton quantum code, i.e., $\C=\bar{\C}$.   We hence reach the following claim:  
\begin{claim}
For a given skelton classical or sparsely liftable quantum LDPC code $\bar{\C}$, there exists a bounded degree CW complex $\L_c$ which admits cup product, and an LDPC code $\C$ defined on it such that $\C=\bar{\C}$. 
\end{claim}
\nin Here, bounded degree means each cell is only adjacent to $O(1)$ cells, similar to the condition for the triangulation $\L$, which is gauranteed by the LDPC condition of the input code and sparse liftability. See also Ref.~\cite{guemard2025lifting} for a similar discussion. The discussion of the cup product will be presented in Sec.~\ref{sec:CW_cup}.

The above claim is crucial for practical implementation, since it tells us we can build a classical or quantum code with a \textit{hidden CW complex structure} using exactly the same number of (qu)bits and checks as the input code.
Note that although the CW complex $\L_c$ built from the skeleton classical code  is $r$-dimensional, the part the store the classical information, i.e., bits and checks, is only the portion of the chain complex containing 2-cells and 3-cells.   Similarly, for the CW complex built from the skeleton quantum code, the quantum information are are only stored in the 3-term portion of the chain complex containing 3-cells, 4-cells and 5-cells respectively, which are those need to be implemented in the device.  
The hidden CW complex structure including levels aside from (qu)bits and checks, which could be stored in a classical computer, is used to design the logical gates via cohomology operation but does not encode any message.  

The only additional subtlety to bear in mind is that for our construction of thickened 3D hypergraph or homological product codes, we also need to take product of the CW complexes $\L_c$. 
Therefore, cells not at the dimensions of the (qu)bits and checks may also participate in the composition of the cells in the product codes.  However, each dressed ``cells" only contain a small amount of lower dimensional cells.   For example, in the classical code case, each ``2-cell" only contains one 0-cell, and each ``3-cell" only contains $f(i)$ 1-cell, where $f(i)$ is the number of checks that the bit $\mathsf{B}_i$ is adjacent to.   Similar for the higher cells with the dual dimensions.   Therefore, the product code we construct with the CW complexes is also very compact with a very small constant overhead compared to the  product of the skeleton codes.   

In addition, we note that this scheme is applicable to arbitrary skeleton classical or quantum code without requiring any specific local combinatorial structure like Ref.~\cite{golowich2024quantum, lin2024transversal, breuckmann2024cups}, and should hence be achievable with near-term codes with O(100) to O(1000) qubits, such as the homological product of a bivariate bicycle code \cite{Bravyi:2024wc} and a classical expander code, or the tricycle codes  based on balanced product construction \cite{jacob2025single, menon2025magic}.

Finally and more interestingly,  due to the isomorphism between the CW and handle complex $\L_c \cong \L_h$,  the Poincar\'e duality isomorphism $H_k(\L_h; \ZZ_2) $$\cong$$ H^{r-k}(\L_h^*; \ZZ_2)$ in the handle-complex of the manifold is still preserved in the CW complex: $H_k(\L_c; \ZZ_2) $$\cong$$ H^{r-k}(\L_c^*; \ZZ_2)$, where $\L_h^*$ denotes the dual handle chain complex, e.g.,  Eq.~\eqref{eq:long_chain_dual}, and $\L_c^*$ denotes the corresponding dual CW complex obtained from the deformation retraction.   Such a CW complex equipped with Poincar\'e duality is called a \textit{Poincar\'e complex} in the literature:
\begin{definition}
An $r$-dimensional CW complex $\L_c$ is  called a Poincar\'e complex if there exists the following isomorphism:
 	\be
 	H_k(\L_c) \cong H^{r-k}(\L_c^*).
 	\ee
\end{definition}
\nin Note that in general the above definition applies to general coefficients in stead of just $\ZZ_2$ coefficients, which is the focus of the present study.


\subsection{Cup product on the CW complex}\label{sec:CW_cup}

\begin{figure}
    \centering
    \includegraphics[width=1\linewidth]{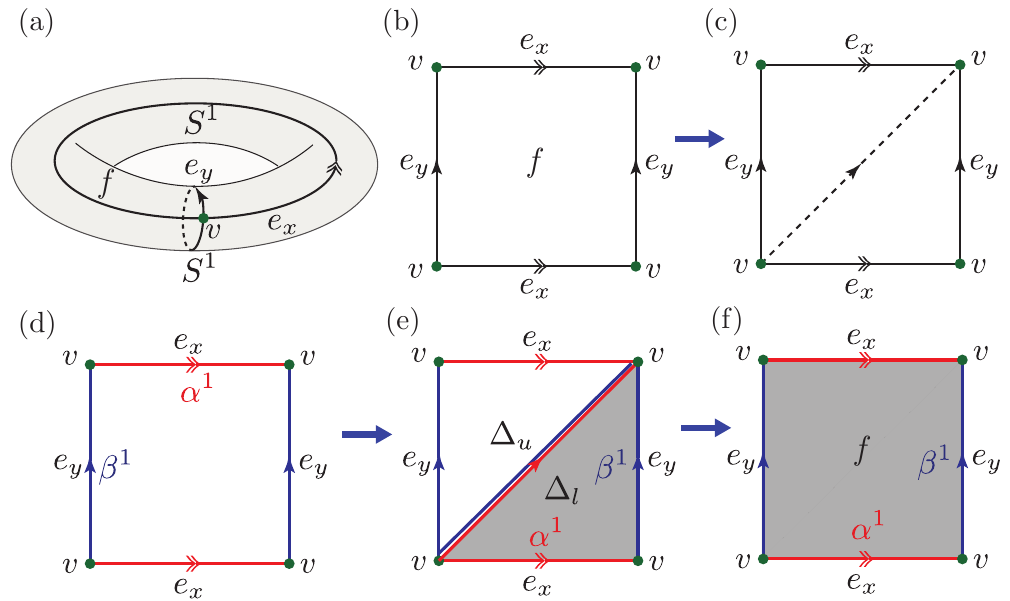}
    \caption{(a) The CW complex of a torus is composed of one 0-cell (vertex) $v$, two 1-cells (edges) $e_x$ and $e_y$ supported on two intersecting circles $S^1$, and one 2-cell (face) $f$. (b) A representation of the torus by identifying the opposite edges and all the corners.  (c) One can subdivide it into a simplicial complex with two triangles.  (d) The configuration of two cocycles $\alpha^1$ (red) and $\beta^1$ (blue), with the highlighted edges taking value 1 for the  cocycle with the corresponding color. (e) After the subdivision, the cocycles are extended to diagonal edges. One then evaluates the cup product on the subdivided complex. Only the lower-right triangle (grey) evaluates to 1 in this example.  (f) One adds the contribution from the two triangles and get the cup product evaluated for the square face $f$, which is 1 in this example (highlighted grey).  }
    \label{fig:CW_cup}
\end{figure}

Similar to the case of simplicial complex, cohomology operations including cup products are also mathematically well-defined on the CW complex \cite{Hatcher:2001ut}.  On the other hand, there is no convenient direct formula for the cup product evaluation on the CW complex like Eq.~\eqref{eq:cup_def} in the case of simplicial complex.
Therefore, one solution is to first subdivide the CW complex $\L_c$ into a simplicial complex $\L_c^\Delta$, and then use Eq.~\eqref{eq:cup_def} to evaluate the cup product. For example, each $k$-cell which is a $k$-ball $D^k$ can be represented by a single $k$-dimensional simplex, while the ``$k$-cell" corresponding to a $k$-sphere $S^k$ can be represented by the $(k-1)$-simplices on the boundary of a single $k$-simplex.  Along this line, we can use the following three aproaches:
\vspace{0.15 in}

\nin (1) One conceptually simple option is to directly define the classical or quantum code on the subdivided simplicial complex $\L_c^\Delta$. This will leads to a small constant overhead since the subdivided simplicial complex  $\L_c^\Delta$ is no longer the triangulation of the manifold $\M^r$, and the new code is  more compact than the homology code defined on the triangulation of $\M^r$. The constant overhead proportional to number of simplices in the subdivided cells corresponding to the (qu)bits and checks. Note that the new code is not a topological code even in the general sense (allowing non-Euclidean geometry).  Nevertheless, this simplicial complex $\L_c^\Delta$ is a \textit{Poincar\'e complex} which inherits the Poincar\'e duality of the manifold.  

In this approach, the dimension $r$ of the underlying manifold $\M^r$ and the CW complex, as well as the dimension $k$ of the (qu)bits,  do mildly affect the qubit overhead and stabilizer weight (connectivity).   For example, in the classical code the bits located on a $k$-cell $D^k$ which can be triangulated into a $k$-simplex, and it will be adjacent to  $(k+1)$ number of $(k-1)$-simplex on the faces of the $k$-simplex, which means each bit will participate in $k+1$ checks.  The additional overhead on the number of checks is also proportional to $k+1$. Now if we put bits on 3-cells as the focus of this section, or on 2-cells in the 4-manifold construction in Sec.~\ref{sec:4-manifold}, the number of checks that each bit participate is 4 and 3 respectively. Similarly, the overhead on the number of checks will be proportional to a factor of 4 and 3 respectively. Note that even in the skeleton code, a single bit will be typically coupled to 2 or more checks, and is likely to have a similar or even large number compared to 3 or 4. Therefore, the connectivity is not necessarily worsened and could perhaps even be improved when doing the subdivision. The only disadvantage is the additional space overhead. The total dimension $r$ also matters when taking the product of the simplicial complex  $\L_c^\Delta$ built from the classical or quantum codes.   

\vspace{0.15 in}
\nin (2) An even more economic approach is to keep the codes $\C$ defined exactly on the CW complex $\L_c$, while using the subdivided simplicial complex $\L_c^\Delta$ to derive the evaluation rule on the CW complex $\L_c$.  It is then possible not to introduce extra (qu)bits and checks, and hence reach the minimal overhead, which is the same as the skeleton classical/quantum codes $\bar{\C}$ (before taking the product).   Note that in this appraoch, the dimension of the CW complex and its underlying manifold does not really affect the practical implementation.  

We use a toy example to instruct the general recipe. We consider the evaluation of the cup product on a CW complex of a torus $S^1 \times S^1$, as shown in Fig.~\ref{fig:CW_cup}(a, b).  The CW complex contains one vertex $v$ and two edges $e_x$ and $e_y$ corresponding to the two $S^1$.   Note that in the illustration in (b), the opposite edges are identified, and so do the vertices at the four corners.  One could still imagine that these four identified vertices have a particular order, which gives rise to the so-called branching structure with arrows on the edges pointing from lower order to higher order.  Now in order to apply the simplicial cup product formula in Eq.~\eqref{eq:cup_def},  we can subdivide the CW complex into a simplicial complex by adding a diagonal edge, as shown in (c). 

Since we are always dealing with logical states in the code space, we will always consider the cup product of cocycles. For example, two different cocycles $\alpha^1$ (red) and  $\beta^1$ (blue) on the CW complex are shown in Fig.~\ref{fig:CW_cup}(d), while highlighted edge means it takes value 1 on the corresponding cocycle (0 otherwise). 
Now when subdividing the CW complex into a simplicial complex, we need to assign the cochain value on the newly added edge, which should keep the cocycle class $[\alpha]$ or $[\beta]$ the same before and after the subdivision.  In this case, both the $\alpha$ and $\beta$ cocycle should take value 1 on the diagonal edge.  Now we can apply the simplcial cup product formula in Eq.~\eqref{eq:cup_def}:
\begin{align}
(\alpha^1 \cup \beta^1)(\Delta_l)= \alpha^1(e_x) \beta^1(e_y), \cr
(\alpha^1 \cup \beta^1)(\Delta_u)= \alpha^1(e_y) \beta^1(e_x),
\end{align}
where $\Delta_l$ and $\Delta_u$ represents the lower and upper triangle respectively.   For this particular choice of $\alpha^1$ and $\beta^1$, we have the
first and second terms equaling to 1 and 0 respectively, where the lower triangle with value 1 is highlighted in (e).   It turns out in this special example, we are not using the cochain on the added edge, but in general it may be used.   Now we add the contribution of the two triangles together, such that we can derive the following formula for the CW complex without diagonal edge:
\be
(\alpha^1 \cup \beta^1) (f)= \alpha^1(e_x) \beta^1(e_y) + \alpha^1(e_y) \beta^1(e_x). 
\ee
With the current choice of $\alpha^1$ and $\beta^1$,  we will have $(\alpha^1 \cup \beta^1) (f)=1$, as illustrated in (f).   

The above recipe can be applied to general CW complex,  for example $S^q \times S^q$.  The formula for general cells may vary case by case.    We will leave the specific details of deriving  the cup-product formula for different types of cells in our CW complex construction for qLDPC codes to an upcoming work.   

\vspace{0.15 in}

\nin (3) There is yet another approach which sits in the middle of the previous two approaches. We can still consider the quantum code $\C$ to be defined on the product CW complex $\L_c$.  The subdivision of the CW complex $\L_c$ into a simplicial complex $\L_c^\Delta$ can be physically implemented by a constant-depth circuit $V$ with a constant overhead of additional ancilla qubits initialized in $\ket{0}$ or $\ket{+}$ and then entangled with the code qubits by $V$, which has been demonstrated in Appendix D of Ref.~\cite{Lavasani2019universal} (see also Ref.~\cite{Zhu:2018CodeLong}).  Through this constant-depth circuit and ancillas, we can transform the code $\C$ to a new code $\C'$ defined on the simplicial complex $\L_c^\Delta$, where one can easily use the standard cup product evaluation to implement the constant-depth circuit $U$ corresponding to logical CCZ gate.  
Now the composition of these two types of  unitaries 
$U'=V^\dag U V$ is still a constant-depth circuit which implements a logical CCZ gate directly acted on the original code $\C$ defined on the CW complex $\L_c$ along with ancillas.  Note that after applying the final $V^\dag$ in the composed circuit, the ancillas are again disentangled from the code qubits.  The point is that we never need to do stabilizer measurements or error correction on the code $\C'$ defined on the subdivided simplicial complex $\L_c^\Delta$, since its existence is transient in the middle step of the composed circuit $U'=V^\dag U V$. One only needs to do the stabilizer measurements and error correction on the original code $\C$ defined on the CW complex, which is more compact.

\section{Discussion and outlook}
 Although we have picked good classical and qLDPC codes as input of our construction to reach desirable asymptotic scaling, we emphasize that the construction we have introduced in this paper is quite general: one can use the Tanner graphs of any classical or quantum codes as input for the manifold and product constructions, and the formalism will give rise to the collective logical CCZ gates. The input codes can also be general CSS codes instead of qLDPC codes.   The only difference in this case is that the constructed manifold may not have bounded local geometry.  In particular, this method will also be suitable for the near-term realization of small- and intermediate-scale codes with $O(100)$ to $O(1000)$ qubits such as the homological product of a bivariate bicycle code \cite{Bravyi:2024wc} and a classical expander code, or the tricycle codes based on balanced product construction \cite{jacob2025single, menon2025magic}.  Various computational topology packages, such as the CGAL library \cite{cgal:eb-24b}, can be used to numerically construct the triangulation of manifolds or the CW complexes.

In this paper, we have been focusing on the homological product construction \cite{Bravyi:2014bq}. Nevertheless, the construction also straightforwardly applies to the more general balanced product construction \cite{Breuckmann:2021_balanced}. More concretely, one can map the classical or quantum code in the balanced product construction into a manifold first, and then take a balanced product of the produced manifolds, which is essentially a fibre-bundle construction of manifolds \cite{Freedman_systole_2002, fiberbundlecode21}.   This may lead to constructions with even better distance parameters which can now go beyond the $\sqrt{N}$-distance barrier and eventually construct non-Clifford logical gates in an asymptotically good qLDPC code.  We leave this for future exploration. 

Another future direction is to optimize the parallelizability in the magic state fountain scheme by introducing a more separable intersection structure, and combine it with highly parallelizable logical Clifford gates \cite{xu2024fast} to achieve a fault-tolerant computing scheme with very low space-time overhead.  Even more interestingly, the implementation of such logical non-Clifford gates on quantum locally testable codes may also shed light on the study of quantum PCP conjecture, where the connection between polylog PCP and fault-tolerance has been recently pointed out in Ref.~\cite{anshu2024circuit}.

In a related work \cite{Hsin2024:classifying}, we have systematically classified various types of cohomology operations that can be used to implement logical gates via constant-depth circuits beyond $k$-fold cup products (the so-called color code paradigm), including Steenrod squares and new combinations of higher cup products called higher Pontryagin powers.  These new cohomology operations can lead to more exotic logical gates other than the $C^{n-1}Z$ family such as fine single-qubit rotation including $T$-gate and $R_k$-gate, as well as controlled rotation $CR_k$-gates etc.  It would be very interesting to further incorporate these new cohomology operations with the high-rate and large-distance qLDPC codes, which for example can provide fast $T$-state injection to the magic state fountain.   A systematic classification of logical gates in qLDPC codes via cohomology operations and emergent symmetries will also be an interesting potential direction.

\noindent{\it Acknowledgements} --- We thank Elia Portnoy and Michael Freedman for insightful discussion and early collaboration on this project, including the contribution of various original ideas. We appreciate the discussion with Virgile Guemard for pointing out the connection between the manifold and CW complex through deformation retraction as also stated in Ref.~\cite{guemard2025lifting}. We thank Matthew Hastings for clarifying the sparse-liftability of the balanced product construction. We thank Louis Golowich for the discussion on the related topics. We also thank Andrew Cross, Shehryar Sikander, Ben Brown, Po-Shen Hsin, Ryohei Kobayashi, and Maissam Barkeshli for the previous collaboration on related projects. G.Z. is supported by the U.S. Department of Energy, Office of Science, National Quantum Information Science Research Centers, Co-design Center for Quantum Advantage (C2QA) under contract number DE-SC0012704.

\begin{appendix}\label{app:sparse-liftability}
\section{Sparse liftability}
In order for the constructed manifold to have a bounded local geometry corresponding to the LDPC condition for the homological code built on its triangulation $\L$,  we need to make sure the $\ZZ$-lift in Def.~\ref{def:lift} is sparse:
\begin{definition}\cite{freedman:2020_manifold_from_code}
    We say a $\Z$-lift is sparse, if the sum of absolute values of the $i^\text{th}$ column of lifted boundary map (parity check matrix) $\tilde{\partial}$ dnoted by $\tilde{f}(i)$ satisfies the condition $\tilde{f}(i)= O(1)$.
\end{definition}

In the following, we analyze the sparse liftability of a hypergraph product code and balanced product code respectively.  
\subsection{Lifting hypergraph product codes}

Note that in the construction of thickened hypergraph product code, we directly lift each classical code with a naive lift to build a manifold, and further take the product of manifolds.   Here, we consider directly lift the quantum code and construct the manifold following the Freedman-Hastings mapping.

We start with two $\ZZ_2$ 1-complexes $\cA$ and $\cB$ corresponding to two classical codes, and then construct the product $\ZZ_2$ 2-complex $\cC=\cA \otimes \cB$.   The chain group of the product complex can be expressed as
$\cC_r =\bigoplus_{r=p+q}\cA_p \otimes \cB_q$, with $p,q=0,1$.  We can hence compose the $r^\text{th}$ boundary map in the product complex $\partial_r^\C$ as:
\begin{equation}\label{eq:Z2-boundary}
    \partial^\C_{(p,q)}= \partial^\cA_p  \otimes I + I \otimes \partial^\cB_q, 
\end{equation}
and 
\begin{equation}
    \partial^\C_r= \bigoplus_{p+q=r} \partial^\C_{(p,q)}, 
\end{equation}
where $(p,q)$ indicates the $(p+q)$-cell in $\C$ which is a product of the $p$-cell in $\cA$ and the $q$-cell in $\cB$. Note that here we are discussing a general situation, while for the hypergraph product of classical codes the only non-trivial boundary maps in $\cA$ and $\cB$ are just $\partial^\cA_1$ and $\partial^\cB_1$.

We hence have the exactness condition (mod 2):
\begin{align}
\partial_1 \circ \partial_2 =& (\partial^\cA_1  \otimes I + I \otimes \partial^\cB_1) \circ (\partial^\cA_1  \otimes I + I \otimes \partial^\cB_1)  \cr
=& \partial^\cA_1 \otimes \partial^\cB_1 + \partial^\cA_1 \otimes \partial^\cB_1=0 \ \text{mod} \ 2.
\end{align}

We now prove the following lemma:
\begin{lemma}
Any given hypergrpah product code constructed as a tensor product of two classical codes with $\Z_2$ or $\Z_4$ variables corresponding to the product complex $\cC=\cA \otimes \cB$ and boundary maps $\partial_r$ admits a sparse $\Z$ lift to a lifted complex $\tilde{\cC}=\tilde{\cA} \otimes \tilde{\cB}$ with  lifted boundary maps $\tilde{\partial}_r$. 
\end{lemma}
\begin{proof}
We first lift the 1-complexes $
\cA$ and $\cB$ over $\ZZ_2$  to 1-complexes $
\tilde{\cA}$ and $\tilde{\cB}$ over $\Z$ respectively.  In particular, we use a naive lift that $(0 \ \text{mod} \ 2) \rightarrow 0$ and   $(1 \ \text{mod} \ 2) \rightarrow 1$.  If we start with $\ZZ_4$-chain complexes instead, we will do the following naive lift: $(0 \ \text{mod} \ 4) \rightarrow 0$,    $(1 \ \text{mod} \ 4 ) \rightarrow 1$,  $(2 \ \text{mod} \ 4) \rightarrow 2$ and  $(3 \ \text{mod} \ 4) \rightarrow 3$.   We hence get a sparse lift of the boundary maps (parity check matrices) in the 1-complexes from $\ZZ_2$ ($\ZZ_4$) to $\Z$:
\begin{equation}
    \partial_p^\cA \rightarrow \tilde{\partial}_p^\cA, \quad   \partial_p^\cB \rightarrow \tilde{\partial}_p^\cB.
\end{equation}

Now we can define the $\ZZ$-lifted boundary map in the product complex $\tilde{\cC}$$=$$\tilde{\cA} \otimes \tilde{\cB}$ as
\begin{equation}\label{eq:lifted_boundary}
    \tilde{\partial}^\C_{(p,q)}= \tilde{\partial}^\cA_p  \otimes I + (-1)^p  I \otimes \tilde{\partial}^\cB_q. 
\end{equation}
Note that we have an extra sign $(-1)^p$ compared to the $\ZZ_2$ case in Eq.~\eqref{eq:Z2-boundary}.
We hence obtain the following exactness condition: 
\begin{align}\label{eq:lifted_exactness}
\tilde{\partial}_1 \circ \tilde{\partial}_2 =& (\tilde{\partial}^\cA_1  \otimes I + I \otimes \tilde{\partial}^\cB_1) \circ (\tilde{\partial}^\cA_1  \otimes I - I \otimes \tilde{\partial}^\cB_1)  \cr
=& \tilde{\partial}^\cA_1 \otimes \tilde{\partial}^\cB_1 - \tilde{\partial}^\cA_1 \otimes \tilde{\partial}^\cB_1=0.
\end{align}
Note that $\tilde{\partial}_1 \circ \tilde{\partial}_2 =0$ holds exactly instead of mod 2, so the boundary map  definition  for the product complex in Eq.~\eqref{eq:lifted_boundary} is valid for $\Z$ or $\Z_N$ coefficients.   

We see that such a $\ZZ$-lift in Eq.~\eqref{eq:lifted_boundary}  either from $\Z_2$ or $\Z_4$ preserves the sparsity of the lifted  boundary map (parity-check matrix) $\tilde{\partial}^\C_r$, since the components $\tilde{\partial}^\cA_p$ and $\tilde{\partial}^\cB_q$ are sparse.

\end{proof}
We note that the product of naive lifts, as used in the above proof, is different from the naive lift of the product,  which in general cannot  satisfy the exactness condition in Eq.~\eqref{eq:lifted_exactness}.  

\subsection{Lifting balanced product codes}

We now generalize the sparse liftability of the hypergraph product code to the more general cases of the balanced product codes, which can also be viewed as a fiber bundle.

\begin{lemma}
Any given balanced product code constructed from two classical codes with $\Z_2$ or $\Z_4$ variables (the associated  1-complexes being $\cA$ and $\cB$) corresponding to the 2-complex $\cC=\cA \otimes_G \cB$ admits a sparse $\Z$ lift ($G$ is a ring).
\end{lemma}

\begin{proof}
We start with the two 1-complexes $\cA$ and $\cB$ associated with the classical codes with the associated boundary map being $\partial^\cA \equiv \partial^\cA_1$ and $\partial^\cB= \partial^\cB_1$.  We first do a naive lift from $\ZZ_2$ ($\ZZ_4$) to $\Z$: 
\begin{equation}
\partial^\cA \rightarrow \tilde{\partial}^\cA, \quad   \partial^\cB \rightarrow \tilde{\partial}^\cB,
\end{equation}
which is a sparse lift.   This gives rise to the lifted chain complex $\tilde{\cA}$ and $\tilde{\cB}$.

We then define the $\Z$-lifted boundary map in the $\Z$-lifted balanced product 2-complex $\tilde{\cC}=\tilde{\cA} \otimes_G \tilde{\cB}$ as:
\begin{equation}\label{eq:lifted_boundary_balanced}
    \tilde{\partial}^\C (x \otimes y)= \tilde{\partial}^\cA x  \otimes y + (-1)^p  x \otimes \tilde{\partial}^\cB y, 
\end{equation}
where $x \in \tilde{\cA}_p$ and $y \in \tilde{\cB}_q$ are $p$-chain and $q$-chain respectively ($p,q=0,1$).  For the balanced product we have the equivalence relation between right and left action of $g \in G$:  
\begin{equation}
(x \cdot g)\otimes y =x \otimes (g\cdot y).
\end{equation}
Note that the lifted boundary map respect the above equivalence relation:
\begin{align}
\tilde{\partial}\big( (x \cdot g) \otimes y \big)
&= \tilde{\partial}^{\cA}(x \cdot g) \otimes y
   + (-1)^p \, (x \cdot g) \otimes \tilde{\partial}^{\cB} y \cr
&= (\tilde{\partial}^{\cA} x) \cdot g \otimes y
   + (-1)^p \, x \cdot g \otimes \tilde{\partial}^{\cB} y \cr
&= \tilde{\partial}^{\cA} x \otimes (g \cdot y)
   + (-1)^p \, x \otimes \big( g \cdot \tilde{\partial}^{\cB} y \big) \cr
&= \tilde{\partial} \big( x \otimes (g \cdot y) \big).
\end{align}

We then check the exactness condition:
\begin{align}
\tilde{\partial}^2(x \otimes y)
&= \tilde{\partial}\big( \tilde{\partial}^{\cA} x \otimes y \;+\; (-1)^p\, x \otimes \tilde{\partial}^{\cB} y \big) \cr
&= (\tilde{\partial}^{\cA})^2 x \otimes y 
\;+\; (-1)^{p-1} \, \tilde{\partial}^{\cA} x \otimes \tilde{\partial}^{\cB} y \cr
&\quad\;+\; (-1)^p \, \tilde{\partial}^{\cA} x \otimes \tilde{\partial}^{\cB} y
\;+\; (-1)^{2p} \, x \otimes (\tilde{\partial}^{\cB})^2 y \cr
&= 0 
\;+\; \big[ (-1)^{p-1} + (-1)^p \big] \tilde{\partial}^{\cA} x \otimes \tilde{\partial}^{\cB} y
\;+\; 0 \cr
&= 0.
\end{align}

We hence have successfully lifted the boundary map and chain complex.  The lift is sparse since the lifted boundary map $\tilde{\partial}^\C$ (parity check matrix) in Eq.~\eqref{eq:lifted_boundary_balanced} is a sum of two sparse matrices.

\end{proof}
\end{appendix}

\bibliography{mybib_merge.bib, TI.bib}

\end{document}